\pdfoutput=1
\documentclass[11pt,letterpaper]{article}
\usepackage[margin=1in]{geometry}
\def\arXiv{1} 
\newcommand{\notarxiv}[1]{foo}
\newcommand{\arxiv}[1]{ba}
\ifdefined \arXiv
	\renewcommand{\arxiv}[1]{#1}%
	\renewcommand{\notarxiv}[1]{\ignorespaces}%
\else%
	\renewcommand{\arxiv}[1]{\ignorespaces}%
	\renewcommand{\notarxiv}[1]{#1}%
\fi%

\usepackage{thm-restate}
\usepackage{bm}
\usepackage{bbm}

\arxiv{
\usepackage{tikz}
\usetikzlibrary{positioning,chains,fit,shapes,calc}
}
\usepackage{thmtools} 
\arxiv{\usepackage{thm-restate}}
\arxiv{\usepackage{lipsum}}

\usepackage{subcaption}
\usepackage[linesnumbered,ruled,vlined]{algorithm2e}

\SetKwInput{KwInput}{Input}                %
\SetKwInput{KwOutput}{Output}              %
\SetKwInput{KwReturn}{Return}    
\SetKwInOut{Parameter}{Parameters}
\SetKwComment{Comment}{$\triangleright$\ }{}
\SetCommentSty{mycommfont}
\newcommand{\OuterLoop}{\mathsf{MirrorProx}}

\newcommand{\RegBoxSimp}{\mathsf{RegularizedBS}}
\newcommand{\AltminBS}{\mathsf{AltminBS}}
\newcommand{\GradBS}{\mathsf{GradBS}}
\newcommand{\solve}{\mathsf{Solve}}

\usepackage{enumerate}

\usepackage{amsmath,amssymb, amsthm}    %
\usepackage{verbatim}           %
\usepackage{xspace}             %
\usepackage{graphicx,float}     %
\usepackage{ifthen,calc}        %
\usepackage{textcomp}           %
\usepackage{fancybox}           %
\usepackage{hhline}             %
\usepackage{float}              %

\arxiv{
\usepackage[vflt]{floatflt}     %
\usepackage[small,compact]{titlesec}
\usepackage{setspace}
}

\arxiv{
\usepackage{color}
\definecolor{ForestGreen}{rgb}{0.1333,0.5451,0.1333}
\usepackage{thumbpdf}
\usepackage[letterpaper,
colorlinks,,allcolors=blue,%
bookmarks,bookmarksopen,bookmarksnumbered]
{hyperref}
}
\arxiv{\usepackage{cleveref}}
\notarxiv{
\Crefname{lemma}{Lemma}{Lemmas}
\crefname{thm}{theorem}{theorems}
\Crefname{proposition}{Proposition}{Propositions}
\Crefname{corollary}{Corollary}{Corollaries}
\crefalias{AlgoLine}{line}
}

\newenvironment{psmallmatrix}
  {\left(\begin{smallmatrix}}
  {\end{smallmatrix}\right)}

\arxiv{
\newtheorem{theorem}{Theorem}

\newtheorem{proposition}{Proposition}
\newtheorem{corollary}{Corollary}
\newtheorem{definition}{Definition}

\newtheorem{lemma}{Lemma}
}

\arxiv{}

\usepackage{float}
\floatstyle{plain}\newfloat{myfig}{t}{figs}[section]
\floatname{myfig}{\textsc{Figure}}
\floatstyle{plain}\newfloat{myalg}{H}{algs}[section]
\floatname{myalg}{}
\newcommand{\defeq}{:=}
\newcommand{\norm}[1]{\left\lVert#1\right\rVert}
\newcommand{\inprod}[2]{\left\langle#1, #2\right\rangle}
\newcommand{\eps}{\epsilon}

\newcommand{\argmax}{\textup{argmax}}
\newcommand{\argmin}{\textup{argmin}} 
\newcommand{\R}{\mathbb{R}}

\newcommand{\N}{\mathbb{N}}
\newcommand{\diag}[1]{\textbf{\textup{diag}}\left(#1\right)}
\newcommand{\half}{\frac{1}{2}}

\newcommand{\smin}{\textup{smin}}

\newcommand{\opt}{\textup{OPT}}
\newcommand{\xset}{\mathcal{X}}
\newcommand{\yset}{\mathcal{Y}}
\newcommand{\zset}{\mathcal{Z}}
\newcommand{\ma}{\mathbf{A}}
\newcommand{\ai}{\ma_{i:}}
\newcommand{\aj}{\ma_{:j}}

\newcommand{\jac}{\mathbf{J}}
\newcommand{\inner}[2]{\left<#1,#2\right>}

\newcommand{\innerB}[2]{\Bigg<#1,#2\bigg>}
\newcommand{\tO}{\widetilde{O}}
\newcommand{\nnz}{\textup{nnz}}

\newcommand{\Par}[1]{\left(#1\right)}
\newcommand{\Brack}[1]{\left[#1\right]}
\newcommand{\Brace}[1]{\left\{#1\right\}}

\newcommand{\x}{^{\mathsf{x}}}
\newcommand{\y}{^{\mathsf{y}}}
\newcommand{\md}{\mathbf{D}}
\newcommand{\mzero}{\mathbf{0}}

\newcommand{\veps}{\varepsilon}

\newcommand{\1}{\mathbf{1}}
\newcommand{\0}{\mathbf{0}}
\newcommand{\mb}{\mathbf{B}}
\newcommand{\OTRound}{\mathsf{OTRound}}
\newcommand{\OPT}{\mathsf{OPT}}
\newcommand{\tx}{\tilde{x}}
\newcommand{\hx}{\hat{x}}
\newcommand{\Abs}[1]{\left|#1\right|}
\newcommand{\mm}{\mathbf{M}}

\newcommand{\bz}{\bar{z}}

\newcommand{\bx}{\bar{x}}

\newcommand{\AltMin}{\mathsf{AltMin}}

\newcommand{\by}{\bar{y}}

\newcommand{\del}{\textup{del}}
\newcommand{\mx}{\mathbf{X}}
\newcommand{\mk}{\mathbf{K}}
\newcommand{\mc}{\mathbf{C}}
\newcommand{\td}{\tilde{d}}
\newcommand{\Sinkhorn}{\mathsf{Sinkhorn}}
\newcommand{\MaintainMatch}{\mathsf{DecMatching}}

\newcommand{\Round}{\mathsf{Round}}

\newcommand{\dum}{\mathsf{dum}}
\newcommand{\tot}{\mathsf{tot}}
\newcommand{\Remove}{\mathsf{RemoveOverflow}}

\newcommand{\sink}{\mathsf{sink}}

\newcommand{\fme}{f_{\mu,\eps}}
\newcommand{\gme}{g_{\mu,\eps}}
\newcommand{\rme}{r_{\mu,\eps}}

\newcommand{\MCM}{\textup{MCM}}
\newcommand{\tsolve}{\mathcal{T}}
\newcommand{\tG}{\widetilde{G}}
\newcommand{\tL}{\widetilde{L}}
\newcommand{\tR}{\widetilde{R}}
\newcommand{\tV}{\widetilde{V}}
\newcommand{\tE}{\widetilde{E}}
\newcommand{\tmb}{\widetilde{\mb}}

\newcommand{\Bmax}{B_{\max}}
\newcommand{\Cmax}{C_{\max}}

\newcommand{\vmin}{v_{\min}}

\renewcommand{\varepsilon}{\sigma}

\newenvironment{proof-sketch}{%
	\proof}{\endproof}

\newcommand{\citep}{\cite}
\newcommand{\Mest}{M_{\textup{est}}}

\makeatletter
\let\cref@old@stepcounter\stepcounter
\def\stepcounter#1{%
  \cref@old@stepcounter{#1}%
  \cref@constructprefix{#1}{\cref@result}%
  \@ifundefined{cref@#1@alias}%
    {\def\@tempa{#1}}%
    {\def\@tempa{\csname cref@#1@alias\endcsname}}%
  \protected@edef\cref@currentlabel{%
    [\@tempa][\arabic{#1}][\cref@result]%
    \csname p@#1\endcsname\csname the#1\endcsname}}
\makeatother

\arxiv{
\title{Regularized Box-Simplex Games \\ and Dynamic Decremental Bipartite Matching}

\author{Arun Jambulapati ~~~ Yujia Jin ~~~ Aaron Sidford  ~~~ Kevin Tian\\
Stanford University\\
	\texttt{\{\href{mailto:jmblpati@stanford.edu}{jmblpati},%
		\href{mailto:yujiajin@stanford.edu}{yujiajin},%
		\href{mailto:sidford@stanford.edu}{sidford},%
		\href{mailto:kjtian@stanford.edu}{kjtian}%
		\}@stanford.edu}}
\date{}
}

\notarxiv{
\title{Regularized Box-Simplex Games and Dynamic Decremental Bipartite Matching} %

\author{Arun Jambulapati}{Stanford University}{jmblpati@stanford.edu}{}{}
\author{Yujia Jin}{Stanford University}{yujiajin@stanford.edu}{https://web.stanford.edu/~yujiajin/}{supported by a Stanford Graduate Fellowship and the Dantzig-Lieberman Operations Research Fellowship}
\author{Aaron Sidford}{Stanford University}{sidford@stanford.edu}{www.aaronsidford.com}{supported in part by a Microsoft Research Faculty Fellowship, NSF CAREER Award CCF-1844855, NSF Grant CCF-1955039, a PayPal research award, and a Sloan Research Fellowship.}
\author{Kevin Tian}{Stanford University}{ktian6@gmail.com}{https://kjtian.github.io/}{supported in part by a Google Ph.D. Fellowship, a Simons-Berkeley VMware Research Fellowship, a Microsoft Research Faculty Fellowship, NSF CAREER Award CCF-1844855, NSF Grant CCF-1955039, and a PayPal research award. }%

\authorrunning{Arun, Yujia, Aaron, and Kevin} %

\Copyright{Arun Jambulapati, Yujia Jin, Aaron Sidford, and Kevin Tian} %

\ccsdesc{Theory of computation~Dynamic graph algorithms} %

\keywords{bipartite matching, decremental matching, dynamic algorithms, continuous optimization, box-simplex games, primal-dual method} %

\relatedversion{Full version available on arXiv.}

\acknowledgements{We thank anonymous reviewers for their feedback.}%

\EventEditors{John Q. Open and Joan R. Access}
\EventNoEds{2}
\EventLongTitle{42nd Conference on Very Important Topics (CVIT 2016)}
\EventShortTitle{CVIT 2016}
\EventAcronym{CVIT}
\EventYear{2016}
\EventDate{December 24--27, 2016}
\EventLocation{Little Whinging, United Kingdom}
\EventLogo{}
\SeriesVolume{42}
\ArticleNo{23}
}

\begin{document}

\maketitle

\begin{abstract}%
Box-simplex games are a family of bilinear minimax objectives which encapsulate graph-structured problems such as maximum flow \cite{Sherman17}, optimal transport \cite{JambulapatiST19}, and bipartite matching \cite{AssadiJJST22}.
We develop efficient near-linear time, high-accuracy solvers for regularized variants of these games. Beyond the immediate applications of such solvers for computing Sinkhorn distances, a prominent tool in machine learning, we show that these solvers can be used to obtain improved running times for maintaining a (fractional) $\epsilon$-approximate maximum matching in a dynamic decremental bipartite graph against an adaptive adversary. We give a generic framework which reduces this dynamic matching problem to solving regularized graph-structured optimization problems to high accuracy. Through our reduction framework, our regularized box-simplex game solver implies a new algorithm for dynamic decremental bipartite matching in total time $\widetilde{O}(m \cdot \eps^{-3})$, from an initial graph with $m$ edges and $n$ nodes. We further show how to use recent advances in flow optimization \cite{chen2022maximum} to improve our runtime to $m^{1 + o(1)} \cdot \eps^{-2}$, thereby demonstrating the versatility of our reduction-based approach. These results improve upon the previous best runtime of~$\widetilde{O}(m \cdot \eps^{-4})$~\cite{BernsteinGS20} and illustrate the utility of using regularized optimization problem solvers for designing dynamic algorithms. 
\end{abstract}
\arxiv{
\newpage

 \tableofcontents

\thispagestyle{empty}
\newpage
\pagenumbering{arabic} 
}

\section{Introduction}\label{sec:intro}

Efficient approximate solvers for graph-structured convex programming problems have led to a variety of recent advances in combinatorial optimization. Motivated by problems related to maximum flow and optimal transportation, a recent line of work \cite{Sherman13, KelnerLOS14, Sherman17, SidfordTi18, JambulapatiST19, CohenST21} developed near-linear time, accelerated solvers for a particular family of convex programming objectives we refer to in this paper as \emph{box-simplex games}:
\begin{equation}\label{eq:boxsimplexintro}\min_{x \in \Delta^m} \max_{y \in [-1, 1]^n} y^\top \ma x + c^\top x - b^\top y
\text{ where }
\Delta^m \defeq \{x\in\R^m_{\geq 0} | \norm{x}_1 =1 \}\,.
\end{equation}
Box-simplex games, \eqref{eq:boxsimplexintro}, are bilinear problems where a maximization player is constrained to the box (the $\ell_\infty$ ball) and a minimization player is constrained to the simplex (the nonnegative $\ell_1$ shell). The problem provides a convenient encapsulation of linear programming problems with $\ell_1$ or $\ell_\infty$ structure; \eqref{eq:boxsimplexintro} can be used to solve box-constrained $\ell_\infty$ regression problems (see e.g.\ \cite{Sherman17, SidfordTi18}) and maximizing over the box-constrained player yields the following $\ell_1$ regression problem
\begin{equation}\label{eq:l1regintro}\min_{x \in \Delta^m} c^\top x + \norm{\ma x - b}_1.\end{equation}
Furthermore, solvers for \eqref{eq:boxsimplexintro} and \eqref{eq:l1regintro}  are used in state-of-the-art algorithms for approximate maximum flow \cite{Sherman17}, optimal transport (OT) \cite{JambulapatiST19}, (width-dependent) positive linear programming \cite{BoobSW19}, and semi-streaming bipartite matching \cite{AssadiJJST22}. 

One of the main goals of our work is to develop efficient algorithms for solving \emph{regularized} variants of the problems \eqref{eq:boxsimplexintro} and \eqref{eq:l1regintro}. An example of particular interest is the following
\begin{equation}\label{eq:sinkhornintro}\min_{x \in \Delta^m \mid \mb^\top x = d} c^\top x + \mu H(x),
\text{ where }
\mu \ge 0
\text{ and }
H(x) \defeq \sum_{i \in [m]} x_i \log x_i.
\end{equation}
The particular case of \eqref{eq:sinkhornintro} when $\mb \in \R^{m \times n}$ is the (unsigned) edge-vertex incidence matrix of a complete bipartite graph, and $d$ is a pair of discrete distributions supported on the sides of the bipartition, is known as the \emph{Sinkhorn distance} objective. This objective is used in the machine learning literature \cite{Cuturi13} as an efficiently-computable approximation to optimal transport distances: $c$ corresponds to pairwise movement costs, and $d$ encodes the prescribed marginals. This objective has favorable properties, e.g.\ differentiability \cite{Vialard19}, and there has been extensive work by both theorists and practitioners to solve \eqref{eq:sinkhornintro} and analyze its properties (see e.g.\ \cite{Cuturi13, AltschulerWR17} and references therein). Choosing $\ma$ and $b$ to be sufficiently large multiples of $\mb^\top$ and $d$, it can be shown that solutions to the following regularized variant of \eqref{eq:l1regintro} yield approximate solutions to \eqref{eq:sinkhornintro},
\begin{equation}
	\label{eq:regl1intro}\min_{x \in \Delta^m} c^\top x + \norm{\ma x - b}_1 + \mu H(x)
\,.
\end{equation} 

Beyond connections to Sinkhorn distances, there are additional reasons why it may be desirable to solve regularized box-simplex games. For example, regularization could speed up algorithms and allow high-precision solutions to be computed more efficiently. Further, obtaining a high-precision solution to a regularized version of the problem yields a more canonical and predictable approximate solution than an arbitrary low-precision approximation to the unregularized problem. Moreover, regularization potentially makes optimal solutions more stable to input changes. For box-simplex games stemming from bipartite matching we quantify this stability and show all of these properties allow regularized solvers to yield faster algorithms for a particular dynamic matching problem. 

Altogether, the main contributions of this paper are the following.
\begin{enumerate}
	\item We give improved running times for the problem of \emph{dynamic decremental bipartite matching} (DDBM) with an \emph{adaptive adversary}, a fundamental problem in dynamic graph algorithms. Our algorithm follows from a general black-box reduction we develop from DDBM to solving (variants of) regularized box-simplex games to high precision.
	\item We give efficient solvers for (variants of) the regularized box-simplex problems \eqref{eq:sinkhornintro}, \eqref{eq:regl1intro}.
	\item As a byproduct, we also show how to apply our new solvers (and additional tools from the literature) to obtain state-of-the-art methods for computing Sinkhorn distances.
\end{enumerate}

Formally, the \emph{DDBM} problem we consider in this paper is the following: given a bipartite graph undergoing edge deletions maintain, at all times, an \emph{$\epsilon$-approximate (maximum) matching},\footnote{This is sometimes also referred to as a $(1 + \eps)$-multiplicatively approximate matching in the literature.} that is a matching which has size at least a $(1 - \epsilon)$-fraction of the maximum (for a pre-specified) value of $\epsilon$. Unless specified otherwise, we consider the \emph{adaptive adversary model} where edge deletions can be specified adaptively to the matching returned. Further, unless specified otherwise, we allow the matching output by the algorithm to be \emph{fractional}, rather than integral.%

We show how to reduce solving the DDBM problem to solving a sequence of regularized box-simplex games. This reduction yields a new approach to dynamic matching; this approach is inspired by prior work, e.g.\ \cite{BernsteinGS20}, but conceptually distinct in that it decouples the solving of optimization subproblems from characterizing their solutions. For our specific DDBM problem, the only prior algorithm achieving an amortized polylogarithmic update time (for constant $\eps$) is in the recent work of \cite{BernsteinGS20}, which derives their dynamic algorithm as an application of the \emph{congestion balancing} technique. Our reduction eschews this combinatorial tool and directly argues, via techniques from convex analysis, that solutions to appropriate regularized matching problems can be used dynamically as approximate matchings while requiring few recomputations. We emphasize our use of fast \emph{high-accuracy solvers}\footnote{Throughout, we typically use the term ``high-accuracy'' to refer to an algorithm whose runtime scales polylogarithmically in the inverse accuracy (as opposed to e.g.\ polynomially).} in the context of our reduction to obtain our improved runtimes, as our approach leverages structural characteristics of the exact solutions which we only show are inherited by approximate solutions when solved to sufficient accuracy.

Our work both serves as a proof-of-concept of the utility of regularized linear programming solvers as a subroutine in dynamic graph algorithms, and provides the tools necessary to solve said problems in various structured cases. This approach to dynamic algorithm design effectively separates a ``stability analysis'' of the solution to a suitable optimization problem from the computational burden of solving that problem to high accuracy: any improved solver would then have implications for faster dynamic algorithms as well. As a demonstration of this flexibility, we give three  uses of our reduction framework (which proceed via different solvers) in obtaining our improved DDBM update time.
We hope our work opens the door to exploring the use of the powerful continuous optimization toolkit, especially techniques originally designed for non-dynamic problems, for their dynamic counterparts.

\paragraph{Paper organization.} We give a detailed overview of our contributions in Section~\ref{ssec:results}, and overview related prior work in Section~\ref{ssec:mainrelated}. We state preliminaries in Section~\ref{sec:prelims}. In Section~\ref{sec:dec_framework}, we describe our framework for reducing DDBM to a sequence of regularized optimization problems satisfying certain properties, and in Section~\ref{ssec:solver} we give three different instantiations of the DDBM framework for obtaining a variety of DDBM solvers. Finally in Section~\ref{sec:framework} we provide our main algorithm for regularized box-simplex games. \notarxiv{In~\Cref{app:maxflow}, we provide additional discussions on a recent advancement for faster DDBM solver. }We defer proofs for~\Cref{sec:dynamic} and~\Cref{sec:framework} to~\Cref{app:dynamic} and~\Cref{app:framework} respectively, and provide additional results for approximating Sinkhorn distances efficiently in~\Cref{sec:apprx}.

\subsection{Our results}\label{ssec:results} 

\paragraph{A framework for faster DDBM.}
 We develop a new framework for solving the DDBM problem of computing an $\epsilon$-approximate maximum matching in a dynamic graph undergoing edge deletions from an adaptive adversary.  Our framework provides a reduction from this DDBM problem to solving various regularized formulations of box-simplex games. 

To illustrate the reduction, suppose we have a bipartite graph $G = (V, E)$, and suppose, for simplicity, that we know $M^*$, the size of the (maximum cardinality) matching. As demonstrated in \cite{AssadiJJST22}, solving the $\ell_1$ regression problem $\min_{x \in M^* \Delta^m} -c^\top x + \norm{\ma x - b}_1$,
to $\eps M^*$ additive accuracy for appropriate choices of $\ma$, $b$, and $c$ yields an $\eps$-approximate maximum cardinality matching. Intuitively, $\ma$ and $b$ penalize violations of the matching constraints, and $c$ is a multiple of the all-ones vector capturing the objective of maximizing the matching size. However, $\ell_1$ regression objectives do not necessarily have unique minimizers: as such the output of directly minimizing these objectives is difficult to characterize beyond (approximate) optimality. This induces difficulty in using solutions to such problems directly in  dynamic graph algorithms. 

Our first key observation (building upon intuition from congestion balancing \cite{BernsteinGS20}) is that, beyond enabling faster runtime guarantees, regularization provides more robust solutions which are resilient to edge deletions in dynamic applications. We show that if
\begin{equation}\label{eq:matchreg}x^*_\eps \defeq \min_{x \in M^*\cdot\Delta^m} -c^\top x + \norm{\ma x - b}_1 + \eps H(x)\end{equation}
is the solution to the \emph{regularized} box-simplex formulation of bipartite matching, then $x^*_\eps$ enjoys favorable stability properties allowing us to argue about its size under deletions. 

The stability of solutions to \eqref{eq:matchreg} is fairly intuitive: the entropy regularizer $H(x)$ encourages the objective to spread the matching uniformly among edges, when all else is held equal. For example, when $G$ is a complete bipartite graph on $2 n$ vertices, standard linear programming relaxations of matching do not favor either of (i) an integral perfect matching, and (ii) a fractional matching spreading mass evenly across many edges, over the other. However, using (i) as our approximate matching on a dynamic graph undergoing deletions is substantially more unstable; an adaptive adversary can remove edges corresponding to our matching, forcing $\Omega(n)$ recomputations. On the other hand, no edge deletions can cause this type of instability for strategy (ii): as each edge receives weight $\frac 1 n$ in the fractional matching, the only way to reduce the fractional matching size by $\epsilon n$ is to remove $O(\epsilon n^2)$ edges: thus $O(\eps^{-1})$ recomputations are (intuitively) sufficient for maintaining an $\eps$-approximate maximum matching. This distinction underlies the use of \emph{high-accuracy} solvers in our reduction; indeed, while they obtain large matching values in an original graph, approximate solutions may not carry the same types of dynamic matching value stability. We note similar intuition motivated the approach in \cite{BernsteinGS20}.

To make this argument more rigorous, consider using $x^*_\eps$ as our approximate matching for a number of iterations corresponding to edge deletions, until its size restricted to the smaller graph has decreased by a factor of $1 - O(\eps)$. By using strong convexity of \eqref{eq:matchreg} in the $\ell_1$ norm, we argue that whenever the objective value of $x^*_\eps$ has worsened, the maximum matching size itself must have gone down by a (potentially much smaller) amount. A tighter characterization of this strong convexity argument shows that we only need to recompute a solution to slight variants of \eqref{eq:matchreg} roughly $\tO(\eps^{-2})$ times throughout the life of the algorithm. Combined with accelerated $\tO(\frac m \eps)$-time solvers for regularized box-simplex games (which are slight modifications of (\ref{eq:matchreg})), this strategy yields an overall runtime of $\tO(\frac m {\eps^3})$, improving upon the recent state-of-the-art decremental result of \cite{BernsteinGS20}.

We formalize these ideas in Section~\ref{sec:dynamic}, where we demonstrate that a range of regularization strategies (see Definition~\ref{def:canonical}) such as \eqref{eq:matchreg} are amenable to this reduction. Roughly, as long as our regularized objective is ``at least as strongly convex'' as the entropic regularizer, and closely approximates the matching value in the static setting, then it can be used in our DDBM algorithm. Combining this framework with solvers for regularized matching problems, we give three different results. The first two obtain amortized update times of roughly $\tO(\eps^{-3})$, in Theorems~\ref{thm:main-dec-bs} and~\ref{thm:main-dec-sinkhorn} via box-simplex games and matrix scaling, respectively (though the latter holds only for dense graphs). We give an informal statement of the former here.

\begin{theorem}[informal, see Theorem~\ref{thm:main-dec-bs}]
	Let $G = (V, E)$ be bipartite, $|V|=n$, $|E|=m$, and $\eps \ge \textup{poly}(m^{-1})$. There is a deterministic algorithm maintaining an $\eps$-approximate matching in a dynamic bipartite graph with adversarial edge deletions running in time $O(m \log^{5} m\cdot \eps^{-3})$.
\end{theorem}

We note that our algorithm (deterministically) returns a \emph{fractional matching}. There is a black-box reduction from dynamic fractional matching maintenance to dynamic integral matching maintenance contained in \cite{Wajc20}, but to our knowledge this reduction is bottlenecked at an amortized $\tO(\eps^{-4})$ runtime (see e.g.\ Appendix A.2, \cite{BernsteinGS20}). Improving this reduction is a key open problem.

\paragraph{High-accuracy solvers for regularized box-simplex games.}
 To use our DDBM framework, we give a new algorithm for solving regularized box-simplex games of the form:
\begin{equation}\label{eq:regbsintro}\min_{x \in \Delta^m} \max_{y\in[0,1]^n} f_{\mu, \eps}(x, y) \defeq y^\top \ma^\top x + c^\top x - b^\top y + \mu H(x) - \frac \eps 2 \Par{y^2}^\top |\ma|^\top x, \end{equation}
where $\eps$ and $\mu = \Omega(\eps)$ are regularization parameters and $y^2$, $|\ma|$ denote entrywise operations. The regularization terms $H(x)$ and $(y^2)^\top |\ma|^\top x$ in \eqref{eq:regbsintro} are parts of a primal-dual regularizer proposed in \cite{JambulapatiST19} (and a variation of a similar regularizer of \cite{Sherman17}) used in state-of-the-art algorithms for approximately solving (unregularized) box-simplex games. This choice of regularization enjoys favorable regularization properties over the joint box-simplex domain, and thereby sidesteps the infamous $\ell_{\infty}$-strong convexity barrier that has limited previous attempts at accelerated algorithms for this problem. Under relatively mild restrictions on problem parameters (see discussion at the start of Section~\ref{sec:framework}), we develop a \emph{high accuracy solver} for \eqref{eq:regbsintro}, stated informally here.

\begin{theorem}[informal, see Theorem~\ref{thm:sherman}]\label{thm:shermaninf}
Given an instance of  \eqref{eq:regbsintro}, with $\mu = \Omega(\eps)$, $\norm{\ma}_\infty \le 1$, and $\sigma \ge \textup{poly}(m^{-1})$ Algorithm~\ref{alg:sherman} returns $x$ with $\max_{y \in [0, 1]^n} f_{\mu, \eps}(x, y) - f_{\mu, \eps}(x^\star, y^\star) \le \sigma$ in time $\tO(\nnz(\ma) \cdot \frac {1}{\sqrt{\mu\eps}})$ where  $(x^\star,y^\star)$ is the optimizer of~\eqref{eq:regbsintro}.
\end{theorem}

Our solver follows recent developments in solving unregularized box-simplex games. We analyze an approximate extragradient algorithm based on the mirror prox method of \cite{Nemirovski04}, and prove that iterates of the regularized problem \eqref{eq:regbsintro} enjoy multiplicative stability properties previously shown for the iterates of mirror prox on the unregularized problem \cite{CohenST21}. Leveraging these tools, we also show the regularizer-operator pair satisfies technical conditions known as \emph{relative Lipschitzness} and \emph{strong monotonicity}, thus enabling a similar convergence analysis as in \cite{CohenST21}. This yields an efficient algorithm for solving \eqref{eq:regbsintro}.

Roughly, when the scale of the problem (defined in terms of the matrix operator norm $\norm{\ma}_\infty$ and appropriate norms of $b$ and $c$) is bounded,\footnote{Our runtimes straightforwardly extend to depend appropriately on these norms in a scale-invariant way.} our algorithm for computing a high-precision optimizer to \eqref{eq:regbsintro} runs in $\tO(\frac 1 {\sqrt{\mu\eps}})$ iterations, each bottlenecked by a matrix-vector product through $\ma$. When $\mu \approx \eps$, the optimizer of the regularized variant is an $O(\eps)$-approximate solution to the unregularized problem \eqref{eq:boxsimplexintro}, and hence Theorem~\ref{thm:sherman} recovers state-of-the-art runtimes (scaling as $\tO(\eps^{-1})$) for box-simplex games up to logarithmic factors. We achieve our improved dependence on $\mu$ in Theorem~\ref{thm:sherman} by trading off the scales of the primal and dual domains. This type of argument is well-known for \emph{separable regularizers} \cite{CarmonJST19}, but a key technical novelty of our paper is demonstrating a similar analysis holds for non-separable regularizers compatible with box-simplex games e.g.\ the one from \cite{JambulapatiST19}, which has not previously been done. 
To our knowledge, Theorem~\ref{thm:sherman} is the first result for solving general regularized box-simplex games to high accuracy in nearly-linear time. We develop our box-simplex algorithm and prove Theorem~\ref{thm:sherman} in Section~\ref{sec:framework}.

\paragraph{Improved rates for the Sinkhorn distance objective.} We apply our accelerated solver for \eqref{eq:regbsintro} in computing approximations to the Sinkhorn distance objective \eqref{eq:sinkhornintro}, a fundamental algorithmic problem in the practice of machine learning, at a faster rate. It is well-known that solving the \emph{regularized} Sinkhorn problem \eqref{eq:sinkhornintro} with $\mu$ scaling much larger than the target accuracy $\eps$ enjoys favorable properties in practice \cite{Cuturi13} (compared to its unregularized counterpart, the standard OT distance). In~\cite{altschuler2019massively}, the authors show that Sinkhorn iteration studied in prior work solves \eqref{eq:sinkhornintro} to additive accuracy $\eps$ at an unaccelerated rate of $\tO(\frac 1 {\mu \eps})$. For completeness we provide a proof of this result (up to logarithmic factors) in Appendix~\ref{ssec:sinkhornunaccel}.

As a straightforward application of the solver we develop for \eqref{eq:regbsintro}, we demonstrate that we can attain an accelerated rate of $\tO(\frac 1 {\sqrt{\mu \eps}})$ for approximating \eqref{eq:sinkhornintro} to additive accuracy $\eps$ via a first-order method. More specifically, the following result is based on reducing the ``explicitly constrained'' Sinkhorn objective \eqref{eq:sinkhornintro} to a ``soft constrained'' regression variant of the form \eqref{eq:regl1intro}, where our box-simplex game solver is applicable. We now state our first result on improved rates for approximating Sinkhorn distance objectives.

\begin{theorem}[informal, see Theorem~\ref{thm:sinkhornmin}]\label{thm:sinkbsinf}
	Let $\mu \in [\Omega(\eps), O(\frac{\norm{c}_\infty}{\log m})]$ in \eqref{eq:sinkhornintro} corresponding to a complete bipartite graph with $m$ edges. There is an algorithm based on the regularized box-simplex game solver of Theorem~\ref{thm:sherman} which obtains an $\eps$-approximate minimizer to \eqref{eq:sinkhornintro} in time $\tO(m \cdot \frac{\norm{c}_\infty}{\sqrt{\mu\eps}})$.
\end{theorem}

By leveraging the particular structure of the Sinkhorn distance and its connection to a primitive in scientific computing and theoretical computer science known as \emph{matrix scaling} \cite{LinialSW98, Cohen17, Allen-ZhuLOW17}, we give a further-improved solver for \eqref{eq:sinkhornintro} in Theorem~\ref{lem:ms-sinkhorn}. This solver has a nearly-linear runtime scaling as $\tO(\frac 1 \mu)$, which is a high-precision solver for the original Sinkhorn objective. Our high-precision Sinkhorn solver applies powerful second-order optimization tools from \cite{Cohen17} based on the \emph{box-constrained Newton's method} for matrix scaling, yielding our second result on improved Sinkhorn distance approximation rates.

\begin{theorem}[informal, see Theorem~\ref{lem:ms-sinkhorn}]
	Let $\mu, \eps = O(\norm{c}_\infty)$ in \eqref{eq:sinkhornintro} corresponding to a complete bipartite graph with $m$ edges. There is an algorithm based on the matrix scaling solver of \cite{Cohen17} which obtains an $\eps$-approximate minimizer to \eqref{eq:sinkhornintro} in time $\tO(m \cdot \frac{\norm{c}_\infty}{\mu})$.
\end{theorem}

We present both Theorems~\ref{thm:sinkhornmin} and~\ref{lem:ms-sinkhorn} because they follow from somewhat incomparable solver frameworks. While the runtime of Theorem~\ref{thm:sinkhornmin} is dominated by that of Theorem~\ref{lem:ms-sinkhorn}, it is a direct application of a more general solver (Theorem~\ref{thm:sherman}), which also applies to regularized regression or box-simplex objectives where the optimum does not have a characterization as a matrix scaling. Moreover, the algorithm of Theorem~\ref{lem:ms-sinkhorn} is a second-order method which leverages recent advances in solving Laplacian systems, and hence may be less practical than its counterpart in Theorem~\ref{thm:sinkhornmin}. Finally, we note that due to subtle parameterization differences for our DDBM applications, the DDBM runtime attained by using our box-simplex solver within our reduction framework is more favorable on sparse graphs $(m \ll n^2)$, compared to that obtained by the matrix scaling solver.

Since our results on optimizing Sinkhorn distances follow from machinery developed in this paper and \cite{Cohen17}, we defer them to Appendix~\ref{sec:apprx}. For consistency with the optimal transport literature and ease of presentation, we state our results in Appendix~\ref{sec:apprx} for instances of \eqref{eq:sinkhornintro} corresponding to complete bipartite graphs. However, our algorithms based on Theorem~\ref{thm:sherman} extend naturally to sparse graphs, as do the matrix scaling algorithms of \cite{Cohen17} as discussed in that work. We provide applications of these subroutines to the DDBM problem on sparse graphs in Section~\ref{sec:dynamic} and Appendix~\ref{app:dynamic}.

\notarxiv{
\paragraph{A recent advancement.} Subsequent to the original submission of this paper, a breakthrough result of \cite{chen2022maximum} provided an algorithm which computes high-accuracy solutions to a variety of graph-based optimization objectives in almost-linear time. By using a matrix scaling subroutine provided by \cite{chen2022maximum} within our framework, we obtain a (randomized) DDBM solver running in time $m^{1 + o(1)}\eps^{-2}$, improving Theorem~\ref{thm:main-dec-bs} in the $\eps^{-1}$ dependence at a cost of an $m^{o(1)}$ overhead. We defer a more detailed discussion of \cite{chen2022maximum} and its implications for our work to Appendix~\ref{app:maxflow}.
}
\arxiv{
\paragraph{A recent advancement.} 
Subsequent to the original submission of this paper, a breakthrough result of  \cite{chen2022maximum} provided an algorithm which computes high-accuracy solutions to a variety of graph-based optimization objectives in almost-linear time. One of the many applications of \cite{chen2022maximum} is faster algorithms for computing Sinkhorn distances. Using the algorithm designed in~\cite{chen2022maximum}, in place of~\cite{Cohen17}, yields  a speedup of Theorem~\ref{lem:ms-sinkhorn}, removing the polynomial dependence on $\mu^{-1}$ at the cost of an overhead of $m^{o(1)}$.

\begin{theorem}[informal, see~\Cref{lem:dec-solver-maxflow}]
	For $\eps = \Omega(m^{-3})$ in~\eqref{eq:sinkhornintro} corresponding to a bipartite graph with $m$ edges, there is an algorithm that obtains an $\eps$-approximate minimizer to~\eqref{eq:sinkhornintro} in time $m^{1+o(1)}$. 
\end{theorem}

Further, we show how to use \cite{chen2022maximum} to obtain improved running times for DDBM. By plugging in this Sinkhorn distance computation algorithm into our DDBM framework (and showing it is compatible with the form of our subproblems), we obtain an improved (randomized) DDBM solver in terms of the $\eps^{-1}$ dependence at the cost of a $m^{o(1)}$ overhead. 

\begin{theorem}[informal, see Theorem~\ref{thm:main-dec-sinkhorn-maxflow}]
	Let $G = (V, E)$ be bipartite, $|V|=n$, $|E|=m$, and $\eps \in [\Omega(m^{-3},1)$. There is a randomized algorithm with success probability $1-n^{-\Omega(1)}$  maintaining an $\eps$-approximate maximum matching in an (adaptive) decremental stream running in time $m^{1+o(1)}\epsilon^{-2}$.
\end{theorem}

We believe this result highlights the versatility and utility of our framework for solving DDBM. Given the varied technical machinery involved (and differing runtimes and use of randomness) our Sinkhorn algorithms, i.e.\ our new regularized box-simplex solver (see~\Cref{sec:framework}), \cite{Cohen17}, and \cite{chen2022maximum}, we provide statements of each solver instantiated by these different machinery in~\Cref{ssec:framework-alg}. However, due to the subsequent nature of \cite{chen2022maximum} with respect to our work, and because DDBM running times using  \cite{Cohen17} are no better than those achieved in our paper using our regularized box-simplex solver, the proofs of our DDBM solver based on~\cite{chen2022maximum} are deferred to~\Cref{ssec:dynamic-redx-maxflow}.
}

\subsection{Prior work}\label{ssec:mainrelated}

\paragraph{Dynamic matching.} Dynamic graph algorithms are an active area of research in the theoretical computer science, see e.g.\ \cite{HHS21,DurfeeGGP19,BernsteinGS20,Kiss21,GaoLP21,GoranciHP18,AbrahamDKKP16,GutenbergW20b,Goranci21,GoranciRST21,NanongkaiSW17,GoranciHP17a} and references therein. These algorithms have been developed under various dynamic graph models, including the additions and deletions on \emph{vertices}~or \emph{edges}, and \emph{oblivious adversary} model where the updates to the graph are fixed in advance (i.e.\ do not depend on randomness used by the algorithm), and the \emph{adaptive adversary} model in which updates are allowed to respond to the algorithm, potentially adversarially.  We focus on surveying deterministic dynamic matching algorithms with edge streams, which perform equally well under oblivious and adaptive updates; we remark the dynamic matching algorithms have also been studied under vertex addition and deletion model in~\cite{bosek2014online}. For a more in-depth discussion and corresponding developments in other settings, see \cite{Wajc20}.

Many variants of the particular dynamic problem of maintaining matchings in bipartite graphs have been studied, such as the \emph{fully dynamic} \cite{GuptaP13}, \emph{incremental} \cite{Gupta14, GrandoniLSSS19}, and \emph{decremental} \cite{BernsteinGS20} cases. However, known conditional hardness results \cite{HenzingerKNS15, KopelowitzPP16} suggest that attaining a polylogarithmic update time for maintaining an exact fully dynamic matching may be unattainable, prompting the study of restricted variants which require maintaining an approximate matching. The works most relevant to our paper are those of \cite{Gupta14}, which provides a $\tO(\eps^{-4})$ amortized update time algorithm for computing an $\eps$-approximate matching for incremental bipartite matching, and \cite{BernsteinGS20}, which achieves a similar $\tO(\eps^{-4})$ update time for decremental bipartite matching. Our main DDBM results, stated in Theorems~\ref{thm:main-dec-bs} and~\ref{thm:main-dec-sinkhorn}, improve upon \cite{BernsteinGS20} by roughly a factor of $\eps^{-1}$ in the decremental setting.

\paragraph{Box-simplex games.}
Box-simplex games, as well as $\ell_1$ and $\ell_\infty$ regression, are equivalent to linear programs in full generality \cite{LeeS15}, have widespread utility, and hence have been studied extensively by the continuous optimization community. Here we focus on discussing \emph{near-linear time} approximation algorithms, i.e.\ algorithms which run in time near-linear in the sparsity of the constraint matrix, potentially depending inverse polynomially on the desired accuracy. Interior point methods solve these problems with polylogarithmic dependence on accuracy, but are second-order and often encounter polynomial runtime overhead in the dimension (though there are exceptions, e.g.\ \cite{BrandLLSS0W21} and references therein).

A sequence of early works e.g.\ \cite{Nemirovski04, Nesterov05, Nesterov07} on primal-dual optimization developed first-order methods for solving games of the form \eqref{eq:boxsimplexintro}. These works either directly operated on the objective \eqref{eq:boxsimplexintro} as a minimax problem, or optimized a smooth approximation to the objective recast as a convex optimization problem. While these techniques obtained iteration complexities near-linear in the sparsity of the constraint matrix $\ma$, they either incurred an (unaccelerated) $\eps^{-2}$ dependence on the accuracy $\eps$, or achieved an $\eps^{-1}$ rate of convergence at the cost of additional dimension-dependent factors. This was due to the notorious ``$\ell_\infty$ strong convexity barrier'' (see Appendix A, \cite{SidfordTi18}), which bottlenecked classical acceleration analyses over an $\ell_\infty$-constrained domain. \cite{Sherman17} overcame this barrier by utilizing the primal-dual structure of \eqref{eq:boxsimplexintro} through a technique called ``area convexity'', obtaining a $\tO(\eps^{-1})$-iteration algorithm. Since then, \cite{CohenST21} demonstrated that fine-grained analyses of the classical algorithms of \cite{Nemirovski04, Nesterov07} also obtain comparable rates for solving \eqref{eq:boxsimplexintro}. Finally, we mention that area convexity has found applications in optimal transport and positive linear programming \cite{JambulapatiST19, BoobSW19}.

\paragraph{Sinkhorn distances.}
Since \cite{Cuturi13} proposed Sinkhorn distances for machine learning applications, a flurry of work has aimed at developing algorithms with faster runtimes for \eqref{eq:sinkhornintro}. A line of work by \cite{AltschulerWR17, DvurechenskyGK18, LinHJ19} has analyzed the theoretical guarantees of the classical Sinkhorn matrix scaling algorithm for this problem, due to the characterization of its solution as a diagonal rescaling of a fixed matrix. These algorithms obtain rates scaling roughly as $\tO(\norm{c}_\infty^2 \eps^{-2})$ for solving \eqref{eq:sinkhornintro} to additive accuracy $\eps$. Perhaps surprisingly, to our knowledge no guarantees for solving \eqref{eq:sinkhornintro} which improve as the regularization parameter $\mu$ grows are currently stated in the literature, a shortcoming addressed by this work. Finally, we remark that Sinkhorn iteration has also received extensive treatment from the theoretical computer science community, e.g.\ \cite{LinialSW98}, due to connections with algebraic complexity; see \cite{GargGOW20} for a recent overview of these connections. %

\section{Preliminaries}\label{sec:prelims}

\paragraph{General notation.} We denote $[n] \defeq \{1,2,\ldots, n\}$ and let $\0$ and $\1$ denote the all-$0$ and all-$1$ vectors. Given $v\in\R^d$, $v_i$ or $[v]_i$ denotes the $i^{\text{th}}$ entry of $v$, and for any subset $E\subseteq [d]$ we use $v_E$ or $[v]_E$ to denote the vector in $\R^d$ zeroing out $v$ on entries outside of $E$. We use $([v]_i)_+ = \max([v]_i,0)$ to denote the operation truncating negative entries. We use $v\circ w$ to denote elementwise multiplication between any $v,w\in\R^d$.  Given matrix $\ma\in\R^{m\times n}$, we use $\ma_{ij}$ to denote its $(i,j)^{\text{th}}$ entry, and denote its $i^{\text{th}}$ row and $j^{\text{th}}$ column by $\ai$ and $\aj$ respectively; its nonzero entry count is $\nnz(\ma)$. We use $\diag{v}$ to denote the diagonal matrix where $[\diag{v}]_{ii} = v_i$, for each $i$. Given two quantities $M$ and $M'$, for any $c>1$ we say $M$ is a $c$-approximation to $M'$ if it satisfies $\frac{1}{c}M'\le M\le cM'$. For $\eps\ll 1$, we say $M$ is an $\eps$(-multiplicative)-approximation of $M'$ if $(1-\eps)M'\le M\le (1+\eps)M'$. Throughout the paper, we use $|\ma|$ to denote taking the elementwise absolute value of a matrix $\ma$, and $v^2$ to denote the elementwise squaring of a vector $v$ when clear from context.

\paragraph{Norms.} $\norm{\cdot}_p$ denotes the $\ell_p$ norm of a vector or corresponding operator norm of a matrix. In particular, $\norm{\ma}_\infty = \max_{i}\norm{\ma_{i:}}_1$. We use $\norm{\cdot}$ interchangeably with $\norm{\cdot}_2$. We use $\Delta^m$ to denote an $m$-dimensional simplex, i.e.\ $x\in\Delta^m \iff x\in\R^d_{\ge0},\norm{x}_1=1$.

\paragraph{Graphs.} A graph $G=(V,E)$ has vertices $V$ and edges $E$; we abbreviate $n \defeq |V|$ and $m \defeq |E|$ whenever the graph is clear from context. For bipartite graphs,  $V = L\cup R$ denotes the bipartition. We let $\mb\in\{0,1\}^{E\times V}$ be the (unsigned edge-vertex) incidence matrix with $\mb_{ev} = 1$ if $v$ is an endpoint of $e$ and $\mb_{ev} = 0$ otherwise. 

\paragraph{Bregman divergence.} Given any convex distance generating function (DGF) $q(x)$, we use $V^q_{x'}(x) = q(x)-q(x')-\inprod{\nabla q(x')}{x-x'}\ge0$ as its induced Bregman divergence. When the DGF is clear from context, we abbreviate $V \defeq V^q$. By definition, $V$ satisfies
\begin{equation}\label{eq:three-point}
	\inprod{-\nabla V_{x'}(x)}{x-u} = V_{x'}(u) - V_{x}(u)-V_{x'}(x)~\text{for any}~x,x',u.
\end{equation}

\paragraph{Computational model.} We use the standard word RAM model, where one can perform each basic arithmetic operations on $O(\log n)$-bit words in constant time.  %

\section{Dynamic decremental bipartite matching}~\label{sec:dynamic}
Here we provide a reduction from maintaining an approximately maximum matching in a decremental bipartite graph to solving regularized matching problems to sufficiently high precision. In \Cref{sec:dec_framework} we give this framework and then, in \Cref{ssec:solver}, we provide various instantiations of our framework based on different solvers, to demonstrate its versatility.

\subsection{DDBM framework}
\label{sec:dec_framework}

Here we provide our general framework for solving DDBM, which assumes that for bipartite $G = (V,E)$ and approximate matching value $M$ there is a canonical regularized matching problem with properties given in \Cref{def:canonical}; we later provide multiple such examples. Throughout this section, $\mathrm{MCM}(E)$ denotes the size of the maximum cardinality matching on edge set $E$; the vertex set $V$ is fixed throughout, so we omit it in definitions.

\begin{definition}[Canonical regularized objective]\label{def:canonical}

Let $G = (V, E_0)$ be a bipartite graph and $M \geq 0$ be an $8$-approximation of $\MCM(E_0)$. For all $E \subseteq E_0$ with $\MCM(E) \ge \frac M 8$, let $f_{M,E} : \R_{\ge0}^{E}\rightarrow \R$,
 \begin{equation}\label{eq:nuxdef}
 	\nu^E \defeq \min_{x \in \R_{\ge 0}^E} f_{M, E}(x)
 	\text{ , and }
 	x^E \defeq \argmin_{x \in \R_{\ge 0}^E} f_{M, E}(x).\end{equation}
 We call the set of $\{f_{M,E}\}_{\MCM(E) \ge \frac M 8}$ a family of $(\eps,\beta)$-\emph{canonical regularized objectives} (CROs) for $G(E_0)$ and $M$ if the following four properties hold.

	\begin{enumerate}
	\item \label{item:def11} For all $E \subseteq E_0$ with $\MCM(E) \ge \frac M 8$, $-\nu^{E}$ is an $\frac \epsilon 8$-approximation of $\MCM(E)$.
	\item \label{item:def12} For all $E \subseteq E_0$ with $\MCM(E) \ge \frac M 8$, $f_{M,E}$ is equivalent to $f_{M,E_0}$ with the extra constraint that $x_{E_0 \setminus E}$ is fixed to $0$ entrywise.
	\item \label{item:def13} For any $E'\subseteq E \subseteq E_0$ with $\MCM(E) \geq \frac M 8$ and $\MCM(E') \ge \frac M 8$, 
	\begin{equation}\label{eq:cond-canonical-one}
		f_{M,E'}(x^{E'}) - f_{M,E}(x^{E})\ge \beta V^H_{x^{E}}(x^{E'})\text{ where }H(x) \defeq \sum_e x_e \log x_e 
	\end{equation}
	\item \label{item:def14} For any $x \in \R^E_{\ge 0}$ such that $8Mx$ is a feasible matching on $(V,E)$, 
	\begin{equation}\label{eq:cond-canonical-two}
		8M\norm{x_{E}}_1-\frac{\epsilon}{128}M \le -f_{M,E}(x)\le 8M\norm{x_{E}}_1+\frac{\epsilon}{128}M.
	\end{equation}
	\end{enumerate}
\end{definition}

We further define the following notion of a canonical solver for a given CRO, which solves the CRO to sufficiently high accuracy, and rounds the approximate solution to feasibility.

\begin{definition}[Canonical solver]\label{def:solver}
	For $(\eps,\beta)$-CROs $\{f_{M,E}\}_{\MCM(E) \ge \frac M 8}$, we call $\mathcal{A}$ an $(\eps,\mathcal{T})$-\emph{canonical solver} if it has subroutines $\solve$ and $\Round$ running in $O(\tsolve)$ time, satisfying:
\begin{enumerate}
	\item $\solve$ finds an approximate solution $\hat{x}^E$ of $f_{M,E}$ satisfying %
	\begin{subequations}
		\begin{align}
\Par{1+\frac{\epsilon}{8}}\nu^E\le f_{M,E}(\hat{x}^E)\le \Par{1-\frac{\epsilon}{8}}\nu^E. \label{eq:cond-apprx-function} \\
\norm{\hat{x}^E - x^E}_1\le \frac{\eps}{1100}.\label{eq:cond-apprx-l1}
\end{align}
 \end{subequations}
\item $\Round$ takes $\hat{x}^E$ and returns $\tilde{x}^E$ where $8M\tilde{x}^E$ is a feasible matching for $G(E)$, and:
\begin{subequations}
	\begin{align}
	\Par{1+\frac{\epsilon}{8}}\nu^E\le f_{M,E}(\tilde{x}^E)\le \Par{1-\frac{\epsilon}{8}}\nu^E.\label{eq:cond-round-guarantee} \\
		\tilde{x}^E\le \hat{x}^E~\text{monotonically}.\label{eq:cond-round-guarantee-l1}
	\end{align}
\end{subequations}
\end{enumerate}
\end{definition}

Our DDBM framework, \Cref{alg:decmatch}, uses CRO solvers satisfying the approximation guarantees of~\Cref{def:solver} to dynamically maintain an approximately maximum matching. We state its correctness and runtime in~\Cref{prop:decmatch}, and defer a proof to~\Cref{app:dec_framework}.

\begin{algorithm}[ht!]
	\DontPrintSemicolon
	\KwInput{$\eps\in(0,\frac 1 8)$, graph $G=(V,E)$ 
	}
	\Parameter{Family of CROs $\{f_{M,E}\}_{E \text{ is MP}}$, an $(\eps,\mathcal{T})$-canonical solver $(\solve, \Round)$}
	Compute $M$ with $\frac{1}{2}\MCM(E)\le M\le \MCM(E)$, via the greedy algorithm\;
	$\hat{x}^E \gets \solve(f_{M,E})$\;\label{line:dec-req-1}
	$\tilde{x}^E \gets \Round(\hat{x}^E)$\;
	$\Mest \gets M$\;
	\While{$\Mest > \frac{1}{4}M$}{\label{line:terminate-dec-outer}
	$E_\del \gets \emptyset$\;
		\While{edge $e$ is deleted and $\|\tilde{x}^E_{E_\del}\|_1\le \frac{\eps}{8}\norm{\tilde{x}^E}_1$}{\label{line:inner-dec-start}\Comment*{recompute whenever the deleted approximate matching size reaches a factor $\Theta(\eps)$}
		$E_\del \gets E_\del \cup\{e\}$\label{line:inner-dec-end}\;
	}
	$E\gets E\setminus E_\del$, $E_\del \gets \emptyset$ \label{line:recompute}\;
	$\hat{x}^E\gets\solve(f_{M,E})$\Comment*{find high-accuracy minimizer of $F_{M,E}$ satisfying~\eqref{eq:cond-apprx-function} and~\eqref{eq:cond-apprx-l1}}\label{line:hatx-dec-approximation}
		$\tilde{x}^E \gets \Round(\hat{x}^E)$\Comment*{round to feasible matching satisfying~\eqref{eq:cond-round-guarantee} and~\eqref{eq:cond-round-guarantee-l1}}
	Compute $\Mest$ with $\frac{1}{2}\MCM(E)\le \Mest \le \MCM(E)$, via the greedy algorithm
		}
	\caption{$\MaintainMatch(\eps, G = (V,E))$}\label{alg:decmatch}
\end{algorithm}

In the following, we let $E_0$ be the original graph's edge set, and $E_1, E_2, \ldots , E_K$ be the sequence of edge sets recomputed in \Cref{line:recompute}, before termination for $E_{K+1}$  on \Cref{line:terminate-dec-outer}.

\begin{restatable}{proposition}{propdecmatch}\label{prop:decmatch}
Let $\eps\in(0, 1)$ and $M\ge0$. Given a family of $(\eps,\beta)$-CROs $\{f_{M,E}\}$ for $G = (V, E_0)$, and an $(\eps,\mathcal{T})$-canonical solver for the family,~\Cref{alg:decmatch} satisfies the following.
\begin{enumerate}
	\item When $\Mest > \frac{1}{4}M$ on \Cref{line:terminate-dec-outer}, where $\Mest$ estimates $\MCM(E_k)$: at any point in the loop of Lines~\ref{line:inner-dec-start} to~\ref{line:inner-dec-end}, $8M \tilde{x}^{E_k}_{E_k \setminus E_{\textup{del}}}$ is an $\eps$-approximate matching of $G(V, E_k \setminus E_{\textup{del}})$.
	\item When $\Mest \leq \frac{1}{4}M$ on \Cref{line:terminate-dec-outer}, where $\Mest$ estimates $\MCM(E)$: $\MCM(E) \le \frac 1 {2} \MCM(E_0)$.
\end{enumerate} 
The runtime of the algorithm is $O(m + \Par{\mathcal{T}+m} \cdot \frac{M}{\beta\eps})$. 
\end{restatable}

\begin{proof-sketch} We summarize proofs of the two properties, and our overall runtime bound.

	\textbf{Matching approximation properties.} By the greedy matching guarantee in~\Cref{line:terminate-dec-outer}, it  holds that for any $E_k$ (the edge set recomputed in the $k^{\text{th}}$ iteration of \Cref{line:recompute} before termination), its true matching size $\MCM(E_k)$ must be no smaller than $\frac M 4$. Consequently, we can use the CRO family to approximate the true matching size up to $O(\epsilon)$ multiplicative factors, and by the guarantee~\eqref{eq:cond-round-guarantee}, this implies $8M \tilde{x}^{E_k}_{E_k \setminus E_{\textup{del}}}$ is an $O(\epsilon)$ approximation of the true matching size. Also, our algorithm's termination condition and the guarantee on $\Mest$ immediately imply $\MCM(E_{K+1})\le \frac{1}{2}\MCM(E_0)$.
	
\textbf{Iteration bound.} We use a potential argument. Given $E_{k+1}\subset E_k$, corresponding to consecutive edge sets requiring recomputation, we use the following inequalities:
\begin{equation}\label{eq:dec-progress}
	\begin{aligned}
f_{M,E_{k+1}}&\Par{x^{E_{k+1}}} - f_{M,E_{k}}\Par{x^{E_k}} \stackrel{(i)}{\ge} \beta V^H_{x^{E_k}}\Par{x^{E_{k+1}}}\\
&  \stackrel{(ii)}{\ge} \beta\sum_{i\in E_\del} \left(
[x^{{E_{k+1}}}]_i\log [x^{{E_{k+1}}}]_i - [x^{E_k}]_i\log [x^{E_k}]_i - \Par{1+\log [x^{E_k}]_i}\cdot \Par{[x^{{E_{k+1}}}]_i-[x^{E_k}]_i}
\right)\\
& \stackrel{(iii)}{=} \beta \sum_{i\in E_\del}[x^{E_k}]_i \stackrel{(iv)}{\ge} \beta \left(\norm{\hat{x}^{E_k}_{E_\del}}_1 - \norm{\hat{x}^{E_k} - x^{E_k}}_1\right) \stackrel{(v)}{\ge} \beta \left(\norm{\tilde{x}^{E_k}_{E_\del}}_1 - \norm{\hat{x}^{E_k} - x^{E_k}}_1\right),
\end{aligned}
\end{equation}
where $(i)$ uses the third property in~\eqref{eq:cond-canonical-one}, $(ii)$ uses convexity of the scalar function $c \log c$, $(iii)$ uses that $x^{E_{k+1}}_{E_\del}$ is $0$ entrywise, $(iv)$ uses the triangle inequality, and $(v)$ uses the monotonicity property~\eqref{eq:cond-round-guarantee-l1}. Moreover, between recomputations we have that the $\ell_1$-norm of deleted edges satisfies $\norm{\tilde{x}^{E_k}_{E_\del}}_1 = \Omega(\eps)$, and our solver guarantees $\norm{\hat{x}^{E_k} - x^{E_k}}_1 = O(\eps)$. Since the overall function decrease before termination is $O(M)$ given the stopping criterion in~\Cref{line:terminate-dec-outer}, the algorithm terminates after $O(\frac M {\beta\eps})$ recomputations.
\end{proof-sketch}

Using this generic DDBM framework, we obtain improved decremental matching algorithms by defining families of CROs $f_{M,E(G)}$ with associated $(\eps,\mathcal{T})$-canonical solvers satisfying~\Cref{def:solver}.  In~\Cref{ssec:dynamic-redx-bs}, we give a regularized primal-dual construction of $f_{M,E}$, and adapt the solver of~\Cref{sec:framework} to develop a canonical solver for the family (specifically, as the subroutine $\solve$). Similarly, in \Cref{ssec:dynamic-redx-high}, we show how to construct an appropriate family of $f_{M,E}$ using Sinkhorn distances, and apply the matrix scaling method presented in~\Cref{ssec:high-OT} (based on work of \cite{Cohen17}) to appropriately instantiate $\solve$. 

While both algorithms, as stated, only maintain an approximate fractional matchings, this fractional matching can be rounded at any point to an explicit integral matching via e.g.\ the cycle-canceling procedure of Proposition 3 in~\cite{AssadiJJST22} in time $O(m\log m)$, or dynamically (albeit at amoritized cost $\tO(\eps^{-4})$ using \cite{Wajc20}). Moreover, our algorithm based on the regularized box-simplex solver~(\Cref{thm:main-dec-bs}) is deterministic, and both work against an adaptive adversary. Repeatedly applying Proposition~\ref{prop:decmatch}, we obtain the following overall claim.

\begin{corollary}\label{coro:ddbm-reduce}
Let $G = (V, E(G))$ be bipartite, and suppose for any subgraph  $(V, E_0 \subseteq E(G))$, we are given a family of $(\eps, \beta)$-CROs and an $(\eps, \tsolve)$-canonical solver for the family. There is a deterministic algorithm maintaining a fractional $\eps$-approximate matching in a dynamic bipartite graph with adversarial edge deletions, running in time
$O\Par{m\log^3 n + (\tsolve + m)\cdot \frac{M}{\beta\eps} \cdot \log n}$.
\end{corollary}
\begin{proof}
It suffices to repeatedly apply \Cref{prop:decmatch} until we can safely conclude $\MCM(E) = 0$, which by the second property can only happen $O(\log n)$ times.
\end{proof}

\subsection{DDBM solvers}\label{ssec:solver}

In this section, we demonstrate the versatility of the DDBM framework in~\Cref{sec:dec_framework} by instantiating it with different classes of CRO families, and applying different canonical solvers on these families. By using regularized box-simplex game solvers developed in this paper (see~\Cref{sec:framework}), we give an $\widetilde{O}(m \epsilon^{-3})$ time algorithm for maintaining a $\epsilon$-multiplicatively approximate fractional maximum matching in a $m$-edge bipartite graph undergoing a sequence of edge deletions, improving upon the previous best running time of $\widetilde{O}(m \epsilon^{-4})$~\cite{BernsteinGS20}. We also use our framework to obtain different decremental matching algorithms with runtime $\widetilde{O}(n^2 \epsilon^{-3})$ and ${O}(m^{1+o(1)} \epsilon^{-2})$, building on recent algorithmic developments in the literature on matrix scaling. The former method uses box-constrained Newton's method solvers for matrix scaling problems in~\cite{Cohen17} (these ideas are also used in~\Cref{ssec:high-OT}), and the latter uses a recent almost-linear time high-accuracy Sinkhorn-objective solver in~\cite{chen2022maximum}, a byproduct of their breakthrough maximum flow solver. We defer readers to corresponding sections in~\Cref{app:dynamic} for omitted proofs.

Given a bipartite graph initialized at $G=(V,E_0)$ with unsigned incidence matrix $\mb\in\{0,1\}^{E\times V}$; we denote $n\defeq |V|$ and $m\defeq|E_0|$. The first family of CROs one can consider is the regularized box-simplex game objective in form:
\begin{equation}\label{eq:def-matching-abstract}
\begin{aligned}
\min_{(x, \xi)\in  \Delta^{E+1}}\max_{y\in[0,1]^{V}} f_{M,E}(x,\xi,y) &\defeq  -\1_E^\top(8Mx) - y^\top \Par{8M\mb^\top x-\1} +\gamma\x H(x, \xi)+\gamma\y \Par{y^2}^\top \mb^\top x,\\
&\text{where}~\gamma\x = \widetilde{\Theta}\left(\eps M\right),~\gamma\y = \Theta\left(\eps M\right),~\text{ and }\\
f_{M,E}(x) & \defeq \min_{\xi|(x,\xi)\in\Delta^{E+1}}\max_{y\in[0,1]^V}f_{M,E}(x,\xi,y)
\end{aligned}
\end{equation}

We prove this is a family of $(\epsilon,\gamma\x)$-CROs (\Cref{lem:canonical-bs}). The canonical solver for this family uses $\Remove$ (Algorithm 4, \cite{AssadiJJST22}) as $\Round$ and uses the regularized box-simplex games developed later in this paper (see~\Cref{sec:framework}) as $\solve$, which finds an $\eps$-approximate solution of~\eqref{eq:def-matching-abstract} in time $\widetilde{O}(\tfrac m \eps)$.  Combining all these components with the DDBM framework in~\Cref{coro:ddbm-reduce} leads to the following DDBM solver based on regularized box-simplex games.

\begin{restatable}{theorem}{thmmaindecbs}\label{thm:main-dec-bs}
Let $G = (V, E)$ be bipartite and let $\eps \in [\Omega(m^{-3}), 1)$. There is a deterministic algorithm for the DDBM problem which maintains an $\eps$-approximate matching, based on solving regularized box-simplex games, running in time
$O\Par{m \epsilon^{-3} \log^{5}n}$.
\end{restatable}

Our second CRO family is the following regularized Sinkhorn distance objective:
\begin{equation}\label{eq:def-matching-two-abstract}
\begin{aligned}
\min_{\begin{psmallmatrix}x \\ x^{\dum}\end{psmallmatrix} \in \R_{\ge 0}^{\tE} ~\mid~ 2|R|\tmb^\top\begin{psmallmatrix}x \\ x^{\dum}\end{psmallmatrix} = d}  f^\sink_{M,E}(x^{\tot}) & \defeq 2|R|\1_E^\top x +\gamma H\Par{x, x^\dum}~\text{where}~\gamma = \widetilde{\Theta}\Par{\eps M},\\
f^\sink_{M,E}(x) & \defeq\min_{x^\dum\in\R_{\ge0}^{E\setminus E_0}} f^\sink_{M,E}(x, x^\dum),
\end{aligned}
\end{equation}
where we extend graph $G=(V,E)$ to a balanced bipartite graph $\tG = (\tV, \tE)$ by introducing dummy vertices and edges. The extended graph allows us to write the inequality constraint $\mb^\top x=\1_V$ equivalently as the linear constraint $2|R|\widetilde{\mb}^\top \begin{psmallmatrix}x \\ x^{\dum}\end{psmallmatrix} = d$ for some defined $d\in\R^{\widetilde{V}}$ as some properly-extended vector of $\1_V$. This allows us to apply known matrix scaling solver to such approximating Sinkhorn distance objective in literature.

We prove this is a family of $(\epsilon,\gamma)$-CROs (\Cref{lem:canonical-sink}). The canonical solver for this family uses truncation to $E$ as $\Round$ and uses the matrix scaling solver from \cite{Cohen17} based on box-constrained Newton's method as $\solve$, which finds an $\eps$-approximate solution of~\eqref{eq:def-matching-two-abstract} in time $\widetilde{O}(n^2/\epsilon)$. Combining all these components with the DDBM framework in~\Cref{coro:ddbm-reduce} leads to the following DDBM solver based on approximating Sinkhorn distances.

\begin{restatable}{theorem}{thmmaindexsinkhorn}\label{thm:main-dec-sinkhorn}
Let $G = (V, E)$ be bipartite and $\eps \in [\Omega(m^{-3}),1)$. There is a randomized algorithm for the DDBM problem which maintains an $\eps$-approximate matching with probability $1 - n^{-\Omega(1)}$, based on matrix scaling solver of~\cite{Cohen17}, running in time $\tO\Par{n^2 \eps^{-3}}$.
\end{restatable}

Alternatively, for the same $(\epsilon,\gamma)$-CRO family as in~\eqref{eq:def-matching-two-abstract}, one can use the same $\Round$ procedure and the recent high-accuracy almost-linear time graph flow problems solver of~\cite{chen2022maximum} for $\solve$ as a canonical solver. Since this new solver can find high-accuracy solutions of entropic-regularized problems of the form~\eqref{eq:def-matching-two-abstract} within a runtime of $(|E_0|+O(|V|))^{1+o(1)} = m^{1+o(1)}$, this gives a third DDBM solver, which yields and improved  an dependence on $\eps^{-1}$.

    \begin{restatable}{theorem}{thmmaindecsinkhornmaxflow}\label{thm:main-dec-sinkhorn-maxflow}
    Let $G = (V, E)$ be bipartite and $\eps \in [\Omega(m^{-3}),1)$. There is a randomized algorithm for the DDBM problem which maintains an $\eps$-approximate matching with probability $1 - n^{-\Omega(1)}$, based on the Sinkhorn objective solver of \cite{chen2022maximum}, running in time $m^{1 + o(1)}\eps^{-2}$.
    \end{restatable}

\section{Regularized box-simplex games}\label{sec:framework}

In this section, we develop a high-accuracy solver for regularized box-simplex games:
\begin{equation}\label{eq:main-reg-apprx}
\begin{aligned}
\min_{x\in\Delta^m}\max_{y\in[0,1]^n} f_{\mu,\eps}(x,y) & \defeq y^
\top \ma^\top x + c^\top x -b^\top y + \mu H(x)-\frac{\eps}{2}\left(y^2\right)^\top|\ma|^\top x,\\
~~\text{where}~~ H(x) &\defeq \sum_{i\in[m]}x_i\log x_i~\text{is the standard entropic regularizer},
\end{aligned}
\end{equation}
where we recall absolute values and squaring act entrywise.

For ease of presentation, we make the following assumptions for some $\delta > 0$.
\begin{enumerate}
\item Upper bounds on entries: $\|\ma\|_{\infty} \le 1$, $\norm{b}_\infty \le \Bmax$, $\norm{c}_\infty \le \Cmax$. For simplicity, we assume $\Bmax \ge \Cmax \ge 1$; else, $\Cmax \gets \max(1, \Cmax)$ and $\Bmax \gets \max(\Cmax, \Bmax)$.
\item Lower bounds on matrix column entries: $\max_{i}|\ma_{ij}|\ge \delta$ for every $j\in[n]$.
\end{enumerate}
We defer the detailed arguments of why these assumptions are without loss of generality to~\Cref{app:framework}.
Our algorithm acts on the induced (monotone) gradient operator of the regularized box-simplex objective \eqref{eq:main-reg-apprx}, namely $(\nabla_x f_{\mu,\eps}(x, y), -\nabla_y f_{\mu,\eps}(x, y))$, defined as
\begin{equation}\label{def:grad-op}
g_{\mu, \eps}(x, y)\defeq \Par{\ma y + c + \mu(\1+\log(x)) - \frac{\eps}{2} \Abs{\ma} (y^2), -\ma^\top x + b + \eps \cdot \diag{y}\Abs{\ma}^\top x}. \end{equation} 
Further, it uses the following joint (non-separable) regularizer of
\begin{equation}\label{def:reg-sm}
	r_{\mu, \eps}(x, y) \defeq \rho \sum_{i \in [m]} x_i \log x_i + \frac{1}{\rho} x^\top \Abs{\ma} \Par{y^2}~~\text{where}~~\rho = \sqrt{\frac{2\mu}{\epsilon}},
\end{equation} 
variants of which have been used in \cite{Sherman17,JambulapatiST19,AssadiJJST22,CohenST21}. When clear from context, we drop subscripts and refer to these as operator $g$ and regularizer $r$. Our method is the first high-accuracy near-linear time solver for the regularized problem~\eqref{eq:main-reg-apprx}, yielding an $O(\veps)$-approximate solution with a runtime scaling polylogarithmically in problem parameters and $\veps$. We utilize a variant of the mirror prox (extragradient)~\cite{Nemirovski04} method for strongly monotone objectives, which appeared in~\cite{CarmonJST19, CohenST21} for regularized saddle point problems with separable regularizers.

In \Cref{ssec:high-level}, we present high-level ideas of our algorithm, which uses a mirror prox outer loop (\Cref{alg:outerloop}) and an alternating minimization inner loop (\Cref{alg:altmin}) to implement outer loop steps; we also provide convergence guarantees. In Section~\ref{ssec:helper}, we state useful properties of the regularizer~\eqref{def:reg-sm}, and discuss a technical detail ensuring iterate stability in our method. In \Cref{ssec:framework-alg}, we provide our full algorithm for regularized box-simplex games,~\Cref{alg:sherman} and give guarantees in~\Cref{thm:sherman}. Omitted proofs are in~\Cref{app:framework}.

\subsection{Algorithmic framework}\label{ssec:high-level}

In this section, we give the algorithmic framework we use to develop our high-precision solver, which combines an outer loop inspired by mirror prox \cite{Nemirovski04} with a custom inner loop for implementing each iteration. We first define an approximate solution for a proximal oracle.

\begin{definition}[Approximate proximal oracle solution]\label{def:approx-prox}
	Given a convex function $f$ over domain $\zset$ and $\veps\ge0$, we say $z' \in \zset$ is a \emph{$\veps$-approximate solution} for a proximal oracle if $z'$ satisfies $	\langle \nabla f(z'), z'-z\rangle\le \veps$. We denote this approximation property by $z' \gets_\veps \argmin_{z \in \zset} f$.
\end{definition}

We employ such approximate solutions as the proximal oracle within our ``outer loop`` method. Our outer loop is a variant of mirror prox (\Cref{alg:outerloop}) which builds upon both the mirror prox type method in~\cite{Sherman17} for solving unregularized box-simplex games and the high-accuracy mirror prox solver developed in~\cite{CarmonJST19, CohenST21} for bilinear saddle-point problems on geometries admitting separable regularizers. We first give a high-level overview of the analysis, which requires bounds on two properties. 
First, suppose $g$ is $\nu$-strongly monotone with respect to regularizer $r$, i.e.
\begin{equation}\label{def-strong-monotone}
	~\text{for any}~w,z\in\zset,~\inprod{g(w)-g(z)}{w-z}\ge\nu\inprod{\nabla r(w)-\nabla r(z)}{w-z}.
\end{equation}
Further, suppose it is $\alpha$-\emph{relatively Lipschitz} with respect to $r$ and Algorithm~\ref{alg:sherman} (see Definition 1 of \cite{CohenST21}), i.e.\ for any consecutive iterates $z_{k-1}$, $ z_{k-1/2}$, $z_{k}$ of our algorithm,\footnote{This property (i.e.\ relative Lipschitzness restricted to iterates of the algorithm) was referred to as ``local relative Lipschitzness'' in \cite{CohenST21}, but we drop the term ``local'' for simplicity.}
 \begin{equation}\label{def-rel-lip}
	~\inner{g(z_{k-1/2})-g(z_{k-1}))}{z_{k-1/2} - z_k}\le \alpha\Big(V_{z_{k-1/2}}(z_k)+V_{z_{k-1}}(z_{k-1/2})\Big).
\end{equation}
With these assumptions, we show that the strongly monotone mirror prox step makes progress by decreasing the divergence to optimal solution $V_{z_k}(z^\star)$ by a factor of $\frac \alpha {\alpha + \nu}$: this implies $\widetilde{O}(\frac \alpha \nu)$ iterations suffice for finding a high-accuracy solution. We provide the formal convergence guarantee of $\OuterLoop$ in~\Cref{prop:outerloopproof-sm}, which also accommodates the error of each approximate proximal step used in the algorithm. This convergence guarantee is generic and does not rely on the concrete structure of $g$, $r$ in our box-simplex problem. 

\begin{algorithm}
	\DontPrintSemicolon
	\KwInput{$\frac \veps 2$-approximate proximal oracle, operator and regularizer pair $(g, r)$ such that $g$ is $\nu$-strongly monotone and $\alpha$-relatively Lipschitz with respect to $r$}
	\Parameter{Number of iterations $K$}
	\KwOutput{Point $z_K$ with $ V_{z_K}(z^\star)\le 
	(\tfrac{\alpha}{\nu+\alpha})^K\Theta+\frac{\veps}{\nu}$}
	$z_0 \gets \argmin_{z\in\zset} r(z)$ \;
	\For{$k = 1, \ldots, K$}
	{
		$z_{k-1/2} \leftarrow_{\veps/2} \argmin_{z \in \zset}\left\{ 
		\inner{g\left(z_{k-1}\right)}{z} + 
		\alpha V_{z_{k-1}}(z)\right\}$  \; \label{line:outer-grad-sm}
		$z_k \leftarrow_{\veps/2} \argmin_{z \in \zset}\left\{ 
		\inner{g\left(z_{k-1/2}\right)}{z} + 
		\alpha V_{z_{k-1}}(z)+\nu V_{z_{k-1/2}}(z)\right\}$\; 
		\label{line:outer-extragrad-sm}
	}
	\Return $z_K$
	\caption{$\OuterLoop()$}\label{alg:outerloop}
\end{algorithm}

\begin{restatable}[Convergence of~\Cref{alg:outerloop}]{proposition}{outerloop}%
\label{prop:outerloopproof-sm}
Given regularizer $r$ with range at most $\Theta$, suppose $g$ is $\nu$-strongly-monotone with respect to $r$ (see \eqref{def-strong-monotone}), and is $\alpha$-relatively-Lipschitz with respect to $r$ (see \eqref{def-rel-lip}). Let $z_K$ be the output of Algorithm~\ref{alg:outerloop}. Then, $
V^r_{z_K}(z^\star)\le\left(\frac{\alpha}{\nu+\alpha}\right)^K\Theta
+\frac{\veps}{\nu}.$
\end{restatable}

Given the somewhat complicated nature of our joint regularizer, we cannot solve the proximal problems required by Algorithm~\ref{alg:outerloop} in closed form. Instead, we implement each proximal step to the desired accuracy by using an alternating minimization scheme, similarly to the implementation of approximate proximal steps in \cite{Sherman17, JambulapatiST19}. 

To analyze our algorithm, we use a generic progress guarantee for alternating minimization from~\cite{JambulapatiST19} to solve each subproblem, stated below.

\begin{algorithm}[ht!]
		\KwInput{$\ma \in \R^{m \times n}_{\ge 0}$, $\gamma\x \in \R^m$, $\gamma\y \in \R^n$, $T \in \N$, $\theta>0$, $x_{\text{init}}\in\Delta^m$, $y_{\text{init}}\in[0,1]^n$}
		\KwOutput{ Approximate minimizer to $\inprod{(\gamma\x, \gamma\y)}{z} + \theta r(z)$ for $r(z)$ in~\eqref{def:reg-sm}}
		 $x^{(0)} \gets x_{\text{init}}$, $y^{(0)} \gets y_{\text{init}}$\;
		\For{$0 \le t \le T$}
		{
	 $x^{(t + 1)} \gets \argmin_{x \in \Delta^m}\Brace{\inprod{\gamma\x}{x} + \theta r(x, y^{(t)})}$\;
		$y^{(t + 1)} \gets \argmin_{y \in [0, 1]^n}\Brace{\inprod{\gamma\y}{y} + \theta r(x^{(t + 1)}, y)}$\;
		}
		\Return $(x^{(T+1)}, y^{(T)})$
		\caption{$\AltMin(\gamma\x, \gamma\y, \ma, \theta, T, x_{\text{init}}, y_{\text{init}})$}\label{alg:altmin}
\end{algorithm}

\begin{lemma}[Alternating minimization progress, Lemma 5 and Lemma 7,~\cite{JambulapatiST19}]\label{lem:progress-altmin}
Let $r: \xset \times \yset \to \R$ be jointly convex, $\theta > 0$, and $\gamma^x$ and $\gamma^y$ be linear operators on $\xset, \yset$. Define
\begin{equation}\label{eq:subproblem}
	x_\OPT, y_\OPT = \argmin_{x\in\xset}\argmin_{y\in\yset} f(x,y)\defeq \langle \gamma^x, x\rangle+\langle\gamma^y,y\rangle+\theta r(x,y).
\end{equation}
Suppose $f(x,y)$ is twice-differentiable and satisfies: for all $x'\ge\frac{1}{2}x$ entrywise, $x',x\in\xset$ and $y',y\in\yset$, $\nabla^2 f(x',y')\succeq \frac{1}{\kappa}f(x,y)$. Then the iterates of~\Cref{alg:altmin} satisfy  
\[
f(x^{(t+2)},y^{(t+1)})-f(x_\OPT, y_\OPT)\le \left(1-\frac{1}{2\kappa}\right)\left(f(x^{(t+1)},y^{(t)})-f(x_\OPT, y_\OPT)\right).
\]
\end{lemma}

Combining this lemma with the structure of our regularizer~\eqref{def:reg-sm}, we obtain the following guarantees, showing \Cref{alg:altmin} finds a $\frac \veps 2$-approximate solution to the proximal oracle. 

\begin{restatable}[Convergence of~\Cref{alg:altmin}]{corollary}{innerloop}\label{coro:altmin-conv}
	Let $\delta,\veps\in(0,1)$, $\rho\ge 1$. Suppose we are given $\gamma\in\zset_* = \xset_*\times \yset_*$ with $\max(\norm{\gamma\x}_\infty,\norm{\gamma\y}_1)\le B$, and define the proximal subproblem solution 
\[x_\OPT, y_\OPT = \argmin_{x\in\Delta^m}\argmin_{y\in[0,1]^n} f(x,y)\defeq\langle \gamma^x, x\rangle+\langle\gamma^y,y\rangle+\theta r(x,y)~~\text{for some}~\theta>0.\] If the Hessian condition in~\Cref{lem:progress-altmin} holds with a constant $\kappa>0$, and all simplex iterates $x$ of \Cref{alg:altmin} satisfy $x\ge \delta$ elementwise, then the algorithm finds a $\frac \veps 2$-approximate solution to the proximal oracle within $T = O\Par{\log\Par{\frac{\rho(B+mn\theta)^2}{\delta\veps\theta}}}$ iterations.
\end{restatable}

\subsection{Helper lemmas}\label{ssec:helper}

Before providing our full method and analysis, here we list a few helper lemmas, which we rely on heavily in our later development. The first characterizes a useful property of $r_{\mu,\epsilon}$, showing that its Hessian is locally well-approximated by a diagonal matrix, which induces appropriate local norms for the blocks $x\in\xset$, $y\in\yset$.  We use this to prove the ``strong monotonicity'' (\Cref{lem:alphabetasm}) and ``relative Lipschitzness'' (\Cref{lem:rel-lip}) bounds required in~\Cref{ssec:framework-alg}. 

\begin{restatable}[Bounds on regularizer]{lemma}{restatercab}\label{lem:reg-convex-and-bound}
	Suppose $\ma \in \R^{m \times n}$ has $\norm{\ma}_\infty \le 1$. For any $z=(x,y)\in\Delta^m\times [0,1]^n$, $r=\rme$ defined as in~\eqref{def:reg-sm}, and $\bar{x} \in \R_{> 0}^m$, $\inprod{x}{\ma y} \le \norm{x}_{\diag{\frac{1}{\bar{x}}}} \norm{y}_{\diag{\Abs{\ma}^\top \bar{x}}}.$
Further, if $\rho\ge 3$, the matrix 
	\begin{equation}\label{eq:def-D}
	\md(x) \defeq \begin{pmatrix} \frac{\rho}{2} \diag{\frac{1}{x}} & \mathbf{0} \\ \mathbf{0} & \frac{1}{\rho} \diag{\Abs{\ma}^\top x}\end{pmatrix}
	\end{equation}
	satisfies the following relationship with the Hessian matrix of $r(z)$:
	\begin{equation}\label{eq:reg-convex-and-bound}
		\md(x)\preceq \nabla^2r(z) = \begin{pmatrix} \rho \cdot \diag{\frac{1}{x}} & \frac{2}{\rho}\ma \diag{y} \\ \frac{2}{\rho}\diag{y} \ma^\top  & \frac{2}{\rho} \diag{\Abs{\ma}^\top x} \end{pmatrix} \preceq 4\md(x).
	\end{equation}
\end{restatable}

We also introduce the following notion of a padding oracle (cf.\ Definition 2 of~\cite{CJST20}), which helps us control the multiplicative stability of iterates when running our algorithm. 

\begin{definition}\label{def:pad}
	Given $\delta>0$, and any $\bz = (\bx,y)\in\Delta^m\times [0,1]^n$, a padding oracle $\mathcal{O}_\delta$ returns $z = (x, y)$ by setting $\hx_i = \max(\bx_i, \delta)$ coordinate-wise and letting $x = \frac{\hx}{\norm{\hx}_1}$.
\end{definition}

This padding oracle has two merits which we exploit. First, the error incurred due to padding is small proportional to the padding size $\delta$, which finds usage in proving the correctness of our main algorithm,~\Cref{alg:sherman} (see~\Cref{prop:outerloopproof-sm-pad}). 

\begin{restatable}[Error of padding, cf.\ Lemma 6,~\cite{CJST20}]{lemma}{lempaderror}\label{lem:padsuffices}
For $\delta>0$ and $\bz = (\bx,y)\in\Delta^m\times[0,1]^n$ let $z = (x,y)\in\Delta^m\times[0,1]^n$ where $x = \mathcal{O}_\delta(\bx)$ (Definition~\ref{def:pad}), then for $r$ in~\eqref{def:reg-sm}, and any $w\in\zset = \Delta^m\times[0,1]^n$, $
V^r_{z}(w)-V^r_{\bz}(w)\le \left(\rho+\frac{8}{\rho}\right)m\delta$.
\end{restatable}

Second, padding ensures that the iterates of our algorithm satisfy $x = \Omega(\delta)$ entrywise, i.e.\ no entries of our simplex iterates $x$ are too small. This helps ensure the stability of iterates throughout one call of~\Cref{alg:altmin}, formally through the next lemma. 

\begin{restatable}[Iterate stability in~\Cref{alg:altmin}]{lemma}{lemstability}\label{lem:stability}
	Suppose $\eps\le 1$, $\rho\ge 6$, and $\alpha\ge \frac{36}{\rho}(\mu\log\frac 4 \delta + 3\Cmax)$. Let $(x_k, y_k)$ denote blocks of $z_k$, the $k^{\text{th}}$ iterate of \Cref{alg:outerloop}. In any iteration $k$  of~\Cref{alg:outerloop}, calling~\Cref{alg:altmin} to implement~\Cref{line:outer-grad-sm}, if $x_{k-1}\ge\frac \delta 2$ entrywise, $x^{(t+1)}\in x_{k-1}\cdot\Brack{\exp\Par{-\frac 1 9},\exp\Par{\frac 1 9}},\text{ for all } t\in[T].$
	Calling~\Cref{alg:altmin} to implement~\Cref{line:outer-extragrad-sm}, if $x_{k-1/2}\ge\frac \delta 4$ entrywise, $x^{(t+1)}\in x_{k-1}^{\frac{\alpha}{\alpha+\nu}}\circ x_{k-1/2}^\frac{\nu}{\alpha+\nu}\cdot\Brack{\exp\Par{-\frac 1 9},\exp\Par{\frac 1 9}}$ for all $t\in[T]$.
\end{restatable}

\subsection{Regularized box-simplex solver and its guarantees}\label{ssec:framework-alg}

We give our full high-accuracy regularized box-simplex game solver as~\Cref{alg:sherman}, which combines~\Cref{alg:outerloop},~\Cref{alg:altmin}, and a padding step to ensure stability. \notarxiv{For space considerations, we defer its statement to Appendix~\ref{ssec:shermanalg}.}

\arxiv{
	
	\SetKwProg{Fn}{function}{}{}
	
	\begin{algorithm}[ht!]
		\DontPrintSemicolon
		\KwInput{ $\ma \in \R^{m \times n}$, $c \in \R^m$, $b \in \R^n$, accuracy $\veps\in(m^{-10},1)$, $72\eps\le \mu\le  1$}
		\KwOutput{Approximate solution pair $(x, y)$ to \eqref{eq:main-reg-apprx}}
		\textbf{Global:} $\delta \gets  \tfrac{\eps\veps^2}{m^2}$, $\rho\gets \sqrt{\tfrac{2\mu}{\eps}}$, $\nu\gets \tfrac{1}{2}\sqrt{\tfrac{\mu\eps}{2}}$, $\alpha \gets 18\Cmax+32\sqrt{\frac{\mu\eps}{2}} \log \frac 4 \delta$ \;
		\textbf{Global:} $T \gets O\Par{\log \frac{mn\Bmax \alpha\rho}{\delta\veps}}$, $K \gets O\Par{\frac \alpha \nu \log\Par{\frac{\nu\log m}{\veps}}}$ for appropriate constants\;
		$(x_0,y_0)\gets (\tfrac{1}{m}\cdot\mathbf{1}_m, \mathbf{0}_n)$\;
		\For{$k =1$ to $K$}{
			$(\gamma\x, \gamma\y) \gets \GradBS(x_{k-1},y_{k-1}, x_{k-1},y_{k-1}, 0)$\;\label{line:altmin-one:start}
			$(x_{k-\frac{1}{2}},y_{k-\frac{1}{2}})\gets \AltminBS(\gamma\x,\gamma\y,\alpha,x_{k-1},y_{k-1})$\;\label{line:altmin-one:end}
			$(\gamma\x, \gamma\y) \gets \GradBS(x_{k-\frac{1}{2}},y_{k-\frac{1}{2}}, x_{k-1}, y_{k-1}, \nu)$\;\label{line:altmin-two:start}
			$(x^{(T+1)},y^{(T)})\gets \AltminBS(\gamma\x,\gamma\y,\alpha+\nu,x_{k-\frac{1}{2}},y_{k-\frac{1}{2}})$\;\label{line:altmin-two:end}
			$x_{k}\gets \frac{1}{\norm{\max\Par{x^{(T+1)},\delta}}_1}\cdot  \max\Par{x^{(T+1)},\delta}$, $y_{k}\gets y^{(T)}$\label{line:return-two}\Comment*{Implement padding $\mathcal{O}_\delta(x^{(T+1)})$}
		}
		\Fn{$\GradBS(x,y,x_0,y_0, \Theta)$}{
			$g\x \gets  \ma y + c + \mu(\1+\log(x)) - \frac{\eps}{2} \Abs{\ma} (y^2)$\;
			$g\y\gets -\ma^\top x + b + \eps\diag{y}|\ma|^\top x$ \;
			$g\x_r \gets -\alpha \rho (1+\log x_0) - \frac{\alpha}{\rho}|\ma| y^2_0 -  \Theta \rho (1+\log x) - \frac{\Theta}{\rho}|\ma| y^2$ \;
			$g\y_r \gets - \frac{2\alpha}{\rho} \diag {y_0}|\ma|^\top x_0 - \frac{2\Theta}{\rho} \diag {y}|\ma|^\top x)$ \;
			\Return $(g\x + g\x_r ,g\y + g\y_r)$	
		}
		\Fn(\Comment*[h]{Implement approximate proximal oracle via $\AltMin$}){$\AltminBS(\gamma\x,\gamma\y, \theta,x^{(0)},y^{(0)})$}{\For{$0 \le t \le  T$}{
				$x^{(t+1)}\gets \frac{1}{\norm{\exp\Par{-\frac{1}{\theta\rho}\gamma\x-\frac{1}{\rho^2}|\ma|\Par{y^{(t)}}^2}}_1}\cdot \exp\Par{-\frac{1}{\theta\rho}\gamma\x-\frac{1}{\rho^2}|\ma|\Par{y^{(t)}}^2}$ \;
				$y^{(t+1)} \gets \mathrm{med}\left(0,1, -\frac{\rho}{2\theta}\cdot\frac{\gamma\y}{|\ma|^\top x^{(t+1)}}\right)$ \;
			}	
			\Return $(x^{(T+1)},y^{(T)})$	
		}
		\caption{$\RegBoxSimp(\ma, b, c, \epsilon,\mu,\veps)$}\label{alg:sherman}
	\end{algorithm}
}

In order to analyze the convergence of Algorithm~\ref{alg:sherman}, we begin by observing that the operator in \eqref{def:grad-op} satisfies strong monotonicity with respect to our regularizer \eqref{def:reg-sm}.

\begin{restatable}[Strong monotonicity]{lemma}{lemsm}\label{lem:alphabetasm}
Let $\mu \ge \frac \eps 2$ and $\rho \defeq \sqrt{\tfrac{2\mu}{\eps}}$. The gradient operator $\gme$~\eqref{def:grad-op} is $\nu \defeq \tfrac{1}{2}\sqrt{\tfrac{\mu\eps}{2}}$-strongly monotone (see \eqref{def-strong-monotone}) with respect to $\rme$ defined in \eqref{def:reg-sm}.
\end{restatable}

Next, we show iterate stability through each loop of alternating minimization (i.e.\ from~\Cref{line:altmin-one:start} to~\Cref{line:altmin-one:end}, and~\Cref{line:altmin-two:start} to~\Cref{line:altmin-two:end} respectively), via~\Cref{lem:stability}.

\begin{restatable}[Iterate stability in~\Cref{alg:sherman}]{corollary}{corostability}\label{coro:stability}
	Assume the same parameter bounds as Lemma~\ref{lem:stability}, and that $\delta \in (0, m^{-1})$. In the $k^{\text{th}}$ outer loop of~\Cref{alg:sherman}, $x_{k-1}\ge\frac \delta 2$ entrywise. Further, for all iterates $x^{(t+1)}$ computed in~\Cref{line:altmin-one:start} to~\Cref{line:altmin-one:end} and $x_\OPT$ as defined in~\eqref{eq:subproblem} with $\theta = \alpha$, $\half x_{k - 1} \le x^{{(t + 1)}}, x_\OPT \le 2x_{k - 1}$, and $x^{(t+1)}, x_\OPT\ge \frac{\delta}{4}$, entrywise. Similarly, for all iterates $x^{(t+1)}$ computed in~\Cref{line:altmin-two:start} to~\Cref{line:altmin-two:end} and $x_\OPT$ as defined in~\eqref{eq:subproblem} with $\theta = \alpha+\nu$, $	\frac{1}{2} x_{k-1/2}\le x^{(t+1)}, x_\OPT \le 2x_{k-1/2}~\text{and}~x^{(t+1)}, x_\OPT\ge \frac{\delta}{4},~\text{entrywise}$.
\end{restatable}

Under iterate stability, our next step is to prove that our operator $\gme$ is relatively Lipschitz with respect to our regularizer $\rme$ (as defined in~\eqref{def-rel-lip}). 

\begin{restatable}[Relative Lipschitzness]{lemma}{lemrellip}\label{lem:rel-lip}
Assume the same parameter bounds as in Lemma~\ref{lem:stability}. In the $k^{\text{th}}$ outer loop of~\Cref{alg:sherman}, let $\bz_{k}\gets (x^{(T+1)}, y^{(T)})$ from \Cref{line:altmin-two:end} be $z_k$ before the padding operation. Then, $x_{k-1/2}, \bx_{k}\in[\tfrac{1}{2}x_{k-1},2x_{k-1}]$ elementwise and 
\begin{equation*}
\inner{g(z_{k-1/2})-g(z_{k-1})}{z_{k-1/2}-\bz_k}\le \alpha \left(V_{z_{k-1}}(z_{k-1/2})+V_{z_{k-1/2}}(\bz_{k})\right)~~\text{for}~~\alpha = 4+32\sqrt{\frac{\mu\eps}{2}}.
\end{equation*}
\end{restatable}

Next, we give a convergence guarantee on the inner loops (from~\Cref{line:altmin-one:start} to~\Cref{line:altmin-one:end}, and~\Cref{line:altmin-two:start} to~\Cref{line:altmin-two:end}) in \Cref{alg:sherman}, as an immediate consequence of~\Cref{coro:altmin-conv}.

\begin{restatable}[Inner loop convergence in~\Cref{alg:sherman}]{corollary}{coroaltmin}\label{coro:altmin-conv-sherman}
Assume the same parameter bounds as in Lemma~\ref{lem:stability}. For $\gamma$ defined in~\Cref{line:altmin-one:start}, suppose for an appropriate constant $T = \Omega\Par{\log\frac{mn \Bmax \alpha\rho}{\delta\veps}}$. Then, for all $k$ iterate $z_{k-1/2}=(x_{k-1/2}, y_{k-1/2})$ of~\Cref{line:altmin-one:end} satisfies
	\[
	\inprod{\nabla g(z_{k-1})+\alpha \nabla V_{z_{k-1}}(z_{k-1/2})}{z_{k-1/2}-w}\le \frac{\nu\veps}{4},\text{ for all }w\in\zset.
	\] 
	Similarly, for $\gamma$ defined in~\Cref{line:altmin-two:start}, iterate $\bz_{k}=(x_{(T+1)}, y_{(T)})$ of~\Cref{line:altmin-two:end} satisfies
	\[
	\inprod{\nabla g(z_{k-1})+\alpha \nabla V_{z_{k-1}}(\bz_{k})+\nu \nabla V_{z_{k-1/2}}(\bz_{k})}{\bz_{k}-w}\le \frac{\nu\veps}{4},\text{ for all }w\in\zset.
	\]
\end{restatable}

We now analyze the progress made by each outer loop of~\Cref{alg:sherman}. The proof is very similar to that of~\Cref{prop:outerloopproof-sm}; the only difference is controlling the extra error incurred in the padding step of~\Cref{line:return-two}, which we bound via~\Cref{lem:padsuffices}.

\begin{restatable}[Convergence of~\Cref{alg:sherman}]{proposition}{propouterpad}%
\label{prop:outerloopproof-sm-pad}
Assume the same parameter bounds as in Lemma~\ref{lem:stability}, and that $\delta\le \frac{\veps}{4\rho\alpha m}$. \Cref{alg:sherman} returns $z_K$ satisfying $V^r_{z_K}(z^\star)\le\frac{3\veps}{\nu}$, 
letting (for an appropriate constant) $K = \Omega(\frac \alpha \nu \log(\frac{\nu \log m}{\veps})).$
\end{restatable}

We are now ready to prove the main theorem of this section, which gives a complete convergence guarantee of \Cref{alg:sherman} by combining our previous claims.

\begin{restatable}[Regularized box-simplex solver]{theorem}{thmsherman}\label{thm:sherman}
Given regularized box-simplex game \eqref{eq:main-reg-apprx} with $72 \eps\le \mu\le 1$ and optimizer $(x^\star, y^\star)$, and letting $\veps\in(m^{-10},1)$, $\RegBoxSimp$~(\Cref{alg:sherman}) returns $x^K$ satisfying $\norm{x^K - x^\star}_1 \le \frac{\veps}{C_{\max}\log^2 m}$ and $\max_{y\in[0,1]^n}f_{\mu,\eps}(x^K,y) - f_{\mu,\eps}(x^\star,y^\star)\le \veps$.
The total runtime of the algorithm is $O(\nnz(\ma) \cdot (\frac{\Cmax}{\sqrt{\mu\eps}} + \log(\frac{m}{\sigma\eps}))\cdot \log(\frac{C_{\max}\log m} \veps) \log(\frac{mn\Bmax} \veps))$.
\end{restatable}

As a corollary, we obtain an approximate solver for regularized box-simplex games in the following form (which in particular does not include a quadratic regularizer):
\begin{equation}\label{eq:main-reg}
\begin{aligned}
\min_{x\in\Delta^m}\max_{y\in[0,1]^n} f_{\mu}(x,y)  = y^
\top \ma^\top x + c^\top x -b^\top y + \mu H(x),
~~\text{where}~~ H(x)  \defeq \sum_{i\in[m]}x_i\log x_i.
\end{aligned}
\end{equation}

\begin{restatable}[Half-regularized approximate solver]{corollary}{corosherman}\label{coro:sherman}
Given regularized box-simplex game \eqref{eq:main-reg} with regularization parameters $72 \epsilon\le \mu\le 1$ and optimizer $(x^\star, y^\star)$, and letting $\eps\in(m^{-10},1)$, \Cref{alg:sherman} with $\veps \gets \frac \eps 2$ returns $x^K$ satisfying $\max_{y\in[0,1]^n}f_\mu(x^K,y) - f_\mu(x^\star,y^\star)\le \eps.$	
The total runtime of the algorithm is $O\Par{\nnz(\ma) \cdot  (\frac{\Cmax}{\sqrt{\mu\eps}} + \log\Par{\frac{m}{\eps}}) \cdot \log\Par{\frac{C_{\max}\log m} \eps} \log\Par{\frac{mn\Bmax} \eps}}$.
\end{restatable}

\arxiv{
\section*{Acknowledgements}

We thank anonymous reviewers for their feedback, Amin Saberi and David Wajc for helpful conversations, Jason Altschuler for providing a reference for the unaccelerated convergence rate of Sinkhorn's algorithm (in the original submission, we claimed no such rate had been stated previously), and Monika Henzinger and Thatchaphol Saranurak for helpful information regarding prior work.

AS was supported in part by a Microsoft Research Faculty Fellowship, NSF CAREER Award CCF-1844855, NSF Grant CCF-1955039, a PayPal research award, and a Sloan Research Fellowship. KT was supported in part by a Google Ph.D. Fellowship, a Simons-Berkeley VMware Research Fellowship, a Microsoft Research Faculty Fellowship, NSF CAREER Award CCF-1844855, NSF Grant CCF-1955039, and a PayPal research award. YJ was supported on a Stanford Graduate Fellowship and the Dantzig-Lieberman Operations Research Fellowship. 
}

\notarxiv{
\newpage
}
\arxiv{
\bibliographystyle{alpha}
} %
\newcommand{\etalchar}[1]{$^{#1}$}

\newpage

\appendix

\notarxiv{
\section{Discussion of a recent advancement~\cite{chen2022maximum}}\label{app:maxflow}

Subsequent to the original submission of this paper, a breakthrough result of  \cite{chen2022maximum} provided an algorithm which computes high-accuracy solutions to a variety of graph-based optimization objectives in almost-linear time. One of the many applications of \cite{chen2022maximum} is faster algorithms for computing Sinkhorn distances. Using the algorithm designed in~\cite{chen2022maximum}, in place of~\cite{Cohen17}, yields  a speedup of Theorem~\ref{lem:ms-sinkhorn}, removing the polynomial dependence on $\mu^{-1}$ at the cost of an overhead of $m^{o(1)}$.

\begin{theorem}[informal, see~\Cref{lem:dec-solver-maxflow}]
	For $\eps = \Omega(m^{-3})$ in~\eqref{eq:sinkhornintro} corresponding to a bipartite graph with $m$ edges, there is an algorithm that obtains an $\eps$-approximate minimizer to~\eqref{eq:sinkhornintro} in time $m^{1+o(1)}$. 
\end{theorem}

Further, we show how to use \cite{chen2022maximum} to obtain improved running times for DDBM. By plugging in this Sinkhorn distance computation algorithm into our DDBM framework (and showing it is compatible with the form of our subproblems), we obtain an improved (randomized) DDBM solver in terms of the $\eps^{-1}$ dependence at the cost of a $m^{o(1)}$ overhead. 

\begin{theorem}[informal, see Theorem~\ref{thm:main-dec-sinkhorn-maxflow}]
	Let $G = (V, E)$ be bipartite, $|V|=n$, $|E|=m$, and $\eps \in [\Omega(m^{-3},1)$. There is a randomized algorithm with success probability $1-n^{-\Omega(1)}$  maintaining an $\eps$-approximate maximum matching in an (adaptive) decremental stream running in time $m^{1+o(1)}\epsilon^{-2}$.
\end{theorem}

We believe this result highlights the versatility and utility of our framework for solving DDBM. Given the varied technical machinery involved (and differing runtimes and use of randomness) our Sinkhorn algorithms, i.e.\ our new regularized box-simplex solver (see~\Cref{sec:framework}), \cite{Cohen17}, and \cite{chen2022maximum}, we provide statements of each solver instantiated by these different machinery in~\Cref{ssec:framework-alg}. However, due to the subsequent nature of \cite{chen2022maximum} with respect to our work, and because DDBM running times using  \cite{Cohen17} are no better than those achieved in our paper using our regularized box-simplex solver, the proofs of our DDBM solver based on~\cite{chen2022maximum} are deferred to~\Cref{ssec:dynamic-redx-maxflow}.
}

\section{Proofs for~\Cref{sec:dynamic}}\label{app:dynamic}

\subsection{Proofs for~\Cref{sec:dec_framework}}\label{app:dec_framework}

\propdecmatch*

\begin{proof}

Throughout this proof, let the edge sets recomputed by \Cref{line:recompute} be denoted $E_1, E_2, \ldots, E_K$ and assume the termination condition~\Cref{line:terminate-dec-outer} breaks the loop for $E_{K+1}$, where $E_0$ is the original edge set of the graph. Before proving the two claims, we first observe that for all $k \in [K]$, $\MCM(E_k) \ge \frac M 8$. The base case of $k = 0$ clearly holds because of the approximation guarantee of the greedy algorithm, which yields $M \ge \half \MCM(E)$. Now, for any $k \in [K]$, let $M_{\textup{est}}^k$ be the greedily computed estimate of $\MCM(E_k)$. We conclude by 
\[\MCM(E_k) \ge M_{\textup{est}}^k > \frac 1 4 M \ge \frac 1 8 \MCM(E).\]

\textbf{Matching approximation.} For any $0 \le k \le K$, by item $1$ of the CRO definition applied to $E_{k}$, we have $(1 + \frac \epsilon 8) \MCM(E_{k})\ge -\nu^{E_{k}} \ge (1 - \frac \epsilon 8) \MCM(E_{k})$, and hence combining with \eqref{eq:cond-apprx-function} and~\eqref{eq:cond-round-guarantee} implies
\begin{align}
	-f_{M,E_{k}}(\hat{x}^{E_{k}})& \ge -\Par{1-\frac{\epsilon}{8}}\nu^{E_{k}} \ge \Par{1-\frac{\epsilon}{8}}^2\frac{1}{4}M \ge\frac{3}{16}M,~\label{eq:bound-xhat}\\
	-f_{M,E_{k}}(\tilde{x}^{E_{k}})& \ge -\Par{1-\frac{\epsilon}{8}}\nu^{E_{k}} \ge \Par{1-\frac{\epsilon}{8}}^2\frac{1}{4}M \ge\frac{3}{16}M.~\label{eq:bound-xtilde}
\end{align}
By the definition of the CRO in~\eqref{eq:cond-canonical-two} and feasiblity of the matching $8M\tilde{x}^{E_k}$ (which follows by the assumed guarantees on $\Round$), we have 
\begin{align}
8M\norm{\tilde{x}^{E_k}}_1 \ge	-f_{M,E_k}\Par{\tilde{x}^{E_k}}-\frac{\epsilon}{128}M\ge\frac{5}{32}M 
\implies\norm{\tilde{x}^{E_k}}_1 \ge \frac{5}{8\cdot 32}.\label{eq:lb-xtilde}
\end{align}
Also, following the above calculation and~\eqref{eq:bound-xtilde} we have that
\begin{align}
8M\norm{\tilde{x}^{E_k}}_1 & \ge  	-f_{M,E_k}\Par{\tilde{x}^{E_k}}-\frac{\eps}{128}M\nonumber \ge\Par{1-\frac{\epsilon}{8}}^2\MCM(E_k)-\frac{\epsilon}{32}M\ge \Par{1-\frac{\eps}{2}}\MCM(E_k)\label{eq:xtile-approx-1}
\end{align}
where for the last inequality we use $\MCM(E_k)\ge \frac{1}{4}M$, as shown in the first part of this proof. Consequently, considering any $E$ between recomputations satisfying $E_{k+1}\subset E\subseteq E_k$ such that $E=E_k\setminus E_\del$, we have by the condition on~\Cref{line:inner-dec-start} that
\begin{align*}
	8M\norm{\tilde{x}^{E_k}_{E_k\setminus E_\del}}_1 & = 8M\norm{\tilde{x}^{E_k}}_1 - 8M\norm{\tilde{x}^{E_k}_{E_\del}}_1\ge \Par{1-\frac{\epsilon}{8}}8M\norm{\tilde{x}^{E_k}}_1\\
	& \ge \Par{1-\frac{\epsilon}{8}}\Par{1-\frac{\epsilon}{2}}\MCM(E_k)\ge \Par{1-\epsilon}\MCM(E_k\setminus E_\del),
\end{align*}
which combined with feasibility of $8M\tx^{E_k}$ (and hence, feasibility of any restriction) implies that $8M\tilde{x}^{E_k}_{E_k\setminus E_\del}$ is always an $\epsilon$-approximate matching of current graph.

\textbf{Multiplicative matching decrease.} When the algorithm terminates after $K+1$ loops, using the termination condition, and the guarantee of $\Mest$, we have
\[\MCM({E_{K+1}})\le 2\Mest \le \frac{1}{2}M\le \frac{1}{2}\MCM(E_0).\]
Hence, the $\MCM$ value decreases by a factor of at least $2$.

\textbf{Iteration bound.} We next show that the algorithm will stop after $K+1$ loops (from~\Cref{line:inner-dec-start} to~\Cref{line:inner-dec-end}) for bounded $K$. For some loop $0 \le k \le K-1 $, letting $E_{k + 1} = E_k \setminus E_\del$,
\begin{equation}\label{eq:dec-progress-app}
	\begin{aligned}
& f_{M,E_{k+1}}\Par{x^{E_{k+1}}} - f_{M,E_{k}}\Par{x^{E_k}} \stackrel{(i)}{\ge} \beta V^H_{x^{E_k}}\Par{x^{E_{k+1}}}\\
& \quad \stackrel{(ii)}{\ge} \beta\sum_{i\in E_\del} \left(
[x^{{E_{k+1}}}]_i\log [x^{{E_{k+1}}}]_i - [x^{E_k}]_i\log [x^{E_k}]_i - \Par{1+\log [x^{E_k}]_i}\cdot \Par{[x^{{E_{k+1}}}]_i-[x^{E_k}]_i}
\right)\\
& \quad\stackrel{(iii)}{=} \beta \sum_{i\in E_\del}[x^{E_k}]_i \stackrel{(iv)}{\ge} \beta \left(\norm{\hat{x}^{E_k}_{E_\del}}_1 - \norm{\hat{x}^{E_k} - x^{E_k}}_1\right)\\
& \quad\stackrel{(v)}{\ge} \beta \left(\norm{\tilde{x}^{E_k}_{E_\del}}_1 - \norm{\hat{x}^{E_k} - x^{E_k}}_1\right),
\end{aligned}
\end{equation}
where we used $(i)$ the assumption~\eqref{eq:cond-canonical-one}, $(ii)$ convexity of the scalar function $c \log c$ allowing us to restrict the divergence to a subset, $(iii)$ the property that $[x^{E_{k+1}}]_i = 0$ on all deleted edges by Item~\ref{item:def12} in Definition~\ref{def:canonical}, $(iv)$ the triangle inequality, and $(v)$ the monotonicity property~\eqref{eq:cond-round-guarantee-l1}.
 
We now proceed to bound $\|\tilde{x}^{E_k}_{E_\del}\|_1$ and $\norm{\hat{x}^{E_k} - x^{E_k}}_1$. By the termination condition on Line~\ref{line:inner-dec-start}, $\|\tilde{x}^{E_k}_{E_\del}\|_1\ge \tfrac{\eps}{8}\norm{\tilde{x}^{E_k}}_1$. Moreover,
\begin{align*}
\norm{\hat{x}^{E_k} - x^{E_k}}_1\le \frac{\eps}{1000}\le 	\frac{\eps}{16}\norm{\tilde{x}^{E_k}}_1,
\end{align*}
where we used the $\ell_1$ bound in~\eqref{eq:cond-apprx-l1} and the lower bound  on $\norm{\tx^{E_k}}_1$ implied by~\eqref{eq:lb-xtilde}. 

Plugging these two bounds back in~\eqref{eq:dec-progress-app}, and once again using \eqref{eq:lb-xtilde}, we thus have
\[
f_{M,E_{k+1}}\Par{x^{E_{k+1}}} - f_{M,E_{k}}\Par{x^{E_k}}\ge \frac{\beta\eps}{16}\norm{\tilde{x}^{E_k}}_1 = \Omega(\beta\eps).
\]
At the beginning of the algorithm, recall by $(1 + \frac 1 8) \MCM(E_{0})\ge -\nu^{E_{k}} = -f_{M, E_0}(x^{E_0})$ due to~\Cref{def:canonical} we have $-f_{M, E_0}(x^{E_0})\le 2.25M$. Similarly, due to ~\Cref{def:canonical} we also have for $E_K$, it holds that $ -\nu^{E_{K}} \ge (1 - \frac 1 8) \MCM(E_{K})\ge \frac{7}{32}M$. Thus, this shows $K\le O\Par{\frac{M}{\beta\epsilon}}$ and thus the algorithm must terminate within $ O(\tfrac{M}{\beta \eps})$ loops. This implies the runtime bound after combining with the cost of the greedy approximation.
\end{proof}

\subsection{DDBM via regularized box-simplex games}\label{ssec:dynamic-redx-bs}

Throughout this section, let $\eps \in (0, 1)$, and assume that the DDBM problem is initialized with bipartite $G = (V, E_0)$ with unsigned incidence matrix $\mb \in \{0, 1\}^{E \times V}$; we denote $n \defeq |V|$ and $m \defeq |E_0|$.  For $E \subseteq E_0$ and $M$, an $8$-approximation to $\MCM(E)$, we consider the following regularized box-simplex game approximating the matching problem:
\begin{equation}\label{eq:def-matching-bs}
\begin{aligned}
\min_{(x, \xi)\in  \Delta^{E+1}}\max_{y\in[0,1]^{V}} f_{M,E}(x,\xi,y) &\defeq  -\1_E^\top(8Mx) - y^\top \Par{8M\mb^\top x-\1} +\gamma\x H(x, \xi)+\gamma\y \Par{y^2}^\top \mb^\top x,\\
&\text{where}~\gamma\x = \frac{\eps M}{256\log(m)},~\gamma\y = \frac{\eps M}{256}.
\end{aligned}
\end{equation}

We recall a rounding procedure from~\cite{AssadiJJST22}; we restate its guarantees below for completeness.

\begin{lemma}[Rounding]\label{lem:remove}
	Given bipartite $(V,E)$, $\Remove$ (Algorithm 4, \cite{AssadiJJST22}) takes $\ell \in\R^E_{\ge0}$, runs in $O(|E|)$ time, and outputs $\tilde{\ell}$ satisfying $\tilde{\ell}\le \ell$ entrywise, $\mb^\top\tilde{\ell}\le \1$, and  
	\[\norm{\tilde{\ell}}_1\ge \norm{\ell}_1-\sum_{i\in V}\Par{[\mb^\top \ell-1]_i}_{+}.\]
\end{lemma}

Next, we show that the set of $f_{M,E}(x) =\min_{\xi \mid (x, \xi)\in  \Delta^{E+1}}\max_{y\in[0,1]^{V}} f_{M,E}(x, \xi,y)$ forms a family of $(\eps,\gamma\x)$-CROs.

\begin{restatable}[CROs from regularized box-simplex games]{lemma}{lemcanonicalbs}\label{lem:canonical-bs}
Define (overloading notation)
\[f_{M,E}(x) \defeq \min_{\xi \mid (x, \xi)\in  \Delta^{E+1}}\max_{y\in[0,1]^{V}} f_{M,E}(x, \xi,y).\]
Then, the family of $f_{M,E}(x)$ is a family of $(\eps,\gamma\x)$-CROs.
\end{restatable}

\begin{proof} We prove each property in turn. Item~\ref{item:def12} (restricting coordinates to zero) is immediate from the problem definition, so we focus on showing the other three.
	
\textbf{Multiplicative value approximation (Item~\ref{item:def11}).}
By assumption, $M$ is an $8$-multiplicative approximation to $\MCM(E)$, so there is a maximum matching $8Mx_\star^E$ such that $\norm{x_\star^E}_1\le 1$. For such $x_\star^E$ and $\xi \defeq 1-\norm{x_\star^E}_1$, and because the entropic and quadratic regularization terms are nonpositive, 
\begin{align*}
\nu^E & \le \max_{y\in[0,1]^{V}} f_{M,E}(x_\star^E,\xi,y) \\
&\le \max_{y \in [0,1]^V} -\1_E^\top(8Mx_\star^E) - y^\top\Par{8M\mb^\top x_\star^E - \1} = -\MCM(E).
\end{align*}
The last equality holds because we assumed $8Mx_\star^E$ was a maximum matching. 
For the other direction, we proceed by contradiction. Suppose we have a feasible $(x,\xi)$ with
\[\max_{y}f_{M,E}(x,\xi,y)
< -\Par{1+\frac{\eps}{4}}\MCM(E).\]
Applying $\Remove$ on the approximate matching $\ell \defeq 8Mx$ yields $\tilde{\ell} \defeq 8M\tilde{x}$ such that 
\begin{align*}
-\MCM(E)\le -\norm{8M\tilde{x}}_1 & \le -\norm{8Mx}_1+\norm{8M\mb^\top x-\1}_1\\
& \stackrel{(i)}{\le} \max_{y\in[0,1]^{V}}f_{M,E}(x,\xi,y) + \frac{\eps M}{128}\\
& \stackrel{(ii)}{<}-\Par{1+\frac{\eps}{4}}\MCM(E) + \frac{\eps}{16}\MCM(E) < -\MCM(E),
\end{align*}
where we used $(i)$ the definition of $f_{M,E}$ and bounds on $\gamma\x$, $\gamma\y$ leading to an additive range of at most $\frac{\eps M}{128}$, and $(ii)$ that $M\le 8\MCM(E)$ by assumption, leading to a contradiction. Hence, combining these bounds yields Item~\ref{item:def11}.

\textbf{Relative strong convexity (Item~\ref{item:def13}).} For simplicity, we drop the subscript $M$. Let $(x^E,\xi^E,y^E)$ and $(x^{E'},\xi^{E'},y^{E'})$ be the unique optimal solutions of $f_{E}$ and $f_{E'}$ respectively. By the first order optimality condition of $(x^E,\xi^E)$, we have
\begin{align*}
f_{E'}(x^{E'},\xi^{E'},y^{E}) - f_{E}(x^E,\xi^E,y^E) &\ge \int_{0}^1 (1-\alpha) v^\top\nabla^2_{(x,\xi)(x,\xi)}f(x_\alpha, \xi_\alpha,y^E) v d\alpha,
\end{align*}
where we let $v \defeq (x^{E'}-x^{E},\xi^{E'}-\xi^{E})$, $x_\alpha \defeq \alpha x^{E'} + (1 - \alpha) x^E$, and $\xi_\alpha \defeq \alpha \xi^{E'} + (1 - \alpha)\xi^E$. By joint convexity of $f$ as a function of $(x, \xi)$ (with $y$ fixed),
\begin{align*}
\int_{0}^1 (1-\alpha) v^\top\nabla^2_{(x,\xi)(x,\xi)}f(x_\alpha, \xi_\alpha,y^E) v d\alpha & \ge \int_{0}^1 (1-\alpha) v^\top\nabla^2_{xx}f(x_\alpha, \xi_\alpha,y^E) v d\alpha\\
&  = \gamma\x V^H_{x_E}(x^{E'}).
\end{align*}
Here we let $\nabla^2_{xx}f$ zero out blocks of $\nabla^2_{(x, \xi)(x,\xi) } f$ appropriately. The conclusion \eqref{eq:cond-canonical-one} follows since $f_{E'}(x^{E'},\xi^{E'},y^{E'}) \ge f_{E'}(x^{E'},\xi^{E'},y^{E})$ by optimality of $y^{E'}$.

\textbf{Additive approximation (Item~\ref{item:def14}).} It suffices to note that for any feasible matching $8Mx$, the value of $f_{M, E}$ without the added entropic and quadratic regularization terms is exactly $-8M\norm{x_E}_1$, since $8M\mb^\top x \le \1$ entrywise. Using standard bounds on the ranges of $H(x,\xi)\in[-\log (m+1),0)$, $\Par{y^2}^\top \mb^\top x\in[0,1]$, we have the additive approximation \eqref{eq:cond-canonical-two} due to the values of $\gamma\x$ and $\gamma\y$.
\end{proof}

We now apply our high-accuracy box-simplex solver $\RegBoxSimp$~(\Cref{thm:sherman}), together with our rounding procedure $\Remove$, to develop a $(\eps,\widetilde{O}(m\eps^{-1}))$-canonical solver for the CRO family induced by the regularized objectives~\eqref{eq:def-matching-bs}.

\begin{lemma}[Canonical solver for regularized box-simplex games]\label{lem:dec-solver-bs}
For $\eps = \Omega(m^{-3})$, there is an $(\eps,\mathcal{T})$-canonical solver for the family of $f_{M,E}$ defined in Lemma~\ref{lem:canonical-bs}, using $\RegBoxSimp$~(\Cref{alg:sherman}) as $\solve$ and $\Remove$ as $\Round$, achieving $\mathcal{T} = O\Par{m\eps^{-1}\log^{3}n }$.
\end{lemma}
\begin{proof}
We define $x^E$ as in \eqref{eq:nuxdef}. $\solve$ applies the procedure in Theorem~\ref{thm:sherman} to solve $\frac 1 {16M} f_{M, E}$ to additive accuracy $O(\frac{\eps^3}{\log m})$ in the stated runtime, for a sufficiently small constant. By the multiplicative approximation guarantee on $\nu^E$ and $\MCM(E)$ in \Cref{def:canonical}, and since $\MCM(E)$ is $8$-approximated by $M$, we obtain the guarantee \eqref{eq:cond-apprx-function}. The $\ell_1$ bound \eqref{eq:cond-apprx-l1} then follows since the objective is $\Omega(\frac \eps {\log m})$-strongly convex in the $x$ coordinates.

To show that $\Remove$ is a valid choice of $\Round$, the monotonicity and feasibility conditions follow immediately from \Cref{lem:remove}. Moreover, for $\ell \gets 8M\hat{x}^E$, the definition of $f_{M, E}$ and the guarantees of \Cref{lem:remove} imply
\begin{align*}\norm{\tilde{\ell}}_1 &\ge \norm{\ell}_1 - \sum_{i \in V}\Par{\Brack{\mb^\top \ell - 1}}_+ \\
&= \min_{y \in [0, 1]^V} \1_E^\top(8M\hat{x}^E) - y^\top\Par{8M\mb^\top \hat{x}^E - \1} \ge -\nu^E - O\Par{\frac{\eps^3 M}{\log m}} - \frac{\eps M}{128}.\end{align*}
Combined with our approximation guarantees between $\nu^E$, $\MCM(E)$, and $M$ via Items~\ref{item:def11} and~\ref{item:def14} of~\Cref{def:canonical}, this implies \eqref{eq:cond-round-guarantee}, showing all required properties of $\Remove$.
\end{proof}

\thmmaindecbs*
\begin{proof}
It suffices to combine Lemma~\ref{lem:canonical-bs}, Lemma~\ref{lem:dec-solver-bs}, and Corollary~\ref{coro:ddbm-reduce}.
\end{proof}

\subsection{DDBM via matrix scaling and box-constrained Newton's method}\label{ssec:dynamic-redx-high}

We follow the same notational conventions as in \Cref{ssec:dynamic-redx-bs}, and further assume $|L| \le |R|$ without loss of generality. Moreover, assume for simplicity that both $L$ and $R$ have at least one vertex which has no adjacent edges, denoted $v_L^\star$ and $v_R^\star$; this clearly will not affect any matching sizes if we add such vertices. We will further expand the graph $G = (V, E)$ to a new graph $\tG = (\tV, \tE)$, with bipartition $\tV = \tL \cup \tR$. In particular, $\tL$ consists of the original left vertices $L$, a new dummy vertex set $L_0$ such that $|L| + |L_0| = |R|$, and an extra dummy vertex $v^\dum_L$. Similarly, $\tilde{R}$ consists of the original right vertices $R$ and a new dummy vertex $v^\dum_R$. The new edge set is
\[\tE \defeq E\cup\Par{\bigcup_{v\in L\cup L_0}(v, v^\dum_R)}\cup\Par{\bigcup_{v'\in R}(v', v^\dum_L)}\cup\Par{v^\dum_L,v^\dum_R}.\]

We define $\tmb \in \{0, 1\}^{\tE \times \tV}$ such that $\tmb_{e,v}=1$ whenever $v\in e$. To instantiate our matrix scaling solver for approximating Sinkhorn distance in~\Cref{ssec:high-OT}, we construct a demand vector $d \in \R_{\ge 0}^{\tV}$ and a cost vector $c \in \R^{\tE}$ as follows.
\begin{enumerate}
	\item Set $d_v = 1$ for all $v \in L \cup L_0 \cup R$.
	\item Set $d_{v_L^\dum} = d_{v_R^\dum} = |R|$.
	\item Set $c_e = -1$ for all $e \in E$, and $c_e = 0$ for all $\tE \setminus E$.
\end{enumerate}
For a vector in $\R^{\tE}_{\ge 0}$, we refer to its restrictions to the sets $E$ and $\tE \setminus E$ by $x$ and $x^\dum$ when clear from context. It is clear that any matching $ 2|R| x \in \R^E_{\ge 0}$ on the original graph $(V, E)$ can be extended to a matching 
\[2|R| x^\tot = 2|R|\begin{pmatrix}x \\ x^\dum\end{pmatrix} \in \R^{\tE}_{\ge 0}\]
such that $\norm{x^\tot}_1 = 1$, and $2|R| \tmb^\top x^\tot = d$, by placing additional flow on the edges in $\tE \setminus E$. Specifically, we will extend any flow on $(v, v') \in E$ to put the same amount of flow on $(v', v_L^\dum)$, $(v, v_R^\dum)$, and $(v_L^\dum, v_R^\dum)$, and then for any additional demand not routed to a vertex $v \in L \cup L_0 \cup R$ by the original edges $E$, we route it arbitrarily over the additional edges $\tE$ (similarly extending this flow to $v^\dum_L$ and $v^\dum_R$). This is always feasible, as we can route to $v_L^\star$ and $v_R^\star$. The total $\ell_1$ norm of $d$ is $4|R|$ by construction, and every unit of $x^\tot$ contributes two units of demand, thus $2|R| \norm{x^\tot} = 2|R|$ and rearrangement gives the $\ell_1$ guarantee. We then consider the Sinkhorn distance objective
\begin{equation}\label{eq:def-matching-two}
\begin{aligned}
\min_{x^{\tot} = \begin{psmallmatrix}x \\ x^{\dum}\end{psmallmatrix} \in \R_{\ge 0}^{\tE} \mid \tmb^\top(2|R|x^\tot) = d}  & f^\sink_{M,E}(x^{\tot}) \defeq c^\top\Par{2|R|x^\tot} +\gamma H\Par{x^\tot}~\text{where}~\gamma = \Theta\Par{\frac{\eps M}{\log(m)}}.
\end{aligned}
\end{equation}

We show such objectives also form a family of CROs.

\begin{restatable}[CROs from Sinkhorn distances]{lemma}{lemcanonicalsink}\label{lem:canonical-sink}
	Define (overloading notation) \[f^\sink_{M,E}(x) =\min_{x^\dum\in\R_{\ge0}^{E\setminus E_0}} f^\sink_{M,E}(x, x^\dum).\] 
	Then, the family of $f^\sink_{M,E}(x)$ is a family of $(\eps, \gamma)$-CROs.
\end{restatable}

\begin{proof}
We focus on proving Item~\ref{item:def11}; Item~\ref{item:def12} is immediate from the problem definition, and Items~\ref{item:def13} and~\ref{item:def14} follow analogously to the proof in Lemma~\ref{lem:canonical-bs}.
	
\textbf{Multiplicative value approximation (Item~\ref{item:def11}).} By construction, any optimal matching $2|R| x_\star$ on the edge set $E$ can be extended to $x^\tot_\star$ such that $\norm{x^\tot_\star}_1 = 1$ and $\tmb^\top(2|R|x^\tot)=d$. By nonpositivity of entropy and the definition of $c$, we obtain
	\begin{align*}
	\nu^E \le f^\sink_{M,E}(x_\star^E) = -\MCM(E).
	\end{align*}
For the other direction, we proceed by contradiction. Suppose we have a feasible $x^\tot$ with $x \defeq x^\tot_E$, so that $f^\sink_{M,E}(x^\tot)<-\Par{1+\frac{\eps}{4}}\MCM(E)$. This implies $2|R|x$ is a feasible matching, and then an analogous argument to the one used in Lemma~\ref{lem:canonical-bs} yields the contradiction.
\end{proof}

We next construct a canonical solver for the family of $f^\sink_{M,E}$, which makes use of the matrix scaling solver in~\Cref{prop:ms-solver} of~\Cref{ssec:high-OT}. 

\begin{lemma}[Canonical solver for Sinkhorn objectives]\label{lem:dec-solver-ms}
For $\eps = \Omega(m^{-3})$, there is an $(\eps,\mathcal{T})$-canonical solver for the family of $f_{M,E}^\sink$ defined in Lemma~\eqref{lem:canonical-sink}, using the matrix scaling solver from \cite{Cohen17} as $\solve$ and truncation to $E$ as $\Round$, running in time $\tO(n^2 / \eps)$.
\end{lemma}
\begin{proof}
Let $r$ be the restriction of $\frac d {2|R|}$ to the vertices in $R \cup \{v^\dum_R\}$, and let $\ell$ be the restriction of $\frac d {2|R|}$ to the vertices in $L \cup L_0 \cup \{v^\dum_L\}$. Finally, let $\mk$ be the $(|R| + 1) \times (|R| + 1)$ matrix with $\mk_{ij} = \exp(-\frac {2|R|} \gamma c_{ij})$ for all indices corresponding to $(i, j) \in E$, and $\mk_{ij} = 0$ otherwise. Lemma 2 of \cite{Cuturi13} shows that solving the $(r, c)$-matrix scaling problem on $\mk$ computes the minimizer to \eqref{eq:def-matching-two}. 

We next apply Proposition~\ref{prop:ms-solver}, similarly to the proof of Theorem~\ref{lem:ms-sinkhorn}. Clearly $s_{\mk} \le m$, and by Lemma 10 of \cite{BlanchetJKS18}, we may bound $B$ by
\[O\Par{\log(n) + \frac{|R|}{\gamma}} = \tO\Par{\frac{|R|}{\eps M}} = \tO\Par{\frac{n^2}{\eps m}}.\]
Here, we used that $M = \Omega(\frac{m}{n})$ which follows from the fact that the minimum vertex cover size is the same as the maximum matching size. If there was a vertex cover of size $o(\frac m n)$, this would be a contradiction since each vertex can cover at most $n - 1$ edges. The runtime then follows from Proposition~\ref{prop:ms-solver}, where the required accuracy to satisfy \eqref{eq:cond-apprx-function} and \eqref{eq:cond-apprx-l1} is a polynomial in problem parameters, following the proofs of Theorem~\ref{lem:ms-sinkhorn} and Lemma~\ref{lem:dec-solver-bs} (see~\Cref{ssec:high-OT} for more details). The requirements of $\Round$ are clear from inspection and the additive range of entropy.
\end{proof}

\thmmaindexsinkhorn*
\begin{proof}
It suffices to combine Lemma~\ref{lem:canonical-sink}, Lemma~\ref{lem:dec-solver-ms}, and Corollary~\ref{coro:ddbm-reduce}.
\end{proof}

\subsection{DDBM via Sinkhorn objective solver in~\cite{chen2022maximum}}\label{ssec:dynamic-redx-maxflow}

In this section we give an alternative $m^{1+o(1)} \eps^{-2}$-time algorithm for decremental bipartite matching, leveraging a recent breakthrough of a $m^{1+o(1)}$-time algorithm for general graph flow problems in \cite{chen2022maximum}. We follow the same notational conventions as in previous sections.

Our approach is to simply leverage the recent algorithm of \cite{chen2022maximum} to solve the regularized subproblems \eqref{eq:def-matching-two} from the previous section. We recall the following result from \cite{chen2022maximum}.

\newcommand{\f}{\mathbf{f}}
\begin{proposition}[Theorem 10.16 from \cite{chen2022maximum}]\label{thm:ent_maxflow}
	Given a graph $G = (V,E)$, demands $d \in \R^V$, costs $c \in \R^E$, and weights $w \in \R^E_{\geq 0}$, all entrywise bounded by $\exp(\log^{O(1)} m)$, let $\mb$ be its unsigned adjacency matrix, and  $h(f) = \sum_{e \in E} c_e f_e + w_e f_e \log f_e$. Then in $m^{1+o(1)}$ time we can find a flow $f$ with $\mb^\top f = d$, $f \geq 0$, and for any constant $C >0$
\[
h(f) \leq \min_{B^\top f^\star= d, f^\star \geq 0} h(f^\star) + m^{-C} .
\]
\end{proposition}

\begin{lemma}[Canonical solver for Sinkhorn objectives]\label{lem:dec-solver-maxflow}
    For $\eps = \Omega(m^{-3})$, there is an $(\eps,\mathcal{T})$-canonical solver for the family of $f_{M,E}^\sink$ defined in \Cref{lem:canonical-sink}, using \Cref{thm:ent_maxflow} as $\solve$ and truncation to $E$ as $\Round$, running in $m^{1+o(1)}$ time. 
    \end{lemma}
    \begin{proof}
    To implement $\solve$, we apply \Cref{thm:ent_maxflow} to the problem \eqref{eq:def-matching-two} and solve the subproblem to accuracy $O(\frac{\eps^3}{m})$, for a sufficiently small constant. We observe that \eqref{eq:def-matching-two} is a problem exactly of the form where \Cref{thm:ent_maxflow} applies: by the multiplicative guarantee on $\nu^E$ and $\MCM(E)$ and since $\MCM(E)$ is $8$-approximated by $M$, \eqref{eq:cond-apprx-function} holds immediately. We additionally obtain \eqref{eq:cond-apprx-l1}, as the function $f_{M,E}^\sink$ is $\Omega(\frac{\eps}{\log m})$-strongly convex in the coordinates of $x$.  The requirements of $\Round$ are clear by inspection and by the additive range of entropy, as shown in~\Cref{lem:dec-solver-ms}.
    \end{proof}
    
    With this result, we follow \Cref{thm:main-dec-sinkhorn} and obtain an improved runtime for DDBM. 
    \thmmaindecsinkhornmaxflow*

    \begin{proof}
    It suffices to combine Lemma~\ref{lem:canonical-sink}, Lemma~\ref{lem:dec-solver-maxflow}, and Corollary~\ref{coro:ddbm-reduce}.
    \end{proof}

\section{Proofs for~\Cref{sec:framework}}\label{app:framework}

We first remark on the assumptions made throughout~\Cref{sec:framework}, restated as below: For some $\delta > 0$, we have
\begin{enumerate}
\item upper bounds on entries: $\|\ma\|_{\infty} \le 1$, $\norm{b}_\infty \le \Bmax$, $\norm{c}_\infty \le \Cmax$. For simplicity, we assume $\Bmax \ge \Cmax \ge 1$, as otherwise we set $\Cmax \gets \max(1, \Cmax)$ and $\Bmax \gets \max(\Cmax, \Bmax)$.
\item lower bounds on matrix column entries: $\max_{i}|\ma_{ij}|\ge \delta$ for every $j\in[n]$.
\end{enumerate} 

The second assumption can be satisfied by padding entries of $\ma$ by $\delta$, which at most incurs $O(\delta) \ll O(\epsilon)$ error in the objective value for any $\delta\ll\epsilon$. Our runtimes extend to general $\ma$ in a scale-invariant way (i.e.\ by scaling the whole problem by a factor of $\frac{1}{\norm{\ma}_\infty}$ and then scaling up the resulting error), so the first assumption (on $\norm{\ma}_\infty$) is for simplicity. All of our runtimes depend logarithmically on the quantity $\Bmax$, which is polynomially bounded in all our applications.

Finally, regarding $\Cmax$, it is clear we can assume without loss of generality that $c \le \1$ entrywise, since shifting $c$ by a multiple of the all-ones vector changes the objective value in \eqref{eq:main-reg-apprx}. In the unregularized case, i.e.\ $f_{\mu, \eps}$ with $\mu = \eps = 0$, \cite{AssadiJJST22} shows that we can also \emph{lower bound} the vector $c$ entrywise by $-1$. In the regularized case, we provide the following lemma which shows that by truncating the entries of $c$, we do not lose too much in objective value (see Appendix~\ref{app:framework}).

\begin{restatable}{lemma}{restatecmaxbound}\label{lem:cmaxbound}
Let $\tau \ge 0$, and define $S_\tau = \{i \in [m] \mid c_i \ge \min_{i \in [m]} c_i + \tau\}$ to be the large entries of $c$. Define $\xset' \defeq \{x \in \Delta^m \mid x_i = 0 \text{ for all } i \in S_\tau\}$. Then,
\begin{align*}
\min_{x \in \xset'} \max_{y \in [0, 1]^n} \fme(x, y) \le \min_{x \in \Delta^m} \max_{y \in [0, 1]^n} \fme(x, y) + \mu m \exp\Par{-\frac{\tau - 3}{\mu}}.
\end{align*}
\end{restatable}
In other words, Lemma~\ref{lem:cmaxbound} shows that if our goal is to find an $x \in \Delta^m$ which approximately minimizes $\max_{y \in [0, 1]^n} \fme(x, y)$, we may restrict ourselves to considering only $x$ supported on coordinates where $c$ is polylogarithmically bounded, without much loss in the error. In particular, setting $\tau = \Theta(\log \frac{m}{\sigma})$ above, where the final desired error is $\sigma$, yields this claim for $\mu \le 1$. Our methods will have runtimes depending linearly on $\Cmax$, which upon applying the preprocessing of Lemma~\ref{lem:cmaxbound}, is an overall polylogarithmic dependence on the final target accuracy.

Now we prove Lemma~\ref{lem:cmaxbound}. Before we do so, we state two simple helper lemmas used in its proof.

\begin{lemma}\label{lem:eachy}
Let $\xset, \yset$ be compact and convex, let $\xset' \subseteq \xset$ also be compact and convex, and let $f: \xset \times \yset \to \R$ be convex-concave. Suppose for some $\Delta > 0$ it is the case that for all $y \in \yset$,
\[\min_{x \in \xset'} f(x, y) \le \min_{x \in \xset} f(x, y) + \Delta.\]
Then,
\[\min_{x \in \xset'} \max_{y \in \yset} f(x, y) \le \min_{x \in \xset} \max_{y \in \yset} f(x, y) + \Delta.\]
\end{lemma}
\begin{proof}
Let $(x_\star, y_\star)$ be the optimizer to $f$ over $\xset \times \yset$, and let $(x_\star', y_\star')$ be the optimizer to $f$ over $\xset' \times \yset$. We conclude with the following sequence of inequalities:
\begin{align*}
f(x'_\star, y'_\star) \le \min_{x \in \xset} f(x, y'_\star) + \Delta \le f(x_\star, y'_\star) + \Delta \le f(x_\star, y_\star) + \Delta.
\end{align*}
The three inequalities respectively used our assumption, $x_\star \in \xset$, and $y_\star = \argmax_{y \in \yset} f(x_\star, y)$.
\end{proof}

\begin{lemma}\label{lem:sumofexp}
For $\mu > 0$, define the function $\smin_\mu: \R^m \to \R$ by
\[\smin_\mu(v) \defeq -\mu\log\Par{\sum_{i \in [m]} \exp\Par{-\frac 1 \mu v_i}}.\]
Fix $v$, and let $S \subseteq [m]$ be a set such that for all $i \in [m] \setminus S$, we have $v_i \ge \min_{i \in [m]} v_i + T$ for some $T > 0$. Further let $v' \in \R^S$ be the restriction of $v$ to the set $S$. Then,
\[\smin_\mu\Par{v'} \defeq -\mu \log\Par{\sum_{i \in S} \exp\Par{-\frac 1 \mu v'_i}} \le \smin_\mu(v) + \mu m \exp\Par{-\frac{T}{\mu}}.\]
\end{lemma}
\begin{proof}
Let $\vmin \defeq \min_{i \in [m]} v_i$. Note that
\begin{align*}
\sum_{i \in [m] \setminus S} \exp\Par{-\frac 1 \mu v_i} \le m\exp\Par{-\frac 1 \mu (\vmin + T)} \le m\exp\Par{-\frac T \mu} \sum_{i \in S}  \exp\Par{-\frac 1 \mu v_i}.
\end{align*}
Here we used that some element in $S$ achieves $v_i = \vmin$. Hence,
\[\sum_{i \in [m]} \exp\Par{-\frac 1 \mu v_i} \le \Par{1 + m\exp\Par{-\frac T \mu}} \sum_{i \in S} \exp\Par{-\frac 1 \mu v_i}.\]
Taking logarithms and scaling by $-\mu$ yields the claim, where we use $\log(1 + c) \le c$ for $c \ge 0$.
\end{proof}

Combining these lemmas allow us to prove~\Cref{lem:cmaxbound} formally.

\begin{proof}[Proof of~\Cref{lem:cmaxbound}]
By Lemma~\ref{lem:eachy}, it suffices to first fix $y \in [0, 1]^n$, and prove that
\begin{equation}\label{eq:eachy}\min_{x \in \xset'} f_{\mu, \eps}(x, y) \le \min_{x \in \Delta^m} f_{\mu, \eps}(x, y) + \mu m \exp\Par{-\frac{\tau - 3}{\mu}}.\end{equation}
Clearly the term $b^\top y$ cancels from both sides. Next, define
\[v \defeq \ma y + c - \frac \eps 2 |\ma| (y^2).\]
Let $v'$ be the restriction of $v$ to the indices in $S_\tau$. By explicitly minimizing over $x \in \Delta^m$ and $x \in \xset'$ respectively, \eqref{eq:eachy} is equivalent to proving
\[\smin_\mu(v') \le \smin_\mu(v) + \mu m \exp\Par{-\frac{\tau - 3}{\mu}}.\]
Since $\ma y - \frac \eps 2 |\ma|(y^2)$ has a range of at most $3$, every truncated entry of $v'$ must be at least $(\tau - 3)$ larger than the smallest entry of $v'$. We conclude by applying Lemma~\ref{lem:sumofexp} with $T = \tau - 3$.
\end{proof}

\subsection{Proofs for~\Cref{ssec:high-level}}

\outerloop*

\begin{proof}
Fix an iteration $k$. We have
\begin{flalign}
\nu V_{z_{k-1/2}}(z^\star) & \stackrel{(i)}{\le}  
\inner{g(z_{k-1/2})-g(z^\star)}{z_{k-1/2} - z^\star}
 \stackrel{(ii)}{\le}  \inner{g(z_{k-1/2})}{z_{k-1/2} - z^\star}
 \nonumber \\ &
 =  \inner{g(z_{k-1/2})}{z_{k} - z^\star} + \inner{g(z_{k-1/2})}{z_{k-1/2} - 
 z_{k}},\label{eq:strong-eq-1}
\end{flalign}
where we used $(i)$ the definition of $\nu$-strong monotonicity,  and $(ii)$ optimality of $z^\star$.

Next, %
\begin{flalign}
	\inner{g(z_{k-1/2})}{z_{k} - z^\star}- \frac{\veps}{2}  \stackrel{(i)} {\leq}&  -\langle\alpha 
	\nabla V_{z_{k-1}}(z_k)+\nu\nabla V_{z_{k-1/2}}(z_k),z_k-z^\star\rangle  
	\nonumber \\ 
	 \stackrel{(ii)}{\leq} &
	 \nu\Big(V_{z_{k - 1/2}}(z^\star) - V_{z_k}(z^\star) - V_{z_{k - 
	 1/2}}(z_k)\Big)
\\& 
+\alpha\Big(V_{z_{k-1}}(z^\star)-V_{z_{k}}(z^\star)-V_{z_{k-1}}(z_k)\Big) 
\nonumber \\ 
\le &\alpha V_{z_{k-1}}(z^\star)-(\nu+\alpha)V_{z_k}(z^\star)+\nu 
V_{z_{k-1/2}}(z^\star)-\alpha V_{z_{k-1}}(z_k),\label{eq:strong-eq-2}
\end{flalign}
and similarly, 
\begin{flalign}
	& \inner{g(z_{k-1/2})}{z_{k-1/2} - z_k} - \frac{\veps}{2}\nonumber\\
	& \hspace{1em} \stackrel{(i)}{\leq} -\langle\alpha 
	\nabla V_{z_{k-1}}(z_{k-1/2}),z_{k-1/2}-z_k\rangle + \inner{g(z_{k-1/2})-g(z_{k-1}))}{z_{k-1/2} - z_k}
	\nonumber \\
	& \hspace{1em} \stackrel{(ii)}{\leq}  
	 \alpha\Big(V_{z_{k-1}}(z_k)-V_{z_{k-1/2}}(z_k)-V_{z_{k-1}}(z_{k-1/2})\Big)  + \inner{g(z_{k-1/2})-g(z_{k-1}))}{z_{k-1/2} - z_k}
\nonumber \\ 
& \hspace{1em} \stackrel{(iii)}{\leq}
	 \alpha V_{z_{k-1}}(z_k),\label{eq:strong-eq-3}
\end{flalign}
where we used $(i)$ the approximate proximal oracle optimality conditions for $z_k$ and $z_{k - 1/2}$, $(ii)$ 
the three-point property of Bregman divergence~\eqref{eq:three-point}, and $(iii)$ the relative Lipschitzness conddition of $g$. In 
the last inequality, we also use nonnegativity of divergence term $V$. 

Now, 
combining~\eqref{eq:strong-eq-2} and~\eqref{eq:strong-eq-3} with~\eqref{eq:strong-eq-1} yields
\begin{equation*}
\label{eq:strong-eq-both}	
\begin{aligned}
\nu V_{z_{k-1/2}}(z^\star) & \le \alpha 
V_{z_{k-1}}(z^\star)-(\nu+\alpha)V_{z_k}(z^\star)+\nu 
V_{z_{k-1/2}}(z^\star) +\veps.
\end{aligned}
\end{equation*}
Rearranging terms then yields
\[
V_{z_k}(z^\star)\le\frac{\alpha}{\nu+\alpha}
V_{z_{k-1}}(z^\star)+\frac{\veps}{\nu+\alpha}.
\]
Applying this bound recursively $K$ times and using that 
$V_{z_0}(u) \le 
r(u)-r(z_0) \le \Theta$ for $z_0$ the minimizer of $r$, we have
\[
V_{z_K}(z^\star)\le\left(\frac{\alpha}{\nu+\alpha}\right)^K\Theta 
+\sum_{k = 0}^{K-1}\left(\frac{\alpha}{\nu + \alpha}\right)^{k}  
\left(\frac{\veps}{\nu + \alpha}\right) 
\le 
\left(\frac{\alpha}{\nu+\alpha}\right)^K\Theta +
\frac{\veps}{\nu}.
\]
\end{proof}

\innerloop*

To prove~\Cref{coro:altmin-conv}, we first provide technical lemma that shows how to transfer the accuracy desired for transferring from function error in~\Cref{alg:altmin} guarantees to the duality gap error desired by an approximate extragradient step in~\Cref{alg:outerloop}, in a controllable sense.

 \begin{lemma}\label{lem:l1diffneeded}
Let $\delta,\veps \in(0,1)$, $\rho\ge 1$. Suppose we are given $\gamma\in\zset_* = \xset_*\times \yset_*$ satisfying $\max(\norm{\gamma\x}_\infty,\norm{\gamma\y}_1)\le B$, and define the proximal oracle subproblem solution  
\[x_\OPT, y_\OPT = \argmin_{x\in\Delta^m}\argmin_{y\in[0,1]^n} f(x,y)\defeq\langle \gamma^x, x\rangle+\langle\gamma^y,y\rangle+\theta r(x,y)~~\text{for some}~\theta>0.\] Let $z = (x,y)$ be an approximate solution satisfying  $f(x,y)\le f(x_\OPT,y_\OPT)+\tfrac{\veps^2\delta^2\theta}{64\rho(8n\theta+B)^2}$, and suppose $x, x_\OPT\ge \delta$ holds elementwise. For any $w\in\zset = \xset\times\yset$, 
\[\inprod{\gamma+\theta \nabla r(z)}{z - w} \le \frac{\veps}{2}.\]
\end{lemma}

\begin{proof}
Let $z_\OPT = (x_\OPT,y_\OPT)$. Under the assumption that $f(z)-f(z_\OPT)\le A \defeq \tfrac{\veps^2\delta^2\theta}{64\rho(8n\theta+B)^2}$, using the lower bound in \eqref{eq:reg-convex-and-bound} with $|\ma|^\top x \ge \delta^2$ entrywise by the assumption $\max_i|\ma_{ij}| \ge \delta$ for every column $j$, and since entropy is strongly convex in the $\ell_1$ norm, we have
\begin{equation}\label{eq:error-condition}	
\begin{aligned}
& \frac{\theta \rho}{4}\norm{x-x_\OPT}_1^2+\frac{\delta^2\theta}{2\rho}\norm{y-y_\OPT}_\infty^2\le f(z)-f(z_\OPT)\le A\\
~\Longrightarrow~&~\norm{x-x_\OPT}_1\le \frac{\veps}{4\Par{B+\frac{8n\rho\theta}{\delta}}}~\text{and}~\norm{y-y_\OPT}_\infty\le \frac{\veps}{4\Par{B+\frac{8n\theta}{\rho}}}.
\end{aligned}
\end{equation}
For $w=(u,v)$, by definition of $z_\OPT$, $\inprod{\gamma+\theta \nabla r(z_\OPT)}{z_\OPT-w}\le 0$. Hence, we bound
\begin{equation}\label{eq:error-bounds-total}
\begin{aligned}
\inprod{\gamma+\theta\nabla r(z)}{z-w} &\le \inprod{\gamma+\theta\nabla r(z)}{z-w}-\inprod{\gamma+\theta\nabla r(z_\OPT)}{z_\OPT-w}\\
&= \inprod{\gamma}{z-z_\OPT} + \theta \inprod{\nabla r(z)-\nabla r(z_\OPT)}{z_\OPT-w} + \theta\inprod{\nabla r(z)}{z-z_\OPT}.
\end{aligned}
\end{equation}
We now bound the terms on the right-hand side of \eqref{eq:error-bounds-total}:
\begin{equation}\label{eq:error-bounds}
\begin{aligned}
\inprod{\gamma}{z-z_\OPT}& \le \|\gamma\x\|_\infty\|x-x_\OPT\|_1+\|\gamma\y\|_1\|y-y_\OPT\|_\infty\\
&  \le B\norm{x-x_\OPT}_1 + B\norm{y-y_\OPT}_\infty,\\
\inprod{\nabla_x r(z)-\nabla_x r(z_\OPT)}{x_\OPT-u} & \le \|x_\OPT-u\|_1\norm{\nabla_x r(z)-\nabla_x r(z_\OPT)}_\infty\\
&\le 2\rho \norm{\log \frac{x}{x_\OPT}}_\infty + \frac{2}{\rho}\norm{|\ma|( (y + y_\OPT) \circ (y - y_\OPT))}_\infty \\
&\le  \frac{2\rho}{\delta} \norm{x - x_\OPT}_1 + \frac{4}{\rho}\norm{y - y_\OPT}_\infty, \\
\inprod{\nabla_y r(z)-\nabla_y r(z_\OPT)}{y_\OPT-v} & \le  \norm{y_\OPT-v}_\infty \norm{\nabla_y r(z)-\nabla_y r(z_\OPT)}_1\\
& \le \frac{2n}{\rho}\norm{y-y_\OPT}_\infty + \frac{2n}{\rho}\norm{x - x_\OPT}_1,\\
\inprod{\nabla r(z)}{z-z_\OPT} & \le \|\nabla_x r(z)\|_\infty\|x-x_\OPT\|_1+\|\nabla_y r(z)\|_1\|y-y_\OPT\|_\infty\\
&  \le \left(\rho\left(1+\log\Par{\frac 1 \delta}\right)+\frac{1}{\rho}\right)\|x-x_\OPT\|_1 + \frac{2}{\rho}\norm{y-y_\OPT}_\infty.
\end{aligned}
\end{equation}
For the second bound, we used $\log(1 + c) \le c$ for whichever of $c = \frac{x_j}{[x_\OPT]_j}$ or $c = \frac{[x_\OPT]_j}{x_j}$ is larger, and the entrywise lower bounds on $x$ and $x_\OPT$. For the third bound, we used
\begin{flalign*}
& \norm{\diag{y_\OPT}|\ma|^\top x_\OPT - \diag{y}|\ma|^\top x}_1 \le \sum_{i \in [m], j \in [n]} |\ma|_{ij} \Abs{\Brack{x_\OPT}_i\Brack{y_\OPT}_j - x_i y_j} \\
&\hspace{10em}\le \sum_{i \in [m], j \in [n]} |\ma|_{ij} \Par{\Abs{\Brack{y_\OPT}_j - y_j} + \Abs{\Brack{x_\OPT}_i - x_i} } \\
&\hspace{10em}\le n (\norm{y_\OPT - y}_\infty + \norm{x_\OPT - x}_1).
\end{flalign*}

Plugging the bounds of~\eqref{eq:error-bounds} back in~\eqref{eq:error-bounds-total}, we can thus conclude that for $\rho\ge 1$, $\delta\le 1$
\[
\inprod{\gamma+\theta\nabla r(z)}{z-w}\le \left(B+\frac{8n\rho\theta}{\delta}\right)\norm{x-x_\OPT}_1 + \left(B+\frac{8n\theta}{\rho}\right)\norm{y-y_\OPT}_\infty\le \frac{\veps}{2},
\]
where for the last inequality we use conditions in~\eqref{eq:error-condition}.
\end{proof}

We use this lemma to prove~\Cref{coro:altmin-conv} formally below.

\begin{proof}[Proof of~\Cref{coro:altmin-conv}]
	We first bound $f(x^{(1)}, y^{(0)})-f(x_\OPT, y_\OPT)\le \max_{z\in\zset} f(z)-f(z_\OPT)\le 2\norm{\gamma\x}_\infty+2\norm{\gamma\y}_1+\theta \Par{\rho \log m + \frac{2}{\rho}}\le 4B + \theta \Par{\rho \log m + \frac{2}{\rho}}$ where the last inequality uses $\ell_1$-$\ell_\infty$ H\"older, the domain definitions $\xset = \Delta^{m}$, $\yset = [0,1]^n$, and the definition of $r=\rme$ as in~\eqref{def:reg-sm}. 

	Now applying~\Cref{lem:progress-altmin}, after $O\left(\log\Par{\frac{\rho(B+mn\theta)^2}{\delta\veps\theta}}\right)$ iterations for a sufficiently large constant, 
	\begin{flalign*}
	& f(x^{(T+1)},y^{(T)})-f(x_\OPT, y_\OPT) \le \left(1-\frac{1}{2\kappa}\right)^T\left(f(x^{(1)},y^{(0)})-f(x_\OPT, y_\OPT)\right)\\
	& \hspace{10em}\le \left(1-\frac{1}{2\kappa}\right)^T\left(4B + \theta \Par{\rho \log m + \frac{2}{\rho}}\right)\le \frac{\veps^2\delta^2\theta}{64\rho(8n\theta+B)^2},
	\end{flalign*}
	which by~\Cref{lem:l1diffneeded} in turn implies it implements an approximate proximal oracle to $\frac \veps 2$ accuracy.
\end{proof}

\subsection{Proofs for~\Cref{ssec:helper}}

We give a helpful variant of the Cauchy-Schwarz inequality in appropriate ``local norms'' and a consequence about approximating the Hessian of our regularizer. These find uses in proving \Cref{lem:rel-lip}.

\restatercab*

\begin{proof}
	For the first property, it suffices to square both sides and use Cauchy-Schwarz:
	\begin{align*}\Par{\sum_{i \in [m]} \sum_{j \in [n]} \left|\ma_{ij} u_i v_j\right|}^2 &\le \Par{\sum_{i \in [m]} \sum_{j \in [n]} \frac{\Abs{\ma_{ij}}}{[\bar{x}]_i} u_i^2}\Par{\sum_{i \in [m]} \sum_{j \in [n]} \Abs{\ma_{ij}}[\bar{x}]_i v_j^2} \\
	&\le \Par{\sum_{i \in [m]} \frac{\norm{\ma}_\infty}{[\bar{x}]_i} u_i^2}\Par{\sum_{j \in [n]} \Brack{\Abs{\ma}^\top \bar{x}}_j v_j^2} \\
	&\le \norm{u}_{\diag{\frac{1}{\bar{x}}}}^2 \norm{v}_{\diag{\Abs{\ma}^\top \bar{x}}}^2.
	\end{align*}
	
	Next, given any $w = (u,v)$, we have 
	\begin{align*}
		w^\top \nabla^2r(x,y)w & = \rho\norm{u}_{\diag{\frac{1}{x}}}^2+\frac{2}{\rho}\norm{v}_{\diag{\Abs{\ma}^\top x}}^2+\frac{4}{\rho}u^\top\ma\diag{y}v\\
		& \le \rho\norm{u}_{\diag{\frac{1}{x}}}^2+\frac{2}{\rho}\norm{v}_{\diag{\Abs{\ma}^\top x}}^2+\frac{4}{\rho}\norm{u}_{\diag{\frac{1}{x}}}\norm{\diag{y}v}_{\diag{\Abs{\ma}^\top x}}\\
		& \le \Par{\rho+\frac{2}{\rho}}\norm{u}_{\diag{\frac{1}{x}}}^2+\frac{4}{\rho}\norm{v}_{\diag{\Abs{\ma}^\top x}}^2\\
		& \le 4\Par{\frac{\rho}{2}\norm{u}_{\diag{\frac{1}{x}}}^2+\frac{1}{\rho}\norm{v}_{\diag{\Abs{\ma}^\top x}}^2} = 4w^\top\md(x)w.
\end{align*}
Similarly, on the other side we have
	\begin{align*}
		w^\top \nabla^2r(x,y)w & = \rho\norm{u}_{\diag{\frac{1}{x}}}^2+\frac{2}{\rho}\norm{v}_{\diag{\Abs{\ma}^\top x}}^2+\frac{4}{\rho}u^\top\ma\diag{y}v\\
		& \ge \rho\norm{u}_{\diag{\frac{1}{x}}}^2+\frac{2}{\rho}\norm{v}_{\diag{\Abs{\ma}^\top x}}^2-\frac{4}{\rho}\norm{u}_{\diag{\frac{1}{x}}}\norm{\diag{y}v}_{\diag{\Abs{\ma}^\top x}}\\
		& \ge \Par{\rho-\frac{4}{\rho}}\norm{u}_{\diag{\frac{1}{x}}}^2+\Par{\frac{2}{\rho}-\frac{1}{\rho}}\norm{v}_{\diag{\Abs{\ma}^\top x}}^2\\
		& \ge \frac{\rho}{2}\norm{u}_{\diag{\frac{1}{x}}}^2+\frac{1}{\rho}\norm{v}_{\diag{\Abs{\ma}^\top x}}^2 = w^\top\md(x)w.
\end{align*}

\end{proof}

\lempaderror*

\begin{proof}
We write $r(x,y) = \rho H(x)+\frac{1}{\rho}Q(x,y)$ where 
\[H(x) = \sum_{i\in[m]}x_i \log x_i,\; Q(x,y) = \Par{y^2}^\top|\ma|^\top x,\]
and bound the Bregman divergences induced by $H$ and $Q$ respectively. 

First note that by definition of the padding oracle we have $\norm{x_k-\bx_k}_1 \le \norm{x_k - \hx_k}_1 +\norm{\bx_k - \hx_k}_1 \le 2m\delta$, and that $x_k$ is a $m\delta$-padding of $\bx_k$ by the definition of padding in~\cite{CJST20} (cf.\ Definition 2). Thus using Lemma 6 of~\cite{CJST20}, we have
\[
V^H_z(w)-V^H_{\bz}(w)\le m\delta.
 \]
For the quadratic part, letting $w = (u, v)$, we have
\begin{align*}
V^Q_{z}\Par{w} - V^Q_{\bz}\Par{w} 
 & = Q(\bz)-Q(z)+\inprod{\nabla Q(\bz)}{w-\bz}-\inprod{\nabla Q(z)}{w-z} \\
&  \le \norm{x-\bx}_1 + \inprod{|\ma|y^2}{x -\bx} + \inprod{\diag{y}|\ma|^\top (\bx - x)}{v - y} \\
&\le 4\norm{x-\bx}_1\le 8m\delta,
\end{align*}
where the first inequality used $\by = y$, and both the first and second used the assumed bounds $\norm{y}_\infty \le 1$, $\norm{\ma}_\infty \le 1$. Combining these bounds proves the statement.
\end{proof}

\lemstability*

\begin{proof}
We first handle the case of~\Cref{line:outer-grad-sm} in Algorithm~\ref{alg:outerloop}. Using Algorithm~\ref{alg:altmin} to implement this step, we observe that the iterates satisfy
\begin{align*}
x^{(t + 1)} \gets & \argmin_{x \in \Delta^m}\inprod{\gamma\x}{x} + \theta r(x, y^{(t)})~~\text{where}~~\theta = \alpha,\\
~\text{and}~~\gamma\x & = g\x (x_{k-1},y_{k-1}) - \alpha \nabla_x r(z_{k-1})\\
& =  \ma y_{k-1} + c + \mu(\1+\log(x_{k-1})) - \frac{\eps}{2} \Abs{\ma} (y_{k-1}^2) -\alpha \rho (1+\log x_{k-1}) - \frac{\alpha}{\rho}|\ma| y_{k-1}^2,
\end{align*}
which implies
\[x^{(t+1)} \propto x_{k-1}\circ \exp\Par{\frac{1}{\alpha\rho}\left(-\frac{\alpha}{\rho}|\ma|\left(y^{(t)}\right)^2-\ma y_{k-1} - c+\frac{\eps}{2}|\ma|y_{k-1}^2+\frac{\alpha}{\rho}|\ma|y_{k-1}^2-\mu \log(x_{k-1})\right)}.\]

Consequently, under the given assumptions on $x_{k - 1}$, $\rho$, and $\alpha$, and using $\norm{c}_\infty \le \Cmax$,
\begin{flalign*}
	& \left|\log\Par{\frac{x^{(t+1)}}{x_{k-1}}}\right|\le \frac{1}{\alpha\rho}\Par{1 + \Cmax +\frac{\alpha}{\rho} + \frac \eps 2 +\mu \log\frac 2 \delta}\le \frac{1}{18}\\
	& \hspace{10em}\implies x^{(t+1)}\in x_{k-1}\cdot\Brack{\exp\Par{-\frac 1 9},\exp\Par{\frac 1 9}}.
\end{flalign*}

Next, we handle the case of~\Cref{line:outer-extragrad-sm} in Algorithm~\ref{alg:altmin}. Here, the iterates of Algorithm~\ref{alg:altmin} satisfy
\begin{align*}
x^{(t + 1)} \gets & \argmin_{x \in \Delta^m}\inprod{\gamma\x}{x} + \theta r(x, y^{(t)})~~\text{where}~~\theta = \alpha+\nu,\\
~\text{and}~~\gamma\x &= g\x (x_{k-1/2},y_{k-1/2}) - \alpha \nabla_x r(z_{k-1}) - \nu\nabla_x r(z_{k-1/2})\\
&=  \ma y_{k-1/2} + c + \mu(\1+\log(x_{k-1/2})) - \frac{\eps}{2} \Abs{\ma} (y_{k-1/2}^2) \\
& \quad\quad\quad-\alpha \rho (1+\log x_{k-1}) - \frac{\alpha}{\rho}|\ma| y_{k-1}^2 -\nu \rho (1+\log x_{k-1/2}) - \frac{\nu}{\rho}|\ma| y_{k-1/2}^2.
\end{align*}
Hence,
\begin{align*}
x^{(t+1)} \propto~&~x_{k-1}^{\frac{\alpha}{\alpha+\nu}}\circ x_{k-1/2}^\frac{\nu}{\alpha+\nu}\circ \tau\x\\
\text{where}~\tau\x =~&~\exp\biggl(\frac{1}{(\alpha+\nu)\rho}\biggl(-\frac{\alpha}{\rho}|\ma|\left(y^{(t)}\right)^2-\frac{\nu}{\rho}|\ma|\left(y^{(t)}\right)^2-\ma y_{k-1/2} - c\\
& \quad \quad \quad \quad +\frac{\eps}{2}|\ma|y_{k-1/2}^2+\frac{\alpha}{\rho}|\ma|y_{k-1}^2 + \frac{\nu}{\rho}|\ma|y_{k-1/2}^2-\mu \log(x_{k-1/2})\biggr)\biggr).
\end{align*}
Consequently, under a similar calculation as before,
\begin{flalign*}
& \exp\Par{-\frac 1 {18}}\le \tau\x\le \exp\Par{\frac 1 {18}}~\text{entrywise}\\
& \hspace{10em}~\Longrightarrow x^{(k+1)}\in x_{k-1}^{\frac{\alpha}{\alpha+\nu}}\circ x_{k-1/2}^\frac{\nu}{\alpha+\nu}\cdot\Brack{\exp\Par{-\frac 1 9},\exp\Par{\frac 1 9}}.
\end{flalign*}

\end{proof}

\subsection{Proofs for~\Cref{ssec:framework-alg}}\label{app:framework-alg}

\lemsm*

\begin{proof}
Throughout the proof, let $g \defeq \gme$, $f \defeq \fme$, and $r \defeq \rme$ for notational simplicity. In order to show the desired bound
\begin{equation}\label{eq:smshow}\inprod{g(w) - g(z)}{w - z} \ge \frac{1}{3}\sqrt{\frac{\mu\eps}{2}}\inprod{\nabla r(w) - \nabla r(z)}{w - z},\end{equation}
we begin by putting \eqref{eq:smshow} into a more convenient form. Letting $\jac(z)$ be the Jacobian of $g$ at the point $z$, we have by direct integration (where $z_t \defeq (1 - t)z + tw$ for all $t \in [0, 1]$)
\[\inprod{g(w) - g(z)}{w - z} = \int_0^1 (w - z)^\top \jac(z_t) (w - z) dt.\]
Since quadratic forms through a square matrix $\mm$ are preserved by replacing $\mm$ with its symmetric part $\half(\mm + \mm^\top)$, it suffices to understand the symmetric part of $\jac(z_t)$. The restriction of $\jac_{\mu}$ to the $xy$ and $yx$ blocks is skew-symmetric, since the blocks are respectively $\nabla^2_{xy} f(x, y)$ and $-\nabla^2_{yx} f(x, y)$. Similarly, its restrictions to its $xx$ and $yy$ blocks are symmetric. Hence, we have
\begin{equation}\label{eq:lhssm}\begin{aligned}\inprod{g(w) - g(z)}{w - z}&= \int_0^1 (w - z)^\top \begin{pmatrix} \nabla^2_{xx} f(z_t) & \mzero \\ \mzero & -\nabla^2_{yy} f(z_t) \end{pmatrix} (w - z)dt \\
&= \int_0^1 (w - z)^\top \begin{pmatrix} \mu\cdot \diag{\frac 1 {x_t}} & \mzero \\ \mzero & \eps\cdot \diag{\Abs{\ma}^\top x_t} \end{pmatrix} (w - z)dt. \end{aligned}\end{equation}
In the last line, we denoted $z_t \defeq (x_t, y_t)$ for all $t \in [0, 1]$. Next, to bound the right hand side of \eqref{eq:smshow}, integrating once more yields
\begin{equation}\label{eq:rhssm}
\begin{aligned}
& \inprod{\nabla r(w) - \nabla r(z)}{w - z} = \int_0^1 (w - z)^\top \nabla^2 r(z_t) (w - z) dt \\
& \hspace{5em} = \int_0^1 (w - z)^\top \begin{pmatrix} \rho \cdot \diag{\frac{1}{x_t}} & \frac{2}{\rho}\Abs{\ma} \diag{y_t} \\ \frac{2}{\rho}\diag{y_t} \Abs{\ma}^\top & \frac{2}{\rho} \diag{\Abs{\ma}^\top x_t} \end{pmatrix} (w - z) dt \\
& \hspace{5em} \le \int_0^1 (w - z)^\top \begin{pmatrix} 2\rho \cdot \diag{\frac{1}{x_t}} & \mzero \\ \mzero & \frac{4}{\rho} \diag{\Abs{\ma}^\top x_t}\end{pmatrix} (w - z) dt. 
\end{aligned}\end{equation}
The last line used~\Cref{lem:reg-convex-and-bound} for all $(x_t, y_t) \in \Delta^m \times [0, 1]^n$. Combining the bounds \eqref{eq:lhssm} and \eqref{eq:rhssm} yields the conclusion.
\end{proof}

\corostability*
\begin{proof}
	It suffices to prove $x_{k-1}\ge \frac \delta 2$ entrywise as all other conclusions are then immediate consequences of~\Cref{lem:stability} (the conclusions about $x_\OPT$ follow from taking a limit). When $k=1$, $x_{k-1}\ge \frac \delta 2$ holds by the assumption that $\delta\le \frac 1 {2m}$. When $k>1$, $x_{k-1}$ is the result of padding with parameter $\delta$. We note for  any $x\in\Delta^m$  we have the desired
	\[
	1\le \norm{\max\Par{x,\delta}}_1\le 1+\delta m\implies \mathcal{O}_\delta(x)\ge \frac{\delta}{\norm{\max\Par{x,\delta}}_1}  \ge \frac{\delta}{1+\delta m} \ge \frac{\delta}{2}.
	\]	
\end{proof}

\lemrellip*

\begin{proof}
The conclusion that $x_{k-1/2}, \bx_{k}\in[\tfrac{1}{2}x_{k-1},2x_{k-1}]$ elementwise follows from \Cref{coro:stability}. Next, we have for $z_\beta = z_{k-1}+\beta(z_{k-1/2}-z_{k-1})$, and $\rho\ge 6$, 
		\begin{align*}
			V_{z_{k-1}}(z_{k-1/2}) & = \int_{0}^1(1-\beta)\|z_{k-1/2}-z_{k-1}\|^2_{\nabla^2 r(z_{k-1}+\beta(z_{k-1/2}-z_{k-1}))}d\beta \\
			& \ge \frac{1}{4}\left\|z_{k-1/2}-z_{k-1}\right\|_{D(x_{k-1})}^2,
		\end{align*}
	following the fact that $\nabla^2 r(z_{k-1}+\beta(z_{k-1/2}-z_{k-1})) \succeq \md(x_{k-1}+\beta(x_{k-1/2}-x_{k-1})) \succeq \frac{1}{2}\md(x_{k-1})$ from~\Cref{lem:reg-convex-and-bound} and $x_{k-1/2}\in[\tfrac{1}{2}x_{k-1}, 2x_{k-1}]$. Similarly, we also have 
	\[V_{z_{k-1/2}}(\bz_k)\ge \frac{1}{4}\left\|z_{k-1/2}-\bz_k\right\|_{\md(x_{k-1})}^2.\]
	
	On the other hand, we directly compute
	\begin{equation}\label{eq:lip-RHS}
	\begin{aligned}
		& \inner{g(z_{k-1/2})-g(z_{k-1})}{z_{k-1/2}-\bz_k} \\
		& \hspace{2em} = \innerB{\ma(y_{k-1/2}-y_{k-1})}{x_{k-1/2}-\bx_k}+ \innerB{-\ma^\top (x_{k-1/2}-x_{k-1})}{y_{k-1/2}-\by_{k}}\\
		& \hspace{4em}+\innerB{\mu \log\left(\frac{x_{k-1/2}}{x_{k-1}}\right)-\frac \eps 2|\ma|(y_{k-1/2}^2-y_{k-1}^2)}{x_{k-1/2}-\bx_k} \\
		& \hspace{4em}+\innerB{\epsilon \cdot\diag{y_{k-1/2}}|\ma|^\top x_{k-1/2} - \epsilon \cdot \diag{y_{k-1}}|\ma|^\top x_{k-1}}{y_{k-1/2}-\by_{k}}.
	\end{aligned}
	\end{equation}
	We now bound the terms on the right-hand side of~\eqref{eq:lip-RHS}. By Lemma~\ref{lem:reg-convex-and-bound},
	\begin{equation}
	\label{eq:lip-RHS-one}
	\begin{aligned}
		& \innerB{\ma(y_{k-1/2}-y_{k-1})}{x_{k-1/2}-\bx_k}\\
		& \hspace{5em}\le \frac{1}{2\rho}\|y_{k-1/2}-y_{k-1}\|_{\diag{|\ma|^\top x_{k-1}}}^2+\frac{\rho}{2}\|x_{k-1/2}-\bx_k\|_{\diag{\frac{1}{x_{k-1}}}}^2,\\
		& \innerB{\ma^\top(x_{k-1/2}-x_{k-1})}{\by_k - y_{k-1/2}}\\
		& \hspace{5em}\le \frac{1}{2\rho}\|y_{k-1/2}-\by_k\|_{\diag{|\ma|^\top x_{k-1}}}^2+\frac{\rho}{2}\|x_{k-1/2}-x_{k-1}\|_{\diag{\frac{1}{x_{k-1}}}}^2.
		\end{aligned}
	\end{equation}
	For the rest of the terms we have, letting $(x_\beta,y_\beta) = z_\beta = (1-\beta)z_{k-1}+\beta z_{k-1/2}$ for all $\beta \in [0,1]$,
	\begin{equation}\label{eq:lip-RHS-2}\begin{aligned}
	& \innerB{\mu \log\left(\frac{x_{k-1/2}}{x_{k-1}}\right)-\frac \eps 2|\ma|(y_{k-1/2}^2-y_{k-1}^2)}{x_{k-1/2}-\bx_k}\\
	&  + \innerB{\epsilon \cdot\diag{y_{k-1/2}}|\ma|^\top x_{k-1/2} - \epsilon \cdot \diag{y_{k-1}}|\ma|^\top x_{k-1}}{y_{k-1/2}-\by_k}\\
= & \int_0^1 (z_{k-1/2} - z_{k-1})^\top \begin{pmatrix} \mu \cdot \diag{\frac{1}{x_\beta}} & -\eps\Abs{\ma} \diag{y_\beta} \\ \eps\diag{y_\beta} \Abs{\ma}^\top & \eps \diag{\Abs{\ma}^\top x_\beta} \end{pmatrix} (z_{k-1/2} - \bz_k) d\beta \\
= & \sqrt{\frac {\mu\eps} 2}\int_0^1 (z_{k-1/2} - z_{k-1})^\top \begin{pmatrix} \rho \cdot \diag{\frac{1}{x_\beta}} & -\frac{2}{\rho}\Abs{\ma} \diag{y_\beta} \\ \frac 2\rho \diag{y_\beta} \Abs{\ma}^\top & \frac{2}{\rho} \diag{\Abs{\ma}^\top x_\beta} \end{pmatrix} (z_{k-1/2} - \bz_k) dt\\
\le & 8\sqrt{\frac {\mu\eps} 2}\left(\|z_{k-1/2}-z_{k-1}\|^2_{\md(x_{k-1})}+\|z_{k-1/2}-\bz_k\|^2_{\md(x_{k-1})}\right),
\end{aligned}\end{equation}
where for the last inequality we use similar arguments as in~\Cref{lem:reg-convex-and-bound} to bound the matrix by $4\md(x_\beta) \preceq 8\md(x_{k - 1})$ following the assumption that $x_{k-1/2}\in[\tfrac{1}{2}x_{k-1}, 2x_{k-1}]$ elementwise. Putting~\eqref{eq:lip-RHS-one} and~\eqref{eq:lip-RHS-2} into~\eqref{eq:lip-RHS}, we conclude that the relative Lipschitz condition~\eqref{def-rel-lip} holds with the stated $\alpha$.
\end{proof}

\coroaltmin*

\begin{proof}
We first claim \Cref{lem:progress-altmin} holds with $\kappa = 8$. This follows from the fact that for any $x'\in\Delta^m$, $y'\in[0,1]^n$ satyisfying $x'\ge\tfrac{1}{2}x$, we have
	\[
	\nabla^2 r(x',y')\succeq \md(x')\succeq \frac{1}{2}\md(x) \succeq \frac{1}{8}\nabla^2 r(x,y)\succeq \frac{1}{8}\nabla^2_{yy} r(x,y),
	\] 
	where we use both of the Hessian bounds from~\Cref{lem:reg-convex-and-bound} and the fact that restrictions to blocks only decrease a quadratic form.
	
For each $\gamma = (\gamma\x,\gamma\y)$ defined in~\Cref{alg:sherman} and the fact that $x_{k-1}, x_{k-1/2}\ge \frac \delta 4$ for all $k\in[K]$, we have by the assumptions $\norm{\ma}_\infty\le 1$, $\norm{c}_\infty\le \Cmax$, and $\norm{b}_1\le n\Bmax$ that it suffices to take
\[
B = \max\Par{\norm{\gamma\x}_\infty,\norm{\gamma\y}_1} = O\Par{(\alpha\rho+\mu)\log \frac 1 \delta + n\Bmax}
\]
in~\Cref{coro:altmin-conv}. Hence, the stated number of iterations $T$ suffices.
\end{proof}

\propouterpad*

\begin{proof}
Our proof is based on the proof of~\Cref{prop:outerloopproof-sm}. Lemma~\ref{lem:alphabetasm} and Lemma~\ref{lem:rel-lip} showed that our operator-regularizer pair satisfies $\nu$-strong monotonicity and $\alpha$-relative Lipschitzness. Under these conditions, and letting $\bz_k$ be the unpadded version of $z_k$, recall that the proof of Proposition~\ref{prop:outerloopproof-sm} showed that
\[(\nu + \alpha) V_{\bz_k}(z^\star) \le \alpha V_{z_{k - 1}}(z^\star) + \veps.\]
Moreover, the padding guarantee of~\Cref{lem:padsuffices} shows
\begin{equation*}
\begin{aligned}
(\nu+\alpha)\Par{V_{z_k}(z^\star)-V_{\bz_k}(z^\star)}\le (\nu+\alpha)\left(\rho+
\frac{8}{\rho}\right)m\delta\le \veps,
\end{aligned}
\end{equation*}
by the assumption on $\delta$. Following the proof of Proposition~\ref{prop:outerloopproof-sm}, we conclude
\[
V_{z_K}(z^\star)\le \Par{\frac{\alpha}{\nu+\alpha}}^K \Theta +\frac{2\veps}{\nu}.
\]
Since $\Theta = O(\rho \log m)$ for our choice of regularizer, the claim follows.
\end{proof}

Finally, we prove the main theorem in this section on the convergence of our regularized box-simplex solver for~\eqref{eq:main-reg-apprx}. We first state two helper lemmas used to prove the theorem.

\begin{lemma}[Entropy bounded by $\ell_1$-distance]\label{lem:entropyvsl1}
For any $x, \tx \in \Delta^m$, defining $H(x) = \sum_{i \in [m]} x_i \log x_i$,
\[H(\tx) - H(x) \le 33\log(m)\norm{\tx - x}_1 + m^{-30}.\]
\end{lemma}
\begin{proof}
We define a partition of the coordinates of $\tx$, 
\[L \defeq \Brace{i \in [m] \mid \tx_i \ge m^{-32}} \text{ and } S \defeq \Brace{i \in [m] \mid \tx_i < m^{-32}}.\]
First, we have by convexity of entropy that, letting $v_L \in \R^{|L|}$ denote the $|L|$-dimensional restriction of a vector $v \in \R^m$ to the coordinates of $L \subseteq [m]$,
\begin{align*}
\sum_{i \in L} \tx_i \log \tx_i - \sum_{i \in L} x_i \log x_i &\le \inprod{\log \tx_L + \1_L}{\tx_L - x_L} \\
&\le \norm{\log \tx_L + \1_L}_{\infty}\norm{\tx_L - x_L}_1\\
&\le \Par{32\log(m) + 1}\norm{\tx - x}_1 \le 33\log(m)\norm{\tx_L - x_L}_1,
\end{align*}
for sufficiently large $m$. In the second line, we used the $\ell_1$-$\ell_\infty$ H\"older's inequality. Moreover, let $S' \defeq \{i \in [m] \mid x_i < m^{-32}\}$ be the subset of $[m]$ where $x$ is small. By nonpositivity of entropy,
\begin{align*}
& \sum_{i \in S} \tx_i \log \tx_i - \sum_{i \in S} x_i \log x_i\le -\sum_{i \in S} x_i \log x_i = -\sum_{i \in S'} x_i \log x_i - \sum_{i \in S \setminus S'} x_i \log x_i \\
&\hspace{5em}\le \frac{32|S'|\log(m)}{m^{32}} + 32\log(m) \sum_{i \in S \setminus S'} x_i \\
&\hspace{5em}\le \frac{32|S|\log(m)}{m^{32}} + 32\log(m) \sum_{i \in S \setminus S'} \Par{x_i - m^{-32}} \\
&\hspace{5em}\le m^{-30} + 32\log(m) \sum_{i \in S \setminus S'} \Par{x_i - \tx_i} \le  m^{-30} + 32\log(m) \norm{\tx_S - x_S}_1.
\end{align*}
The second line used that $-c \log c$ is increasing in the range $(0, m^{-32})$, and the fourth used that $m^{-32} > \tx_i$ for all $i \in S$. We conclude by combining the above two displays.
\end{proof}

\begin{lemma}\label{lem:regdivtogap}
Let $(x^\star, y^\star)$ be the optimizer to the regularized box-simplex game in~\eqref{eq:main-reg-apprx}. Suppose $\rho\ge 6$, $\mu\le 1$, $\veps\ge m^{-10}$ and that $z = (x, y)$ satisfies $V_{z}\Par{z^\star }\le \frac{\veps^2}{C_{\max}^2\log^4 m}$. Then, $x$ is an $\veps$-approximate minimizer to \[f_{\mu,\eps}\x(x,y)\defeq\max_{y\in[0,1]^n}f_{\mu,\eps}(x,y).\]
\end{lemma}
\begin{proof}
We first claim that the given divergence condition implies that 
\begin{equation}\label{eq:regdivtol1}
\norm{x - x^\star}_1 \le \frac{\veps}{C_{\max}\log^2 m}.
\end{equation}
To see this, we observe that due to the Hessian bound in~\Cref{lem:reg-convex-and-bound}, the divergence induced by $r(x, y) - \frac{\rho}{2} H(x)$ is nonnegative, and hence
\begin{align*}
 \norm{x^\star - x}_1^2 \le  2V^H_{x}(x^\star) \le \frac \rho 2 V^H_{x}(x^\star) \le V^r_{z}(z^\star).
\end{align*}
The first inequality above was by Pinsker, the second used $\rho\ge 6$, and the last followed by nonnegativity of divergence. This implies \eqref{eq:regdivtol1} by the assumed divergence bound. We next show that \eqref{eq:regdivtol1} implies that $x$ is an approximate minimizer of $f_{\mu,\eps}\x$, which we recall means
\begin{equation}\label{eq:dualitygapparts}
\max_{y' \in \yset} f_{\mu,\eps}(x, y') - f_{\mu,\eps}(x^\star, y^\star) \le \veps.
\end{equation}
To this end, we compute
\begin{gather*}
\max_{y \in \yset} y^\top\Par{\ma^\top x - b} - \max_{y \in \yset} y^\top \Par{\ma^\top x^\star - b} \le \norm{\ma^\top\Par{x - x^\star}}_1 \le \norm{x - x^\star}_1 \le \frac{\veps}{4}, \\
\inprod{c}{x} - \inprod{c}{x^\star} \le C_{\max}\norm{x - x^\star}_1 \le \frac{\veps}{4}, \\
\mu \Par{H(x) - H(x^\star)} \le 33\mu\log(m)\norm{x - x^\star}_1 + \mu m^{-30} \le \frac{\veps}{4}, \\
\max_{y \in \yset} \Par{y^2}^\top \Abs{\ma}^\top x - \max_{y \in \yset} \Par{y^2}^\top \Abs{\ma}^\top x^\star \le \norm{\ma^\top\Par{x - x^\star}}_1 \le \norm{x - x^\star}_1 \le \frac{\veps}{4}.
\end{gather*}
The third line used Lemma~\ref{lem:entropyvsl1}. Combining these pieces proves the claim.
\end{proof}

We are now ready to prove~\Cref{thm:sherman}.

\thmsherman*

\begin{proof}
We handle correctness and runtime separately.

\textbf{Correctness.} Let $(x^\star,y^\star)$ be the minimax optimal solution of $f_{\mu,\eps}$ in~\eqref{eq:main-reg-apprx}. By applying~\Cref{prop:outerloopproof-sm-pad} and~\Cref{lem:regdivtogap}, shifting the definition of $\veps$ appropriately so that $\frac{3\veps}{\nu} \gets \frac{\veps^2}{C_{\max}^2\log^4 m}$, we have with the choice of parameters our returned $x_K$ satisfies the desired bounds.

\textbf{Runtime.} It is immediate to see each line of~\Cref{alg:sherman} runs in time $O(\nnz(\ma))$. Then, the total runtime follows from combining~\Cref{prop:outerloopproof-sm-pad}  and~\Cref{coro:altmin-conv-sherman}.
\end{proof}

We now apply Theorem~\ref{thm:sherman} to prove Corollary~\ref{coro:sherman}.

\corosherman*

\begin{proof}
We apply Theorem~\ref{thm:sherman} with $ \veps\gets \frac \eps 2$ and the same value of $\mu$ as in the regularized objective \eqref{eq:main-reg}. Note that the additive range of the quadratic regularizer is $\frac \eps 2$, and hence
\[
f_{\mu,\eps}(x,y)\le f_{\mu}(x,y)\le f_{\mu,\eps}(x,y) +\frac{\eps}{2},\text{ for all } (x,y)\in\Delta^m\times[0,1]^n.
\]
It thus suffices to obtain a $\frac \eps 2$ error guarantee from Theorem~\ref{thm:sherman}, which yields the runtime guarantee.
\end{proof}

\notarxiv{
	\subsection{Statement of Algorithm~\ref{alg:sherman}}\label{ssec:shermanalg}
	
	\SetKwProg{Fn}{function}{}{}
	
	\begin{algorithm}[ht!]
		\DontPrintSemicolon
		\KwInput{ $\ma \in \R^{m \times n}$, $c \in \R^m$, $b \in \R^n$, accuracy $\veps\in(m^{-10},1)$, $72\eps\le \mu\le  1$}
		\KwOutput{Approximate solution pair $(x, y)$ to \eqref{eq:main-reg-apprx}}
		\textbf{Global:} $\delta \gets  \tfrac{\eps\veps^2}{m^2}$, $\rho\gets \sqrt{\tfrac{2\mu}{\eps}}$, $\nu\gets \tfrac{1}{2}\sqrt{\tfrac{\mu\eps}{2}}$, $\alpha \gets 18\Cmax+32\sqrt{\frac{\mu\eps}{2}} \log \frac 4 \delta$ \;
		\textbf{Global:} $T \gets O\Par{\log \frac{mn\Bmax \alpha\rho}{\delta\veps}}$, $K \gets O\Par{\frac \alpha \nu \log\Par{\frac{\nu\log m}{\veps}}}$ for appropriate constants\;
		$(x_0,y_0)\gets (\tfrac{1}{m}\cdot\mathbf{1}_m, \mathbf{0}_n)$\;
		\For{$k =1$ to $K$}{
			$(\gamma\x, \gamma\y) \gets \GradBS(x_{k-1},y_{k-1}, x_{k-1},y_{k-1}, 0)$\;\label{line:altmin-one:start}
			$(x_{k-\frac{1}{2}},y_{k-\frac{1}{2}})\gets \AltminBS(\gamma\x,\gamma\y,\alpha,x_{k-1},y_{k-1})$\;\label{line:altmin-one:end}
			$(\gamma\x, \gamma\y) \gets \GradBS(x_{k-\frac{1}{2}},y_{k-\frac{1}{2}}, x_{k-1}, y_{k-1}, \nu)$\;\label{line:altmin-two:start}
			$(x^{(T+1)},y^{(T)})\gets \AltminBS(\gamma\x,\gamma\y,\alpha+\nu,x_{k-\frac{1}{2}},y_{k-\frac{1}{2}})$\;\label{line:altmin-two:end}
			$x_{k}\gets \frac{1}{\norm{\max\Par{x^{(T+1)},\delta}}_1}\cdot  \max\Par{x^{(T+1)},\delta}$, $y_{k}\gets y^{(T)}$\label{line:return-two}\Comment*{Implement padding $\mathcal{O}_\delta(x^{(T+1)})$}
		}
		\Fn{$\GradBS(x,y,x_0,y_0, \Theta)$}{
			$g\x \gets  \ma y + c + \mu(\1+\log(x)) - \frac{\eps}{2} \Abs{\ma} (y^2)$\;
			$g\y\gets -\ma^\top x + b + \eps\diag{y}|\ma|^\top x$ \;
			$g\x_r \gets -\alpha \rho (1+\log x_0) - \frac{\alpha}{\rho}|\ma| y^2_0 -  \Theta \rho (1+\log x) - \frac{\Theta}{\rho}|\ma| y^2$ \;
			$g\y_r \gets - \frac{2\alpha}{\rho} \diag {y_0}|\ma|^\top x_0 - \frac{2\Theta}{\rho} \diag {y}|\ma|^\top x)$ \;
			\Return $(g\x + g\x_r ,g\y + g\y_r)$	
		}
		\Fn(\Comment*[h]{Implement approximate proximal oracle via $\AltMin$}){$\AltminBS(\gamma\x,\gamma\y, \theta,x^{(0)},y^{(0)})$}{\For{$0 \le t \le  T$}{
				$x^{(t+1)}\gets \frac{1}{\norm{\exp\Par{-\frac{1}{\theta\rho}\gamma\x-\frac{1}{\rho^2}|\ma|\Par{y^{(t)}}^2}}_1}\cdot \exp\Par{-\frac{1}{\theta\rho}\gamma\x-\frac{1}{\rho^2}|\ma|\Par{y^{(t)}}^2}$ \;
				$y^{(t+1)} \gets \mathrm{med}\left(0,1, -\frac{\rho}{2\theta}\cdot\frac{\gamma\y}{|\ma|^\top x^{(t+1)}}\right)$ \;
			}	
			\Return $(x^{(T+1)},y^{(T)})$	
		}
		\caption{$\RegBoxSimp(\ma, b, c, \epsilon,\mu,\veps)$}\label{alg:sherman}
	\end{algorithm}

} %

\section{Approximating Sinkhorn distances}\label{sec:apprx}

In this section, we consider the problem of approximating Sinkhorn distances~\cite{Cuturi13}, a common computational task in the theory and practice of optimal transport. In this problem, we are given a product domain on two discrete supports $L\times R$, along with demands $d\in\R^{L+R}_{\ge 0}$ (with marginals $\norm{d_L}_1 = \norm{d_R}_1=1$) and a cost matrix $c\in\R^{L\times R}_{\ge0}$. Letting $m = |L||R|$ and $n = |L|+|R|$, the optimal transport (OT) problem asks to find a feasible transport plan $x\in\Delta^{m}$ minimizing $\inprod{c}{x}$, concisely described as 
\begin{equation}\label{def:ot}
\min_{x\in\Delta^m \mid \mb^\top x=d}\inprod{c}{x},\end{equation}
where $\mb$ is the unsigned incidence matrix of the complete bipartite graph on $L\times R$. Prior work by~\cite{Cuturi13} proposed Sinkhorn distances as an efficiently computable, differentiable alternative to~\eqref{def:ot}, where for $\mu>0$, the Sinkhorn distance asks for an optimal \emph{regularized} transport plan:
\begin{equation}\label{def:ot-reg}
x_\mu^\star ~~\text{solves}~~\min_{x\in\Delta^m \mid \mb^\top x=d}\inprod{c}{x}+\mu H(x),~\text{where}~H(x) = \sum_{i\in[m]}x_i\log x_i.
\end{equation}

We use the notion of function error in terms of the regularized objective $\inprod{c}{x}+\mu H(x)$ as the approximation criterion for developing our approximate solvers in this section, and we aim to find an $\eps$-approximate solution $x_\mu\in\mathcal{U}(d_L,d_R) \defeq \{x \in \Delta^m \mid \mb^\top x = d\}$ satisfying
\begin{equation}\label{def:ot-reg-eps}
\inprod{c}{x_\mu}+\mu H(x_\mu)\le  \inprod{c}{x_\mu^\star}+\mu H(x_\mu^\star) +  \eps.
\end{equation}
We call such an $x_\mu$ an $\eps$-approximate minimizer to \eqref{def:ot-reg}. One reason for choosing this definition of an approximate solution is that function error bounds the error in transportation cost and $\ell_1$ distances, i.e.\ 
\begin{align*}
\frac \mu 2 \norm{x_\mu-x^\star_\mu}^2_1 \le \inprod{c}{x_\mu} +\mu H(x_\mu) - \Par{ \inprod{c}{x_\mu^\star} + \mu H(x_\mu^\star)},\\
~~\text{and}~~\inprod{c}{x_\mu}-	\inprod{c}{x^\star_\mu}\le \norm{c}_\infty\norm{x_\mu-x^\star_\mu}_1.
\end{align*}
The first inequality follows from the first-order optimality condition of $x^\star_\mu$ and Pinsker's inequality, which shows the objective minimized by $x^\star_\mu$ is $\mu$-strongly convex in $\ell_1$. The second is $\ell_1$-$\ell_\infty$ H\"older.

We present three different approaches to approximate Sinkhorn distances~\eqref{def:ot-reg}. 
\begin{enumerate}
	\item In Section~\ref{ssec:first-OT}, we develop a first-order method that uses our regularized box-simplex solver, which finds an $\eps$-approximate minimizer to \eqref{def:ot-reg} in $\widetilde{O}(\frac 1 {\sqrt{\mu\eps}})$ time for $\mu = \Omega(\eps)$ (Theorem~\ref{thm:sinkhornmin}).
	\item In Section~\ref{ssec:high-OT}, we present a high-accuracy solver for the problem by building upon recent near-linear time box-constrained Newton's method for matrix scaling problems~\cite{Cohen17}, which finds an $\eps$-approximate minimizer in an improved $\widetilde{O}(\frac 1 \mu)$ time (Theorem~\ref{lem:ms-sinkhorn}).
	\item Lastly, for completeness we provide in~\Cref{ssec:sinkhornunaccel} an unaccelerated Sinkhorn distance solver with runtime $\widetilde{O}(\frac 1{\mu\eps})$ using the classical $\Sinkhorn$ method~\cite{DvurechenskyGK18}, which recovers the result in~\cite{altschuler2019massively}, Theorem 5 up to logarithmic factors.
\end{enumerate}

\subsection{Regularized box-simplex solver for computing Sinkhorn distances: Theorem~\ref{thm:sinkhornmin}}\label{ssec:first-OT}

In this section, we summarize our application of Theorem~\ref{thm:sherman} in approximating the Sinkhorn distance \eqref{def:ot-reg}. We first show approximating \eqref{def:ot-reg} can be reduced to solving a \emph{regularized box-simplex game}:
\begin{equation}\label{eq:sinkhornpddefintro}
\begin{gathered}
\min_{x\in\Delta^{m}} \max_{y \in [0, 1]^{2n}} \frac{1}{4C}\inprod{c}{x} +\frac{\mu}{4C} H(x) +  y^\top \ma^\top x - \inprod{b}{y}, \\
\text{ where } \ma \defeq \frac{1}{4}\begin{pmatrix}
	\mb, -\mb
\end{pmatrix}
,\; b \defeq \frac{1}{4}
\begin{pmatrix}
d, -d	
\end{pmatrix}
,
\end{gathered}
\end{equation}
for some $C \defeq 2\Par{\norm{c}_\infty + 33\mu \log(m)}$. This characterization, which is key to our approach, follows by first replacing the equality constraints $\mb^\top x = d$ with inequality constraints $\mb^\top x\le d, \mb^\top x\ge d$, and then developing an analogous reduction to the one used by \cite{JambulapatiST19} for solving the standard OT objective \eqref{def:ot}. In particular, by maximizing over $y$, problem \eqref{eq:sinkhornpddefintro} is equivalent to the following regularized $\ell_1$-regression problem:
\begin{equation}\label{eq:sinkhornl1def}
\min_{x \in \Delta^m} \frac{1}{4C}\inprod{c}{x} + \frac{\mu}{4C} H(x) + \frac{1}{4}\norm{\mb^\top x - d}_1.
\end{equation}
The quantity $\norm{\mb^\top x - d}_1$ is naturally interpreted as the deviation of the $x$ marginals from the demands $d$. We further require one simple helper fact, which shows that any transport plan $x$ can be ``fixed'' to be feasible for demands $d$ at a cost depending on its deviation $\norm{\mb^\top x - d}_1$.

\begin{lemma}[Lemma 7, \cite{AltschulerWR17}]\label{lem:fixot}
	There is an algorithm, $\OTRound$ (Algorithm 2, \cite{AltschulerWR17}), which takes as input any $x \in \R_{\ge 0}^m$ and produces $\tx \in \Delta^m$ such that $\norm{\tx - x}_1 \le 2\norm{\mb^\top x - d}_1$ and $\mb^\top \tx = d$. The algorithm runs in $O(m)$ time.
\end{lemma}

The guarantees of $\OTRound$ (cf.\ Lemma~\ref{lem:fixot}) give an algorithmic proof of equivalence between \eqref{eq:sinkhornl1def} and \eqref{def:ot-reg} by relating their approximate solutions: Lemma~\ref{lem:relatesinkhorn} gives a formal statement.

\begin{restatable}{lemma}{lemrelatesinkhorn}\label{lem:relatesinkhorn}
	Let $x$ be a $\tfrac{1}{4C}\Delta$-approximate minimizer to \eqref{eq:sinkhornl1def}. Then, letting $\tx$ be the result of $\OTRound(x)$ (cf.\ Lemma~\ref{lem:fixot}), $\tx$ is a $\Delta + \mu m^{-30}$-approximate minimizer to \eqref{def:ot-reg}. 
\end{restatable}

\begin{proof}
Define $\opt_{\textup{Sinkhorn}}$ to be the value of the Sinkhorn objective \eqref{def:ot-reg}, and $\opt_{\textup{Sinkhorn-Pen}}$ to be the value of \eqref{eq:sinkhornl1def}. It is immediate that $\opt_{\textup{Sinkhorn-Pen}} \le \tfrac{1}{4C}\opt_{\textup{Sinkhorn}}$, since any feasible point for \eqref{def:ot-reg} (in particular the minimizer) is also feasible for \eqref{eq:sinkhornl1def} and obtains the same function value. Hence, we have
\[
\inprod{c}{x} + \mu H(x) + C\norm{\mb^\top x - d}_1 \le 4C\cdot\opt_{\textup{Sinkhorn-Pen}} + \Delta \le \opt_{\textup{Sinkhorn}} + \Delta.
\]
Moreover, we bound the objective value of $\tx$:
\begin{align*}
\inprod{c}{\tx} + \mu H(\tx) &= \inprod{c}{x} + \mu H(x) + \inprod{c}{\tx - x} + \mu \Par{H(\tx) - H(x)} \\
&\le \inprod{c}{x} + \mu H(x) + \Par{\norm{c}_\infty + 33\mu \log(m)} \norm{\tx - x}_1 + \mu m^{-30} \\
&\le \inprod{c}{x} + \mu H(x) + 2\Par{\norm{c}_\infty + 33\mu \log(m)} \norm{\mb^\top x - d}_1 + \mu m^{-30} \\
&= \inprod{c}{x} + \mu H(x) + C\norm{\mb^\top x - d}_1 + \mu m^{-30}.
\end{align*}
In the second line, we applied $\ell_1$-$\ell_\infty$ H\"older and Lemma~\ref{lem:entropyvsl1}, in the third we used Lemma~\ref{lem:fixot}, and in the last we used the definition of $C$. Combining the above displays, we have the desired
\[\inprod{c}{\tx} + \mu H(\tx) \le \opt_{\textup{Sinkhorn}} + \Delta + \mu m^{-30}.\]
\end{proof}

Lemma~\ref{lem:relatesinkhorn} shows that to obtain an $\eps \ge 2\mu m^{-30}$-approximate minimizer of \eqref{def:ot-reg}, it suffices to solve \eqref{eq:sinkhornl1def} to $\tfrac{1}{8C}\cdot\eps$ additive accuracy. Applying Corollary~\ref{coro:sherman} in Section~\ref{ssec:framework-alg}, we thus develop a solver for minimizing the Sinkhorn distance objective \eqref{def:ot-reg} with the following guarantees.

\begin{restatable}[Regularized box-simplex solver for Sinkhorn distances]{theorem}{thmsinkhornmin}\label{thm:sinkhornmin}
Consider the Sinkhorn distance objective \eqref{def:ot-reg} on an instance with $n = |L|+|R|$ vertices and $m = |L||R|$ edges, and suppose $\eps\in[m^{-5}\norm{c}_\infty, \norm{c}_\infty]$ , $\mu \in[36\eps, \norm{c}_\infty]$. Using $\OTRound$ on the result of Corollary~\ref{coro:sherman} applied to \eqref{eq:sinkhornpddefintro} obtains an $\eps$-approximate minimizer to \eqref{def:ot-reg} in time
\[O\Par{m\Par{\sqrt{\frac{\mu}{\eps}}\log(m) + \frac{\norm{c}_\infty}{\sqrt{\mu\eps}}} \log(m)\log \Par{\frac{\norm{c}_\infty\log m}{\eps}}}.\]
\end{restatable}

\begin{proof}
Note that the deviation between the regularized Sinkhorn objective \eqref{eq:sinkhornpddefintro} and the actual Sinkhorn objective \eqref{def:ot-reg} at an approximate minimizer is at most $\mu m^{-30} \le \frac \eps {2} =: \Delta$ by Lemma~\ref{lem:relatesinkhorn}. Thus, it suffices to find a $\veps$-approximate minimizer to \eqref{eq:sinkhornpddefintro} with the choice of parameters
\[ \mu' \gets \frac{\mu}{4C} = \frac{\mu}{8\Par{\norm{c}_\infty + 33\mu \log(m)}},~~\text{and}~~\veps \gets \frac{\eps}{8C} = \frac{\eps}{16\Par{\norm{c}_\infty + 33\mu \log(m)}}.\]
It is straightforward to verify that the assumptions at the start of Section~\ref{sec:framework} on $\norm{\ma}_\infty$, $\norm{b}_1$, and $\norm{c}_\infty$ are all met for this problem and the desired accuracy parameter. Thus, by applying Corollary~\ref{coro:sherman}, we obtain the desired runtime.
\end{proof}

\subsection{Box-constrained Newton for computing Sinkhorn distances: Theorem~\ref{lem:ms-sinkhorn}}\label{ssec:high-OT}

 In this section, we show how to apply the recent matrix scaling solvers developed in~\cite{Cohen17} to obtain second-order high-accuracy solvers for approximating Sinkhorn distances. Our approach follows from a direct connection to the scaling problem via the well-known correspondence of matrix scaling and computing Sinkhorn distances: we provide the following results to make this reduction from matrix scaling rigorous for completeness, and to prove our claimed runtime in Theorem~\ref{lem:ms-sinkhorn}.

We first define the matrix scaling problem.

\begin{definition}[Matrix scaling]\label{def:matscale}Given a nonnegative matrix $\mk\in\R^{n_1\times n_2}_{\ge0}$, the $(r,c)$-matrix scaling problem asks to find the vectors $u\in\R^{n_1}$, $v\in\R^{n_2}$ such that
\[
\diag{\exp(u)}\cdot \mk \cdot\diag{\exp(v)}\1 = r,\; \1^\top \diag{\exp(u)}\cdot\mk \cdot\diag{\exp(v)} = c^\top.
\]
\end{definition}

In other words, it asks to rescale $\mk$ by nonnegative diagonal matrices to obtain the specified row and column marginals. Prior work by~\cite{Cohen17} gives a near-linear time matrix scaling solver by using a box-constrained Newton's method, which we restate in~\Cref{prop:ms-solver} for our use.

\begin{proposition}[Theorem 4.6,~\cite{Cohen17}]\label{prop:ms-solver}
	Given $\mk\in\R^{n_1\times n_2}_{\ge0}$ with $s_{\mk} \defeq \sum_{i\in[n_1], j\in[n_2]}\mk_{ij}$, define  \begin{equation}\label{def:potential-ms}P_\mk(u,v) \defeq \sum_{i\in[n_1], j\in[n_2]}\mk_{ij}\exp(u_i-v_j)-\langle r, u\rangle+\langle c,v\rangle.\end{equation}
	Suppose there is an approximate $(r,c)$-scaling $(u,v)$ satisfying $P_\mk(u,v)\le \inf_{u',v'} P_\mk(u',v')+\frac{\veps^2}{n_1+n_2}$ and $\norm{u}_\infty, \norm{v}_\infty \le B$. Then we can compute an approximate $(r,c)$ scaling $(u,v)$ of $\mk$ satisfying 
	\[
	\norm{\diag{\exp(u)}\cdot \mk \cdot\diag{\exp(v)}\1 - r}_2^2+\norm{\diag{\exp(v)}\cdot \mk^
	\top \cdot\diag{\exp(u)}\1 - c}_2^2\le \veps,
	\]
	in time $\widetilde{O}(\nnz\Par{\mk}\cdot B\log^2(\frac{s_\mk} \veps))$.
\end{proposition}

It is well-known that computing Sinkhorn distances is equivalent to finding the optimal matrix scaling for $\mk \defeq \exp(-\tfrac{1}{\mu}\mc)\in\R^{L\times R}$, where $\mc = (c_{ij})_{i\in L,j\in R}$ comes from reshaping the cost vector (see e.g.\ Lemma 2, \cite{Cuturi13}). Consequently, we can apply~\Cref{prop:ms-solver} with $\veps \defeq \tfrac{\eps^2}{8\norm{c}_\infty^2m}$ and $\OTRound$ (Lemma~\ref{lem:fixot}) to obtain a high-precision solver for the Sinkhorn distance.

\begin{restatable}[Box-constrained Newton solver for Sinkhorn distances]{theorem}{thmmssinkhorn}\label{lem:ms-sinkhorn}
	Consider the Sinkhorn distance objective \eqref{def:ot-reg} on an instance with $n = |L|+|R|$ vertices and $m = |L||R|$ edges, and suppose $\eps\in[m^{-5}\norm{c}_\infty, \norm{c}_\infty]$ and $\mu\in(0,\norm{c}_\infty]$. The algorithm in Proposition~\ref{prop:ms-solver} obtains an $\eps$-approximate minimizer to~\eqref{def:ot-reg} in time 
	\[\widetilde{O}\Par{\frac{m\norm{c}_\infty}{\mu}}.\]
\end{restatable}

To prove the theorem, we  first require one helper lemma.

\begin{lemma}\label{lem:approxsinkhorn}
Let $d_L \in \Delta^L, d_R \in \Delta^R$ be demands, let $\mc \in \R^{L \times R}_{\ge 0}$ be costs, and define $\mk \defeq \exp(-\frac 1 \mu \mc)$ for some $\mu \ge 0$. Let $x \in \R^m_{\ge 0}$ be $\mx \in \R^{L \times R}$ appropriately vectorized, such that $\mx = \diag{u} \cdot \mk \cdot \diag{v}$ for some $u \in \R^L_{\ge 0}$, $v \in \R^R_{\ge 0}$. Assume that
\[\norm{\mx \1 - d_L}_1 + \norm{\mx^\top \1 - d_R}_1 \le \Delta.\]
Then, $\hat{x} \defeq \frac{x}{\norm{x}_1}$ is a $3C\Delta + \mu m^{-30}$-approximate minimizer to \eqref{def:ot-reg} with demands $d$, where $C \defeq 2(\norm{c}_\infty + 33\mu\log(m))$.
\end{lemma}
\begin{proof}
Clearly, $\hat{x}$ is also a diagonal scaling of $\mk$, by multiplying $u$ or $v$ by appropriate scalars, and
\begin{align*}\norm{x - \hat{x}}_1 = \Abs{\norm{x}_1 - 1} = \Abs{\1^\top \mx \1 - \1^\top d_L} \le \norm{\mx \1 - d_L}_1 \le \Delta.
\end{align*}
This along with Lemma 2 of \cite{Cuturi13} implies that $\hat{x}$, as a diagonal scaling of $\mk$, is the optimal solution to \eqref{def:ot-reg} for some demands $\hat{d}_L = \widehat{\mx}\1$, $\hat{d}_R = \widehat{\mx}^\top \1$ such that 
\[\norm{\hat{d} - d}_1 = \norm{\mx \1 - d_L}_1 + \norm{\mx^\top \1 - d_R}_1 \le \Delta + \norm{(\mx - \widehat{\mx})\1}_1 + \norm{(\mx - \widehat{\mx})^\top\1}_1 \le 3\Delta.\]
The remainder of the proof follows from Lemma~\ref{lem:optdbound}.
\end{proof}

Next, Lemma~\ref{lem:approxsinkhorn} is key to proving~\Cref{lem:ms-sinkhorn} below.

\begin{proof}[Proof of~\Cref{lem:ms-sinkhorn}]
	Let $C \defeq 2\Par{\norm{c}_\infty + 33\mu \log(m)}$. We first do a preprocessing step to pad the demands $d_L, d_R$ to the modified demands $\tilde{d}_L = \mathcal{O}_\delta(d_L)$ and $\tilde{d}_R = \mathcal{O}_\delta(d_R)$ (see~\Cref{def:pad}), where $\delta \defeq \tfrac{\eps}{48mC}$.
	Next, suppose for $\veps \defeq \tfrac{\eps^2}{288C^2m}$, we find a $\veps$-approximate $(\tilde{d}_L,\tilde{d}_R)$-scaling $(u,v)$ in the sense of Proposition~\ref{prop:ms-solver}, such that $\mx \defeq \diag{u} \cdot \mk \cdot \diag{v}$ satisfies 
	\[
	\norm{\mx\1-\tilde{d}_L}_1+\norm{\mx^\top 1-\tilde{d}_R}_1\le \sqrt{2m}\sqrt{\norm{\mx\1-\tilde{d}_L}_2^2+\norm{\mx^\top 1-\tilde{d}_R}_2^2}\le \sqrt{2m\veps} = \frac{\eps}{12C}.
	\]
	Recalling that the guarantees of padding imply
	\[
	\norm{d_L-	\tilde{d}_L}_1 + \norm{d_R-	\tilde{d}_R}_1 \le 2m\delta+2m\delta = \frac{\eps}{12C},
	\]
	we have by the triangle inequality, letting $x \in \R_{\ge 0}^m$ be the appropriately vectorized $\mx$, and $\hat{x} \defeq \frac{x}{\norm{x}_1}$,
	\begin{align*}
	\norm{\mb^\top \hat{x} -d}_1&\le \norm{\mx\1-\tilde{d}_L}_1+ \norm{d_L-	\tilde{d}_L}_1 + \norm{\mx^\top 1-\tilde{d}_R}_1 + \norm{d_R-	\tilde{d}_R}_1 + 2\norm{\hat{x} - x}_1 \\
	& \le  \frac{\eps}{3C} .
	\end{align*}
In the last inequality, we bound (see the proof of Lemma~\ref{lem:approxsinkhorn}) $\norm{x - \hat{x}}_1 \le \frac{\eps}{12C}$. Next, let $x'$ be the plan which results after applying $\OTRound$ for $\hat{x}$ and demands $d$, and let $\tilde{x}$ be the resulting plan when applying $\OTRound$ for $x_\mu^\star$ and demands $\tilde{d}$, where $x_\mu^\star$ is optimal for \eqref{def:ot-reg} with demands $d$. By a similar argument as used in the proof of~\Cref{lem:relatesinkhorn}, we have 
\begin{equation}\label{eq:round-sinkhorn-one}
\begin{aligned}
\inprod{c}{\tx} + \mu H(\tx) &\le \inprod{c}{x_\mu^\star} + \mu H(x_\mu^\star) + C\norm{\mb^\top x_\mu^\star - \tilde{d}}_1  + \mu m^{-30},\\
\inprod{c}{x'} + \mu H(x') &\le \inprod{c}{\hat{x}} + \mu H(\hat{x}) + C\norm{\mb^\top \hat{x} - d}_1  + \mu m^{-30}.
\end{aligned}
\end{equation}
Moreover, letting $\tilde{x}^\star_\mu$ be optimal for \eqref{def:ot-reg} with demands $\tilde{d}$, applying Lemma~\ref{lem:approxsinkhorn} with $d \gets \tilde{d}$ and $\Delta \gets \frac{\eps}{12C}$, we obtain
\begin{align}\label{eq:round-sinkhorn-two}
	\inprod{c}{\hat{x}} + \mu H(\hat{x})\le \inprod{c}{\tilde{x}_\mu^\star} + \mu H(\tilde{x}_\mu^\star) + \frac{\eps}{4} + \mu m^{-30} \le \inprod{c}{\tilde{x}} + \mu H(\tilde{x}) +\frac{\eps}{4} + \mu m^{-30},
\end{align}
since $\tx$ is feasible for \eqref{def:ot-reg} with demands $\tilde{d}$. Finally,
\begin{equation}\label{eq:dtd}
\norm{\mb^\top x^\star_\mu - \tilde{d}}_1 = \norm{d - \tilde{d}}_1 \le \frac{\eps}{12C}.
\end{equation}
Combining~\eqref{eq:round-sinkhorn-one},~\eqref{eq:round-sinkhorn-two}, and~\eqref{eq:dtd}, as well as our assumed ranges on $\eps$ and $\mu$, we obtain 
\begin{align*}
\inprod{c}{x'} + \mu H(x') & \le  \inprod{c}{x_\mu^\star} + \mu H(x_\mu^\star) +C\norm{\mb^\top x_\mu^\star - \tilde{d}}_1 + C\norm{\mb^\top \hat{x} - d}_1  + \frac{\eps}{4} +  3\mu m^{-30}\\
& \le \inprod{c}{x_\mu^\star} + \mu H(x_\mu^\star) + \frac\eps {12}  + \frac \eps 3 + \frac{\eps}{4} + 3\mu m^{-30} \le \inprod{c}{x_\mu^\star} + \mu H(x_\mu^\star) +  \eps.
\end{align*}
This proves the correctness of returning the transport plan $x'$.

The runtime is bounded by one call to the matrix scaling solver, and one call to $\OTRound$ (which clearly does not dominate). To parameterize~\Cref{prop:ms-solver}, the  matrix scaling problem with $\mk = \exp(-\tfrac{1}{\mu}\mc)\in\R^{L\times R}$ satisfies $s_\mk\le m$ since $c\ge0$ entrywise. Furthermore, there are optimal $(\tilde{d}_L,\tilde{d}_R)$-scaling vectors $(u^\star,v^\star)$ of $\mk$ satisfying (see e.g.~Lemma 10, \cite{BlanchetJKS18}):
	\[\norm{u^\star}_\infty, \norm{v^\star}_\infty = O\Par{\frac{\norm{c}_\infty}{\mu}+\log \frac{m}{\min_{i\in L, j\in R}\Par{[\tilde{d}_L]_i, [\tilde{d}_R]_j}}} = O\Par{\frac{\norm{c}_\infty}{\mu}+\log \frac{m\norm{c}_\infty}{\eps}}.\]
	Thus by applying~\Cref{prop:ms-solver} with $\veps = \tfrac{\eps^2}{288C^2m}$, the runtime is bounded as desired.
\end{proof}

\subsection{Unaccelerated rate of Sinkhorn's algorithm}\label{ssec:sinkhornunaccel}

In this section, we demonstrate that Sinkhorn iteration obtains a complexity for approximately optimizing \eqref{def:ot-reg} scaling as $\tO(\frac{\norm{c}_\infty^2}{\mu\eps})$, the natural unaccelerated counterpart to Theorem~\ref{thm:sinkhornmin}. We note that this rate for solving \eqref{def:ot-reg} (up to logarithmic factors) appeared as Theorem 5 of~\cite{altschuler2019massively}, and we provide a proof for completeness. We consider an algorithm and combine an analysis of Sinkhorn due to \cite{DvurechenskyGK18} and a technical lemma developed in this paper (Lemma~\ref{lem:entropyvsl1} from~\Cref{app:framework-alg}), which we use to reduce \eqref{def:ot-reg} to a regularized box-simplex game up to negligible (inverse-polynomial) additive error.

Throughout, we assume $\eps \le \mu \le \norm{c}_\infty$, and further that $\frac{\norm{c}_\infty}{\eps} \le m^5$; otherwise, the runtime of cutting-plane methods \cite{JiangLSW20} subsume all runtimes claimed in this paper for approximately solving \eqref{def:ot-reg}. Finally, throughout this section we will fix the values of $\eps$, $\mu$, and $c \in \R^m_{\ge 0}$, and define for demands $d$ the optimal value of \eqref{def:ot-reg} by
\[\opt_\mu(d) \defeq \min_{x \in \Delta^m, \mb^\top x = d} \inprod{c}{x} + \mu H(x).\]

We first bound the difference between $\opt_\mu(d)$ and $\opt_\mu(d')$ for ``nearby'' $d$ and $d'$.

\begin{lemma}\label{lem:optdbound}
For any demand vectors $d, d' \in \R^{L + R}_{\ge 0}$ with blockwise restrictions in $\Delta^L$ and $\Delta^R$, 
\[\opt_\mu(d) \le \opt_{\mu}(d') + \Par{2\norm{c}_\infty + 66\mu\log(m)} \norm{d - d'}_1 + \mu m^{-30}.\]
\end{lemma}
\begin{proof}
Throughout the proof define $\Delta \defeq \norm{d - d'}_1$. Let $x$ satisfying $\mb^\top x = d$ achieve the value $\opt_\mu(d)$, and similarly let $x'$ with $\mb^\top x = d'$ achieve $\opt_\mu(d')$. By Lemma~\ref{lem:fixot}, there is a plan $\tx$ with $\mb^\top \tx = d'$ and $\norm{x - \tx}_1 \le 2\Delta$. Since $x'$ minimizes the Sinkhorn objective with demands $d'$,
\begin{align*}
\inprod{c}{x'} + \mu H(x') &= \inprod{c}{\tx} + \mu H(\tx) \\
&= \inprod{c}{x} + \mu H(x) + \Par{\inprod{c}{\tx - x} + \mu H(\tx) - \mu H(x)} \\
&\le \inprod{c}{x} + \mu H(x) + \Par{2\norm{c}_\infty + 66\mu\log(m)} \Delta + \mu m^{-30}.
\end{align*}
The only inequality used $\ell_1$-$\ell_\infty$ H\"older's, and Lemma~\ref{lem:entropyvsl1}.
\end{proof}
For the remainder of the section, we will fix a particular $d$ which we wish to optimize the Sinkhorn objective for (and refer to all other demands by $d'$, $\td$, etc.). We begin by defining the ``padded'' demand vector $\td \in \R^{L + R}_{\ge 0}$ obtained from modifying $d$ as follows:
\begin{equation}\label{eq:tddef}
\td_L = \frac{\max(d_L, m^{-20})}{\norm{\max(d_L, m^{-20})}_1},\; \td_R = \frac{\max(d_R, m^{-20})}{\norm{\max(d_R, m^{-20})}_1}.
\end{equation}
By observation, $\|\td - d\|_1 \le 2m^{-19}$. Next, we recall a result from the literature on the performance of Sinkhorn's algorithm.

\begin{proposition}[Theorem 1, \cite{DvurechenskyGK18}]\label{prop:sinkhornrate}
There is an algorithm, $\Sinkhorn$, which given demand vector $\td$ entrywise at least $m^{-21}$, returns $x'$ such that
\[x' = \argmin_{x \in \Delta^m, \mb^\top x = d'} \inprod{c}{x} + \mu H(x)\]
for some $d'$ with $\norm{d' - d}_1 \le \Delta$, in time
\[O\Par{m\Par{\frac{\norm{c}_\infty}{\mu\Delta} + \frac{\log m}{\Delta}}}.\]
\end{proposition}

Finally, we give our main result on the performance of Sinkhorn's algorithm for optimizing \eqref{def:ot-reg}.

\begin{corollary}\label{cor:unaccel}
For \eqref{def:ot-reg} on an instance with $n$ vertices and $m$ edges, suppose $\eps \le \mu \le \norm{c}_\infty$ and $\frac{\norm{c}_\infty}{\eps} \le m^5$. Using $\Sinkhorn$ (Proposition~\ref{prop:sinkhornrate}) on $\td$ defined in \eqref{eq:tddef} to accuracy 
\[\Delta \defeq \frac{\eps}{10\norm{c}_\infty + 330\mu\log(m)}\]
and then applying $\OTRound$ (Lemma~\ref{lem:fixot}) with demands $d$ results in a vector $x$ with $\mb^\top x = d$ and
\[\inprod{c}{x} + \mu H(x) \le \opt_{\mu}(d) + \eps\]
in time
\[O\Par{\frac{\norm{c}_\infty^2}{\mu\eps} + \frac{\mu \log^2 m}{\eps}}.\]
\end{corollary}
\begin{proof}
By Proposition~\ref{prop:sinkhornrate}, we obtain a vector $x' \in \Delta^m$ satisfying $\mb^\top x' = d'$, $\inprod{c}{x} + \mu H(x)$, and $\norm{d' - d}_1 \le \norm{d' - \td}_1 + \norm{\td - d}_1 \le \Delta + 2m^{-19}$. Moreover, letting $x \gets \OTRound(x')$ for demands $d$, we have by Lemma~\ref{lem:fixot} that $\norm{x - x'}_1 \le 2\Delta + 4m^{-19}$, implying
\begin{align*}
\inprod{c}{x} + \mu H(x) &\le \inprod{c}{x'} + \mu H(x') + \Par{\norm{c}_\infty + 33\mu\log(m)}\Par{2\Delta + 4m^{-19}} + \mu m^{-30} \\
&\le \inprod{c}{x'} + \mu H(x') + \frac{\eps}{2} = \opt_\mu(d') + \frac \eps 2.
\end{align*}
Moreover, Lemma~\ref{lem:optdbound} implies that
\[\opt_\mu(d') \le \opt_\mu(d) + \Par{2\norm{c}_\infty + 66\mu\log(m)}\Par{2\Delta + 4m^{-19}} + \mu m^{-30} \le \opt_\mu(d) +  \frac \eps 2. \]
The conclusion follows from combining the above two displays.
\end{proof}

\end{document}